\PassOptionsToPackage{x11names}{xcolor}
\documentclass[fontsize=12pt,leading=14pt]{fullpaper}

\usepackage{subcaption}

\usetikzlibrary{decorations.pathmorphing, patterns, decorations.pathreplacing, decorations.shapes, decorations.markings, arrows.meta}

\tikzstyle{walkage}=[-latex, thick, decorate, decoration = {coil, aspect = 1.2}]
\tikzstyle{linkage}=[-latex, thick, decorate, decoration = {snake, segment length = 7mm}]

\newcommand{\cref}{\zcref}
\newcommand{\labelcref}{\zcref[noname]}
\newcommand{\citelink}[2]{{#2}}

\newcommand{\N}{\mathbb{N}\xspace}

\newcommand{\Acal}{\mathcal{A}\xspace}
\newcommand{\Bcal}{\mathcal{B}\xspace}
\newcommand{\Ccal}{\mathcal{C}\xspace}

\newcommand{\Lcal}{\mathcal{L}\xspace}
\newcommand{\Ocal}{\mathcal{O}\xspace}

\newcommand{\Qcal}{\mathcal{Q}\xspace}
\newcommand{\Scal}{\mathcal{S}\xspace}

\newcommand{\Tcal}{\mathcal{T}\xspace}
\newcommand{\Wcal}{\mathcal{W}\xspace}
\newcommand{\Zcal}{\mathcal{Z}\xspace}

\newcommand{\wLP}{\fref[w]{prop:linked-path}}
\newcommand{\lISM}{\fref[\ell]{prop:independent-set-multipartite}}
\newcommand{\lISMold}{\fref[\ell]{prop:independent-set-multipartite-non-algo}}
\newcommand{\cCM}{\fref[c]{prop:clique-minor}}
\newcommand{\dCM}{\fref[d]{prop:clique-minor}}
\newcommand{\cCMold}{\fref[c]{prop:clique-minor-non-algo}}
\newcommand{\dCMold}{\fref[d]{prop:clique-minor-non-algo}}
\newcommand{\dP}{\fref[d]{prop:partition}}
\newcommand{\wCW}{\fref[w]{lem:closed-walkage}}
\newcommand{\wUW}{\fref[w]{lem:untangled-walkage}}
\newcommand{\wSW}{\fref[w]{cor:structured-walkage}}
\newcommand{\dCE}{\fref[d]{cor:cyclic-extension}}

\newcommand{\wCWS}{\fref[w]{lem:cyclic-walk-system}}
\newcommand{\lCWS}{\fref[\ell]{lem:cyclic-walk-system}}
\newcommand{\lFCWS}{\fref[\ell]{lem:fully-cyclic-walk-system}}
\newcommand{\lB}{\fref[\ell]{lem:final-bramble}}
\newcommand{\dB}{\fref[d]{thm:main-proof}}

\newcommand{\dBCCM}{\fref[d]{thm:main-proof-congestion-4-masarik}}
\newcommand{\dBCM}{\fref[d]{thm:main-proof-congestion-6-masarik}}
\newcommand{\wSWC}{\fref[w]{lem:structured-walkage-congestion}}
\newcommand{\dCEC}{\fref[d]{lem:untangled-cyclic-extension-congestion}}

\newcommand{\wSB}{\fref[w]{lem:sparse-bramble}}
\newcommand{\lFBCC}{\fref[\ell]{lem:final-bramble-clean-4}}
\newcommand{\lFBC}{\fref[\ell]{lem:final-bramble-clean-6}}
\newcommand{\dHIP}{\fref[d]{thm:half-integral-packing}}

\DeclareMathOperator{\tw}{\mathsf{tw}}
\DeclareMathOperator{\dtw}{\mathsf{dtw}}
\DeclareMathOperator{\bn}{\mathsf{bn}}

\DeclareMathOperator{\degen}{\mathsf{d}}
\DeclareMathOperator{\intersec}{\mathrm{I}}
\DeclareMathOperator{\stLink}{\mathsf{start}}
\DeclareMathOperator{\enLink}{\mathsf{end}}
\DeclareMathOperator{\walks}{Walks}
\DeclareMathOperator{\pairs}{Pairs}
\DeclareMathOperator{\inOp}{in}
\DeclareMathOperator{\outOp}{out}

\usetikzlibrary{mindmap,positioning}
\usetikzlibrary{graphs}
\usetikzlibrary[graphs]
\usetikzlibrary{patterns}
\usetikzlibrary{arrows,automata,positioning}
\usetikzlibrary{decorations.pathreplacing}
\usetikzlibrary{decorations}
\usetikzlibrary{decorations.pathmorphing}
\usetikzlibrary{shapes.geometric}
\usetikzlibrary{quotes}
\usetikzlibrary{shapes.callouts}
\usetikzlibrary{arrows.meta}
\usetikzlibrary{calc}
\usetikzlibrary{fit}
\usetikzlibrary{backgrounds}

\tikzset{faded/.style={gray,very thin}}
\tikzset{vertex/.style={draw,circle,minimum size=5pt,inner sep=0pt}}
\tikzset{novertex/.style={circle,minimum size=5pt,inner sep=0pt}}
\tikzset{blackvertex/.style={draw,circle,minimum size=5pt,inner sep=0pt, fill=black}}
\tikzset{redvertex/.style={draw,circle,minimum size=5pt,inner sep=0pt, fill=red}}
\tikzset{redvertexfaded/.style={draw,circle,faded,minimum size=5pt,inner sep=0pt, fill=red!50}}
\tikzset{greenvertex/.style={draw,circle,minimum size=5pt,inner sep=0pt, fill=green}}
\tikzset{greenvertexfaded/.style={draw,circle,faded,minimum size=5pt,inner sep=0pt, fill=green!50}}
\tikzset{bluevertex/.style={draw,circle,minimum size=5pt,inner sep=0pt, fill=blue}}
\tikzset{bluevertexfaded/.style={draw,circle,faded,minimum size=5pt,inner sep=0pt, fill=blue!50}}
\tikzset{yellowvertex/.style={draw,circle,minimum size=5pt,inner sep=0pt, fill=yellow}}
\tikzset{yellowvertexfaded/.style={draw,circle,faded,minimum size=5pt,inner sep=0pt, fill=yellow!50}}
\tikzset{arrow/.style={-{Latex[scale=1]},shorten >= 1pt}}
\tikzstyle{process} = [rectangle, minimum width=3cm, minimum height=1cm, text centered, draw=black, text width=2.8cm]
\tikzstyle{accept} = [rectangle, minimum width=3cm, minimum height=1cm, text centered, draw=black, fill=blue!20, text width=2.8cm]

\tikzset{edge/.style = {->,> = latex'}}

\tikzset{snake it/.style={decorate, decoration=snake}}
\tikzset{big snake/.style={decorate, decoration={snake,segment length=7mm, amplitude=3mm}}}

\title{Polynomial directed treewidth forces a large bramble with congestion two}
\author{
  \authorname{Victor A. Campos}
  \affiliation{ParGO group, Universidade Federal do Ceará, Fortaleza, Brazil.}
  \authorname{Raul Lopes}
  \affiliation{Hamburg University of Technology, Institute for Algorithms and Complexity, Hamburg, Germany.}
  \authorname{Rodrigo Nogueira}
  \affiliation{ParGO group, Universidade Federal do Ceará, Fortaleza, Brazil.}
  \authorname{Ignasi Sau}
  \affiliation{LIRMM, Université de Montpellier, CNRS, Montpellier, France.}
}
\shorttitle{Polynomial directed treewidth forces a large bramble with congestion two}
\shortauthors{V.\,A. Campos, R. Lopes, R. Nogueira, and I. Sau}

\date{}

\begin{document}

\maketitle

\begin{abstract}
  Brambles of constant congestion in directed graphs have interesting structural and algorithmic applications, and it is desirable to find the smallest function $f_c(k)$ such that if a digraph has directed treewidth at least $f_c(k)$, then it contains a bramble of size at least $k$ and congestion at most $c$.
  By the Directed Grid Theorem of \citelink{KawarabayashiK15}{Kawarabayashi and Kreutzer~[STOC 2015]}, such a  function $f_c(k)$ exists for every value $c \geq 2$ of the congestion.
  However, even using the recent improvements due to~\citelink{HatzelKMM24}{Hatzel, Kreutzer, Milani, and Muzi~[FOCS 2024]}, the resulting function is a power tower of height~$22$.
  Recently, \citelink{MasarikPilipczukRzazewskiSorge22}{Masa{\v{r}}ík, Pilipczuk, Rz{\k{a}}\.{z}ewski, and Sorge~[SIAM J.\ Discrete Math.\ 2022]} proved that, for congestion~$8$, the function $f_8(k)$ can be made polynomial; more precisely, they showed that $f_8(k) = \Ocal(k^{48} \cdot \log^{13} k)$.
  They asked whether it is possible to decrease the congestion, ideally to~$2$, even at the cost of a higher polynomial function.
  In this work we answer their question affirmatively with the best possible value of the congestion, by proving that $f_2(k) = \Ocal(k^{96}\cdot \log^{23} k)$.
  We also obtain smaller polynomials for higher congestion, namely $f_4(k) = \Ocal(k^{72} \cdot \log^{21} k)$ and $f_6(k) = \Ocal(k^{66} \cdot \log^{18} k)$.
  Our results are inspired by the proof for congestion~$8$, which we improve by introducing~\emph{walkages}.
  Walkages allow us to avoid technicalities of the original proof by translating ideas due to \citelink{RW12}{{Reed and Wood~[Eur.\ J.\ Comb.\ 2012]}} that work for undirected graphs to the directed setting.
  Our proofs are constructive and naturally yield polynomial-time algorithms to find the brambles assuming a certificate for large directed treewidth is given.
  Moreover, we discuss algorithmic applications of our results for the \textsc{Directed Disjoint Paths with Congestion} problem, and to half-integral packings of subgraphs with large directed treewidth.
\end{abstract}

\metadatum{Keywords}{digraphs, directed treewidth, brambles, congestion, directed disjoint paths.}

\metadatum{AI disclosure}{All the work that lead to this paper is entirely due to the authors, including posing the problems, finding the solutions, writing the text, and verifying its correctness. We have not used AI tools of any kind.}

\clearpage

\section{Introduction}\label{sec:intro}

\paragraph{Brambles in undirected graphs.} Undirected treewidth, denoted by $\tw$, is undoubtedly one of the most successful graph parameters, both from a structural~\cite{RobertsonS04} and an algorithmic~\cite{Courcelle90} point of view.
One of the key results in the celebrated Graph Minors project of Robertson and Seymour is the Grid Theorem~\cite{RobertsonST94}, stating that a grid minor is the canonical obstruction to having small treewidth.
More precisely, they proved  that there exists a function $f_{\sf GT}(k)$ such that if a graph $G$ satisfies $\tw(G)\geq f_{\sf GT}(k)$, then $G$ contains a grid of order $k$ as a minor.
In a breakthrough result, \textcite{ChekuriC16} reduced the original upper bound on $f_{\sf GT}(k)$ from nonelementary to polynomial.
The currently smallest function has been found by \textcite{ChuzhoyT21}, who proved that $f_{\sf GT}(k) = \Ocal(k^9\cdot\polylog k)$ suffices.

A notion tightly related to undirected treewidth is that of a \emph{bramble}, which can be seen as a relaxation of a grid, in the sense that it is a less structured object that still guarantees ``high global connectivity''.
A bramble $\Bcal$ in an undirected graph $G$ is a collection of connected subgraphs such that for any $B_1,B_2 \in \Bcal$, $B_1$ and $B_2$ intersect or $G$ contains an edge from $B_1$ to $B_2$.
The \emph{size} of a bramble is its number of elements, and its ``complexity'' measure is its \emph{order}, defined as the minimum size of a vertex set of $G$ that intersects all of  its elements.
The importance of brambles stems from the fact that, as proved by \textcite{SeymourT93}, they are an object dual to treewidth: the maximum order of a bramble in a graph is exactly its treewidth plus one.

The \emph{congestion} of a bramble is the maximum number of subgraphs in $\Bcal$ which contain some vertex of $G$.
Note that the size of a bramble is upper-bounded by its order times its congestion.
It is worth mentioning that, if the congestion is unbounded, there may be an exponential gap between the size and the order of a bramble.
Indeed, in some classes of graphs, such as constant-degree expanders, it is known that brambles of high order need to have exponential size~\cite{HatzelKPS22,GroheM09}.

Since treewidth also appears naturally when solving a plethora of graphs problems, and many of them can be solved efficiently once the treewidth is bounded~\cite{Courcelle90}, it is desirable to have a ``reasonable'' upper bound, ideally polynomial, on the value of the treewidth that guarantees the existence of a large bramble of constant congestion.
This can be achieved via the polynomial version of the Grid Theorem discussed above~\cite{ChuzhoyT21}: if the treewidth of a graph is $\Omega(k^9\cdot\polylog k)$, then it contains a grid of order $k$ as a minor, which contains a bramble of size $k$ and congestion~$2$.
Alternative proofs without using the Grid Theorem have also been found.
The current smallest bound shows that if $\tw(G) = \Omega(k^2 \cdot \polylog k)$, then $G$ contains a bramble of size $k$ and congestion~$2$~\cite{HatzelKPS22}.
\textcite{RW12}, on the other hand, show that if $\tw(G) = \Omega(k^4 \sqrt{\log k})$, then $G$ contains a grid-like minor of order $k$, which contains a slightly more organized bramble of size $k$ and congestion~$2$, given by a set of paths whose intersection graph is bipartite and contains a $K_k$-minor.

\paragraph{Brambles in directed graphs.} In directed graphs, which are the focus of this paper, most of what we discussed so far finds its appropriate counterpart, but several issues turn out to be considerably more involved.
\textcite{Johnson2001} defined directed treewidth, denoted by $\dtw$, and conjectured an analog of the Grid Theorem, where the obstruction would be a \emph{cylindrical grid} as a \emph{butterfly minor}.
This conjecture was proved by \textcite{KawarabayashiK15}, with a non-elementary function $f_{\sf DGT}(k)$ on the directed treewidth to guarantee a cylindrical grid of order $k$.
The bound on $f_{\sf DGT}(k)$ was recently improved in~\cite{HatzelKMM24} to a power tower of height 22.
Obtaining a polynomial bound for the function $f_{\sf DGT}$ remains a challenging open problem.

A \emph{bramble} $\Bcal$ in a digraph $D$ is defined as one may expect: a collection of strongly connected subgraphs such that, for any $B_1,B_2 \in \Bcal$, $B_1$ and $B_2$ intersect or $D$ contains an arc from $V(B_1)$ to $V(B_2)$ and an arc from $V(B_2)$ to $V(B_1)$.
The \emph{size}, \emph{order}, and \emph{congestion} of a bramble are defined as in the undirected case.
While brambles are not ``exact'' duals to directed treewidth, they are approximate ones, as proved by \textcite{Reed99}.
For a better bound, combining~\cite[Statement~(1.1)]{Johnson2001} and~\cite[Lemma~6.4.20]{KreutzerOrdyniak14} we get that if $D$ is a digraph with $\dtw(D) \geq 2k$, then $D$ contains a bramble of order $k$.

In directed graphs, where the toolbox of algorithmic techniques is not as rich as in undirected graphs, brambles of constant congestion are peculiarly useful.
In particular, they can be used for routing problems and the widely studied \textsc{Directed Disjoint Paths with Congestion} problem~\cite{BentertCHKKS26,CamposSLS25,KawarabayashiKK14,GiannopoulouKKK22,ChekuriEP18,Edwards2017}.
Thus, it is natural to ask, for each value $c \geq 2$, what is the smallest function $f_c(k)$ such that any digraph $D$ with $\dtw(D) \geq f_c(k)$ contains a bramble of size $k$ and congestion at most $c$ (recall that for brambles of constant congestion, size and order are related by a constant factor).
We remark that $f_1(k)$ does not exist as a cylindrical grid has arbitrarily large directed treewidth, but it does not contain a large bramble whose elements are pairwise disjoint.
On the other hand, $f_c(k)$ exists and is upper bounded by $f_{\sf DGT}(k)$ for every $c \geq 2$, as a cylindrical grid of order $k$ as a butterfly minor implies a bramble of size $k$ and congestion $2$ in the original digraph.
Since even the elementary upper bound on $f_{\sf DGT}(k)$ of~\cite{HatzelKMM24} is still huge, it is relevant to look for ``reasonable'' bounds for $f_c(k)$ for $c \geq 2$, ideally polynomial as in the undirected setting.

It is worth mentioning that a half-integral variant of the Directed Grid Theorem~\cite{KawarabayashiKK14} could be used to obtain a bramble of size $k$ and congestion $4$, but the resulting upper bound on $f_4(k)$ is also very fast-growing.
Restricted to planar digraphs, \textcite{HatzelKK19} showed a polynomial bound with degree $6$ for the Directed Grid Theorem, implying the same asymptotic bound for the function $f_c(k)$ for every $c \geq 2$, restricted to planar digraphs.

The main starting point of this article is~\cite{MasarikPilipczukRzazewskiSorge22}, which provides a polynomial bound for congestion~$8$ in general digraphs by proving that $f_8(k) = \Ocal(k^{48} \cdot \log^{13} k)$.
The authors asked whether it is possible to decrease the congestion, ideally to $2$, even at the cost of a higher polynomial function.

\paragraph{Our contribution.} We answer the question of the authors of~\cite{MasarikPilipczukRzazewskiSorge22} positively, by proving the following result.

\begin{theorem}\label{thm:main-proof-short}
  There exists a function $f\colon \N \to \N$ with $f(k)=\Ocal(k^{96} \cdot \log^{23} k)$ such that any digraph $D$ with $\dtw(D) \ge f(k)$ contains a bramble of size $k$ and congestion $2$.
\end{theorem}

We remark that the produced bramble is \emph{strong}, in the sense that all bags pairwise intersect.
This is a more restrictive notion than the one of bramble defined above and, contrary to the latter, it is closed under \emph{butterfly minors}.
\cref{thm:main-proof-short} can be summarized as $f_2(k) = \Ocal(k^{96}\cdot \log^{23} k)$.
This provides an immediate polynomial upper bound to $f_c(k)$ for all $c\ge 2$.
As previously observed, this is the best possible congestion, as $f_1(k)$ does not exist.
We obtain smaller polynomials for higher values of the congestion, namely $f_4(k) = \Ocal(k^{72} \cdot \log^{21} k)$ (see \cref{thm:main-proof-congestion-4-masarik}) and $f_6(k) = \Ocal(k^{66} \cdot \log^{18} k)$ (see \cref{thm:main-proof-congestion-6-masarik}).
To obtain these, we review the techniques in~\cite{MasarikPilipczukRzazewskiSorge22} and improve them by simplifying the case analysis and reducing the congestion in intermediate structures.
We remark, however, that these techniques do not improve the bound $f_8(k) = \Ocal(k^{48} \cdot \log^{13} k)$ of~\cite{MasarikPilipczukRzazewskiSorge22}.
We give a more detailed overview in \cref{sec:proof-overview}.

\textcite{LS26} showed how to turn the main proof of~\cite{MasarikPilipczukRzazewskiSorge22} into a polynomial-time algorithm that finds the desired bramble of size $k$ and congestion at most $8$, assuming that a large obstruction to directed treewidth (namely, a path system) is given.
We observe that the same holds for our results since, as is the case in~\cite{LS26}, the main bottleneck is on the application of the Lovász Local Lemma~\cite{LE75}.
For simplicity, we elect to apply the algorithmic version of the Lovász Local Lemma due to \textcite{Harris2022} that maintains the same bounds on the parameter as in the original application due to \textcite{RW12}.
This is done through \cref{prop:independent-set-multipartite}.

\paragraph{Applications.} We now present a list of applications of \cref{thm:main-proof-short}.
These topics are further discussed in \cref{sec:main-proof}.

In the \textsc{Directed Disjoint Paths with Congestion} $c$ problem, we are given a digraph $D$ and $k$ pairs $(s_1, t_1), \ldots, (s_k, t_k)$ of vertices of $D$, and the goal is to find, for each $i \in [k]$, a path $P_i$ from $s_i$ to $t_i$ in $D$ such that every vertex of $D$ belongs to at most $c$ of these paths (see \cref{sec:ddp-congestion}).
For integers $k,c$, we say that a digraph $D$ is \emph{$(k,c)$-linked} if every instance of \textsc{Directed Disjoint Paths with Congestion} $c$ on $D$ with $k$ pairs of vertices is positive.
A digraph is \emph{$k$-strong} if it has at least $k+1$ vertices and it remains strongly connected after the removal of any set of fewer than $k$ vertices (see \cref{sec:pre} for a formal definition).
\textcite{Edwards2017} showed that there is a (huge) function $L(k)$ such that every $L(k)$-strong digraph is $(k, 2)$-linked.
\textcite{LS26} proved that there is \textsl{polynomial} function $g(k)$ such that every $g(k)$-strong digraph is $(k, 8)$-linked, and that a corresponding solution can be found in polynomial time.
Combining \cref{thm:main-proof-short} with~\cite{Edwards2017}, it follows that  there is a \textsl{polynomial} function $g(k)$ such that every $g(k)$-strong digraph is $(k, 2)$-linked, and that a corresponding solution can be found in polynomial time, as path systems can be found in polynomial time in this case~\cite[Proposition~10]{LS26}, thus reducing the super exponential bound $L(k)$ in~\cite{Edwards2017}.
We remark that such a result cannot be reduced to congestion $1$ as \textcite{Thomassen1991} showed that \textsc{Directed Disjoint Paths} is \NP-complete even for $k=2$ paths in $p$-strong digraphs, for any integer $p \geq 1$.

Almost immediately from the above paragraph, we also get an \textsf{XP} algorithm with parameter $k$ to solve \textsc{Directed Disjoint Paths with Congestion} $2$ in $n^{\poly(k)}$ time for $k$-strong digraphs.
Prior to our work, only $n^{f(k)}$ for a huge function $f$ was known due to the dependence on cylindrical grids as routing structures.

Also as a direct corollary of \cref{thm:main-proof-short}, we immediately get a stronger version of~\cite[Theorem~5]{MasarikMuziPilipczukRzazewskiSorge21}, where the authors show that a digraph $D$ with $\dtw(D) = \Omega(a^6  b^6 \cdot \polylog (a  b))$ contains $a$ subgraphs of directed treewidth at least $b$ such that each vertex of $D$ appears in at most $4$ of these subgraphs.
With \cref{thm:main-proof-short} we can improve the congestion from $4$ to $2$, although we get a worse polynomial dependence than that of~\cite[Theorem~5]{MasarikMuziPilipczukRzazewskiSorge21}.

\paragraph{Further research.} A natural research direction is to close the gap of the function $f_c(k)$ for every $c \geq 2$.
Since directed treewidth generalizes treewidth~\cite[Statement 2.1]{Johnson2001}, the bounds by \textcite{GroheM09} apply.
In particular, the authors of~\cite{GroheM09} show that, for every $k$, there is an undirected graph $G$ of treewidth $\Omega(k)$ such that every bramble of $G$ of order $k^{1/2 +\varepsilon}$ has size exponential in the number of vertices of $G$.
Hence, for any constant $c \geq 2$, we cannot hope for a better bound than $f_c(k) = \Ocal(k^2 \cdot \polylog k)$.

We are also interested in identifying classes of digraphs for which $f_1(k)$ is defined, and when it can be made polynomial.
For example, consider \emph{semicomplete} digraphs, which are digraphs containing at least one arc between any pair of vertices.
In~\cite{GKRW21} it is shown that the directed treewidth and the directed pathwidth of semicomplete digraphs are functionally equivalent, and \textcite{FradkinS2013} showed that every semicomplete digraph of directed pathwidth doubly exponential in $k$ contains a \emph{$k$-triple} which, in turn, contains a bramble of size $k$ and congestion~$1$.
The double exponential is due to the need of finding large transitive tournaments in localized structures, and reliance on a Ramsey-type argument.
It is not known if this bound can be made polynomial, or even reduced to a single exponential, if we request a bramble of size $k$ and congestion~$1$ instead of a $k$-triple.

In undirected graphs, a slightly more organized bramble with congestion $2$, namely two families of pairwise vertex-disjoint paths with an intersection graph containing a large clique minor with size depending quadratically on the treewidth, has been shown to exist by~\textcite{RW12}.
We ask whether similar ``grid-like'' minors can be obtained in digraphs while maintaining a polynomial dependency on the directed treewidth.

Finally, the overarching question is whether a polynomial bound suffices to build cylindrical grids.
The original proof of the Directed Grid Theorem by~\textcite{KawarabayashiK15}, whose overall structure is inherited by the authors in~\cite{HKMM24}, starts from a path system and tries to clean it into a new system where the linkages (that is, collections of pairwise vertex-disjoint paths) intersecting the base path $P$ are vertex-disjoint and intersect $P$ only in prescribed segments.
If this cleaner system is not found, in~\cite{HKMM24} the authors find another intermediary structure in the construction, known as a \emph{well-linked web}, which consists of a grid-like structure formed by two interweaving linkages $\mathcal{P,Q}$ with the additional property that the union of the starting and ending vertices of $\mathcal{Q}$ forms a well-linked set.
We remark that, with our methods, it is possible to use webs to produce a shorter construction of brambles of size $k$ and congestion $2$ in digraphs of directed treewidth roughly $2^k$.
Although we were not able to extend this strategy into a construction with polynomial dependence on $k$, we hope that our ideas can be applied to shave a part of the current known bound for cylindrical grids, consisting of a tower of exponentials of height 22, by circumventing webs in the proof.

\paragraph{Organization.}
\cref{sec:proof-overview} gives a detailed overview on how our proof works and how it relates to those in~\cite{MasarikPilipczukRzazewskiSorge22,RW12}.
\cref{sec:pre} introduces basic concepts of graphs and digraphs, parameterized complexity, and state some preliminary results.
\cref{sec:walkages} introduces \emph{walkages} and related structures, which are fundamental tools in our construction, and shows how we manipulate them throughout this paper.
In \cref{sec:walk-systems}, we formally define \emph{strong systems}, those being highly connected structures that contain large brambles of congestion~$2$, and show how they can be built from particular types of walkages.
In \cref{sec:main-proof}, we state and prove the complete version of \cref{thm:main-proof-short} and discuss applications of our result.
\cref{sec:higher-congestion-with-masarik} adapts the proof in~\cite{MasarikPilipczukRzazewskiSorge22} using our methods in order to produce brambles with congestion $4$ and $6$ with a smaller dependency on the directed treewidth than that of \cref{thm:main-proof-short}.

\section{Proof overview}\label{sec:proof-overview}
We draw inspiration from the proofs in~\cite{MasarikPilipczukRzazewskiSorge22}.
Our main goal is to translate the main proof in~\cite{RW12}, which only works for undirected graphs, to digraphs to simplify the analysis done in~\cite{MasarikPilipczukRzazewskiSorge22}.
A common process used in both proofs is to analyze the degeneracy of the \emph{intersection graph} of a family of walks $\mathcal{W}$, which is defined as the undirected graph with one vertex for each walk in $\mathcal{W}$ and an edge between any pair of intersecting walks.
We first review how their proofs work.
Some of the intuitive sentences that appear in this section can be found again closer to the corresponding proofs in the next sections, as well as all formal definitions of objects informally defined here.

\paragraph{Overview of the main proof in~\cite{MasarikPilipczukRzazewskiSorge22}.}
Given a digraph $D$ and $A, B \subseteq V(D)$ with $|A| = |B|$, a \emph{linkage}  from $A$ to $B$ is a set of $|A|$ pairwise vertex-disjoint paths from $A$ to $B$ in $D$.
A set $X \subseteq V(D)$ is \emph{well-linked} if, for all $A, B \subseteq X$ with $|A| = |B|$, there is a linkage of size $|A|$ from $A$ to $B$ in $D$.

A \emph{path system} of order $k$, informally, is a set of base paths $P_1, \ldots, P_k$ together with a linkage $\mathcal{L}_{ij}$ of size $k$ from $V(P_i)$ to $V(P_j)$, for all $i,j \in [k]$.
Path systems are known obstructions to bounded directed treewidth, and they can be built from a path $P$ containing a well-linked set $A$ of size $\Theta(k^2)$.
The construction of $P$ and $A$ can be found in \FPT{} time parameterized by the size of $A$ assuming large directed treewidth~\cite{CLMS22} (cf.\ \cref{prop:linked-path}).

Starting from a path system, they build \emph{threaded linkages} by interleaving a linkage $\{R_1, \ldots, R_\ell \}$ with $[\ell-1]$ extra paths, the \emph{threads}, linking the end of $R_i$ to the start of $R_{i+1}$, for $i \in [\ell-1]$.
Informally, in~\cite[Lemma 3.5]{MasarikPilipczukRzazewskiSorge22}, a threaded linkage $\mathcal{Q}$ is constructed using pairs of linkages $\mathcal{L}_{ij}, \mathcal{L}_{ji}$ to zigzag between $V(P_i)$ and $V(P_j)$, forming cycles with segments of $P_i$ that are used to find the threads.
Because vertices appear at most once in each linkage and once in $P_i$, any vertex appears at most three times in the walk built from the concatenation of the paths in one of the linkages and the threads forming $\mathcal{Q}$.
Hence, adopting their terminology, we say that $\mathcal{Q}$ has \emph{overlap} at most 3.
In~\cite[Lemma 3.6]{MasarikPilipczukRzazewskiSorge22}, these are then filtered into sets of closed walks or \emph{untangled threaded linkages}, preserving the bound on the overlap.

In~\cite[Lemma~3.1]{MasarikPilipczukRzazewskiSorge22}, they show that a collection of closed walks with congestion at most $c$ for which the degeneracy of their intersection graph is not below a certain threshold contains a large bramble with congestion at most $c$.
This lemma is backed by \textcite[Theorem~1]{Kostochka84} which forces a large clique minor in reasonably dense graphs (cf.~\cref{prop:clique-minor}) and constitutes the ``dense'' case of their proof, because in all other scenarios they combine collections of closed walks with sparse intersection graphs.

Next, in a localization step, they build an auxiliary graph $G$ whose vertices are associated with the closed walks and (untangled) threaded linkages, encoding intersection graph degeneracies as edges.
That is, the existence of an edge $e$ in $G$ implies that the degeneracy of the intersection graph of the two structures associated with the endpoints of $e$ is above a certain bound.
In~\cite[Claim 4.1]{MasarikPilipczukRzazewskiSorge22}, they show that a $0.6$ fraction of $V(G)$ concentrates into either an independent set whose vertices are all related to sets of closed walks, or in one of two matching structures.
In the first case, the bramble is constructed by cleaning intersections between paths forming the closed walks by using a result by \textcite{RW12} (cf.~\cref{prop:independent-set-multipartite}) whose backend depends on the Lovász Local Lemma~\cite{LE75}.
In the latter two matching cases, threaded linkages are tied into closed walks via the crucial Bowtie Lemma~\cite[Lemma 3.7]{MasarikPilipczukRzazewskiSorge22}.
The Bowtie Lemma says that if the intersection graph of two threaded linkages is not too sparse, then one can build a family of closed walks whose congestion depends on the overlap of the given threaded linkages and on whether they are untangled or not.
We remark that, if both threaded linkages are untangled, the overlap is irrelevant and the Bowtie Lemma yields congestion~$2$.

\paragraph{Overview of the main proof in~\cite{RW12}.}
The authors aim to find a \emph{grid-like minor}, which is defined as a collection of paths whose intersection graph is bipartite and contains a large clique minor.
This proof is conceptually much simpler.
Starting from an (undirected) path system, the authors simply analyze how the paths linking the segments $P_i$ and $P_j$ of the base path intersect.
If the intersection graph $\intersec(\Lcal_{ij} \cup \Lcal_{i'j'})$ of the linkages $\Lcal_{ij}$ and $\Lcal_{i'j'}$ is reasonably dense for some choice of $i$, $j$, $i'$, and $j'$, then they construct a grid-like minor by finding a clique minor in $\intersec(\Lcal_{ij} \cup \Lcal_{i'j'})$.
Otherwise, all linkages are sufficiently disjoint for them to pick a distinct path linking $P_i$ to $P_j$ for each $i$ and $j$, and all these distinct paths are also vertex-disjoint.
This is essentially a subdivision of a clique, which is then used to build the grid-like minor.

\paragraph{Improving tools from~\cite{MasarikPilipczukRzazewskiSorge22}.}
We now describe improvements we made to some of the tools from~\cite{MasarikPilipczukRzazewskiSorge22}.
The formal statements, appearing in later parts of this paper, are different, as they are written not in terms of linkages but \emph{walkages}, which we introduce later.

The reliance on the Bowtie Lemma and the initial overlap of $3$ are the main factors driving the congestion bound up to $8$ in~\cite{MasarikPilipczukRzazewskiSorge22}.
Indeed, the authors of~\cite{MasarikPilipczukRzazewskiSorge22} apply~\cite[Lemma~3.5\,+\,Lemma~3.6]{MasarikPilipczukRzazewskiSorge22} to generate untangled threaded linkages or many closed walks with congestion at most $3$.
We replace these applications by a process that generates untangled threaded linkages or many closed walks with congestion $2$ in \cref{lem:structured-walkage-congestion}.
This is done by building the threaded linkages using only $\Lcal_{ij}$ and $\Lcal_{ji}$, sparing $P_i$ and $P_j$ as opposed to~\cite[Lemma~3.5]{MasarikPilipczukRzazewskiSorge22}.

\cref{lem:structured-walkage-congestion} further enables us to improve on~\cite[Lemma~3.1]{MasarikPilipczukRzazewskiSorge22}.
Given a collection of closed walks with congestion at most $c$, we show in \cref{lem:closed-walkage} how to either build a bramble with congestion at most $c$ or find many vertex-disjoint closed walks in it, regardless of how these closed walks intersect.
This is useful for situations in~\cite{MasarikPilipczukRzazewskiSorge22} where there is a collection of closed walks with congestion $2$ that do not intersect well enough, so they need to combine it with another similar collection to enhance the intersections, but this increases the congestion of this new collection up to $4$.
With our approach, if we have such a collection, and we do not find a bramble with congestion $2$ inside it, we can assume its closed walks are vertex-disjoint, so whenever we merge such collections the congestion never exceeds $2$.

Finally, we improve on the Bowtie Lemma~\cite[Lemma~3.7]{MasarikPilipczukRzazewskiSorge22}.
One of the main uses of this lemma is its special case where the two threaded linkages are untangled, which results in closed walks with congestion $2$.
This is the best case scenario congestion-wise, but it is not always possible to ensure the threaded linkages are untangled due to the nature of~\cite[Lemma~3.6]{MasarikPilipczukRzazewskiSorge22} in their proof and \cref{lem:structured-walkage-congestion} in ours.
However, using \cref{lem:closed-walkage} we ensure the collections of closed walks are always vertex-disjoint, which gives us much more control.
In~\cite{MasarikPilipczukRzazewskiSorge22}, collections of closed walks are often treated as threaded linkages in disguise, and
including closed walks obtained from the Bowtie Lemma in such collections can result in increased congestion.
We show in \cref{cor:cyclic-extension} that the Bowtie Lemma can be extended to use one untangled threaded linkage and one family of vertex-disjoint closed walks, instead of two untangled threaded linkages, to create a collection of closed walks with congestion $2$.

\paragraph{Walkages.}
If we replaced ``paths'' with ``connected subgraphs'' in the main proof of~\cite{RW12}, we would no longer obtain a grid-like minor, but the final object would still contain a large bramble with congestion $2$.
Hence, even though the ideas in~\cite{RW12} are aimed at undirected graphs, if we are able to mimic this undirected path system structure in a digraph, we can then follow their steps replacing ``paths'' by ``strong subgraphs'' and still get a bramble with congestion $2$.
The problem is that in a directed setting there is much more work to be done, as paths are connected graphs, whereas directed paths are not strong subgraphs.

To build our undirected path system we need the notion of a new object called \emph{walkage}, which generalizes a linkage by replacing paths with walks.
Although this is conceptually a very simple change, we show that it is a powerful tool if we manage to make these walks behave like paths.
We introduce this concept as a framework to simplify arguments relying on large linkages intersecting some well-linked set.
We illustrate the usefulness of the walkage framework by showing the high level sketch of the construction of a \emph{strong system}, our directed version of the undirected path system used by~\textcite{RW12}.

Starting from a path system as in~\cite{MasarikPilipczukRzazewskiSorge22}, we observe that not only there is a linkage $\Lcal_{ij}$ between segments $P_i$ and $P_j$ of the base paths, but there is also a linkage $\Lcal_i$ from vertices in a later portion of $P_i$ to vertices in an earlier portion of $P_i$.
We use these linkages to convert these paths into closed walks; furthermore, we choose $\Lcal_i$ leaving some vertices of the well-linked set in the middle, as in \cref{fig:back-linkage}.
This guarantees that the produced closed walk still contains a large enough subset of the well-linked set.
It remains to show how to build a large collection of such closed walks which are vertex-disjoint.
The technical details are non-trivial and are presented in \cref{lem:cyclic-walk-system}.

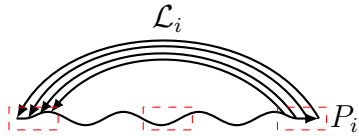
\begin{figure}[htbp]
  \centering
  \begin{tikzpicture}
    \draw[linkage, decoration = {pre length = 1mm, post length = 1mm, segment length = 10mm}] (0,0) -- (4,0);
    \node[right] at (4,0) {$P_i$};
    \node[above] at (2,1) {$\Lcal_i$};

    \draw[thick, -latex] (4,0) to[bend right, in = 240, out = 300] (0,0);
    \draw[thick, -latex] (3.85,0) to[bend right, in = 240, out = 300] (0.15,0);
    \draw[thick, -latex] (3.7,0) to[bend right, in = 235, out = 305] (0.3,0.05);
    \draw[thick, -latex] (3.55,0.05) to[bend right, in = 230, out = 310] (0.45,0.08);

    \draw[red, dashed] (-0.1,-0.15) rectangle (0.55,0.15);
    \draw[red, dashed] (3.45,-0.15) rectangle (4.1,0.15);
    \draw[red, dashed] (1.675,-0.15) rectangle (2.325,0.15);
  \end{tikzpicture}
  \caption{The path $P_i$ and the linkage $\Lcal_i$ from the final portion of $P_i$ to its initial portion. The dashed red boxes represent subsets of a well-linked set.
  }\label{fig:back-linkage}
\end{figure}

By building these disjoint closed walks $W_1, \ldots, W_k$, we can still find linkages $\Lcal_{ij}$ between $W_i$ and $W_j$.
The next step towards the undirected path system of~\cite{RW12} is to convert such linkages into collections of vertex-disjoint closed walks.
This is very close to what is done in the previous paragraph, but now the paths linking different closed walks are not guaranteed to contain a large subset of a well-linked set.
Indeed, it could be the case that only the very first and very last vertices of each path are in a particular well-linked set.
Thus, we cannot directly apply the same idea.
The solution comes by replacing linkages with walkages.
Indeed, if the goal is to find many closed walks, it makes no difference if we combine paths or walks.
Replacing paths by walks has two major features: first, a path is still a walk, so our initial object is still a path system, and we can mix paths and walks at will; second, walks may overlap, so we can actually inject vertices from a well-linked set into a path if we allow it to self-intersect.

With this idea in mind, we can solve the issue of paths in the linkages $\Lcal_{ij}$ not having enough vertices from a well-linked set as follows.
By well-linkedness, there is a linkage $\Lcal_{ij}$ from the closed walk $W_i$ to the closed walk $W_j$, but there is also a linkage $\Lcal_{ji}$ from $W_j$ to $W_i$.
The start vertices of paths in $\Lcal_{ij}$ are the end vertices of the paths in $\Lcal_{ji}$ and vice versa, so the paths of $\Lcal_{ij}$ and $\Lcal_{ji}$ can be concatenated to form closed walks as in \cref{fig:closed-long-walks}.

\begin{figure}[htbp]
  \centering
  \begin{tikzpicture}[scale = 0.5]
    \draw[thick, -latex, goodblue] (3.7,1.6) to[bend left] (0.3,1.6);
    \draw[thick, -latex, goodblue] (3.33,0.5333) to[bend left] (0.65,0.5333);
    \draw[thick, -latex, goodblue] (3.33,-0.5333) to[bend left] (0.65,-0.5333);
    \draw[thick, -latex, goodblue] (3.7,-1.6) to[bend left] (0.3,-1.6);

    \draw[thick, -latex, goodred] (0.3,1.6) to[bend left] (3.7,1.6);
    \draw[thick, -latex, goodred] (0.65,0.5333) to[bend left] (3.33,0.5333);
    \draw[thick, -latex, goodred] (0.65,-0.5333) to[bend left] (3.33,-0.5333);
    \draw[thick, -latex, goodred] (0.3,-1.6) to[bend left] (3.7,-1.6);

    \draw[thick] (0,0) ellipse (0.7cm and 1.75cm);
    \draw[thick] (4,0) ellipse (0.7cm and 1.75cm);
    \node[left] at (-0.7,0) {$W_i$};
    \node[right] at (4.7,0) {$W_j$};
    \node[goodred, above] at (2,2.2) {$\Lcal'_{ij}$};
    \node[goodblue, below] at (2,-2.2) {$\Lcal'_{ji}$};
    \draw[-latex] (-0.7, -0.01) -- (-0.7,-0.01);
    \draw[-latex] (4.7, -0.01) -- (4.7,-0.01);

    \begin{scope}[xshift = 10cm]
      \draw[thick, -latex, goodblue] (3.7,1.6) -- (0.65,0.5333);
      \draw[thick, -latex, goodblue] (3.33,0.5333) -- (0.65,-0.5333);
      \draw[thick, -latex, goodblue] (3.33,-0.5333) -- (0.3,-1.6);
      \draw[thick, -latex, goodblue] (3.7,-1.6) -- (0.3,1.6);

      \draw[thick, -latex, goodred] (0.3,1.6) -- (3.7,1.6);
      \draw[thick, -latex, goodred] (0.65,0.5333) -- (3.33,0.5333);
      \draw[thick, -latex, goodred] (0.65,-0.5333) -- (3.33,-0.5333);
      \draw[thick, -latex, goodred] (0.3,-1.6) -- (3.7,-1.6);

      \draw[thick] (0,0) ellipse (0.7cm and 1.75cm);
      \draw[thick] (4,0) ellipse (0.7cm and 1.75cm);
      \node[left] at (-0.7,0) {$W_i$};
      \node[right] at (4.7,0) {$W_j$};
      \node[goodred, above] at (2,1.6) {$\Lcal'_{ij}$};
      \node[goodblue, below] at (2,-1.6) {$\Lcal'_{ji}$};
      \draw[-latex] (-0.7, -0.01) -- (-0.7,-0.01);
      \draw[-latex] (4.7, -0.01) -- (4.7,-0.01);
    \end{scope}
  \end{tikzpicture}
  \caption{The linkages $\Lcal'_{ij}$ and $\Lcal'_{ji}$ can be concatenated to form many closed walks (left) or some long closed walks (right).
  }\label{fig:closed-long-walks}
\end{figure}
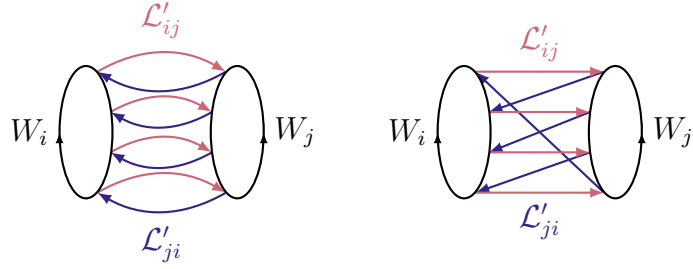

Hence, if the concatenation of the paths gives us enough closed walks, we are done.
Otherwise, $\Lcal_{ij}$ and $\Lcal_{ji}$ interleave in a way that creates at least one closed walk that is long enough to contain many paths from $\Lcal_{ij}$ and from $\Lcal_{ji}$.
As each of these paths contains at least two elements from the well-linked set, we create a walk that starts in the beginning of a path in $\Lcal_{ij}$ and zigzags going from $W_i$ to $W_j$ following paths in $\Lcal_{ij}$ and in $\Lcal_{ji}$ until it gathers enough vertices in the well-linked set.
Here, we need to clean this walk to behave like a path.
We show, in \cref{lem:untangled-walkage-to-walkage}, that \emph{untangled threaded walkages} are the structures which behave like paths.
This means that they can be broken down into smaller disjoint walks just as in the construction of a path system, in which a path containing a large well-linked set is broken down into smaller paths, each containing a large well-linked set.
Cleaning the built walk into an untangled threaded walkage is done in \cref{lem:untangled-walkage}.

Now, we have an \emph{$\ell$-linked walkage} $\Wcal_{ij}$ going from $W_i$ to $W_j$, for all $i$ and $j$, which is a walkage in which each of its walks contains $\ell$ vertices of a well-linked set.
We can then apply the same strategy as before, i.e.,\ find a linkage from the later portion to the earlier portion of each walk in $\Wcal_{ij}$ and use these to find a collection of closed walks that connect $W_i$ and $W_j$.
Note that this is only possible due to the nature of walkages, as the step that builds $\Wcal_{ij}$ is not trivially adapted to build a linkage containing many vertices from a well-linked set instead.
Indeed, to convert these walks into paths we need to discard pieces corresponding to self-intersections, which could potentially throw away a large portion of the well-linked set.

We call a walkage $\Wcal$ an (\emph{untangled}) \emph{threaded walkage} if $\Wcal$ behaves analogously to what an (untangled) threaded linkage would (except the threads must now be disjoint, see \cref{def:threaded-walkage} in \cref{sec:walkages} for more details).
\cref{fig:untangled-threaded-walkage} provides examples of threaded walkages.
We call a walkage $\Wcal$ a \emph{cyclic walkage} if there exists another walkage $\Tcal$ such that it is possible to concatenate each walk in $\Wcal$ with some walk in $\Tcal$ to form a closed walk and all these closed walks are vertex-disjoint.
\cref{def:cyclic-walkage} and \cref{fig:cyclic-walkage} provide a formal definition and an example of a cyclic walkage.

Given the strong system, we translate the proof in~\cite{RW12} to digraphs in \cref{lem:final-bramble}.
A high level overview of the proof of \cref{thm:main-proof-short} is as follows.
If our initial path system is sufficiently large, we find an $\ell$-linked untangled threaded walkage for an appropriate $\ell$.
We proceed to transform it into either a bramble of size $k$ and congestion $2$ or an $\ell'$-linked cyclic walkage, for some $\ell' < \ell$ still sufficiently large.
We can then assume we have an $\ell'$-linked cyclic walkage, which is enough for us to find either a bramble of size $k$ and congestion $2$ or an $(\ell'', w)$-strong system.
Assuming we have an $(\ell'', w)$-strong system, we can always find a bramble of size $k$ and congestion $2$ inside it.
\cref{fig:proof-overview} gives an illustration of this overview.

\newcounter{c}

\newsavebox{\bramblefig}
\begin{lrbox}{\bramblefig}
  \begin{tikzpicture}[
    >=Latex,
    bag/.style={draw=none, fill=#1, opacity=0.35, smooth cycle, tension=0.5}
]

    \node[blackvertex, scale=2] (v12) at (90:3.6) {};
    \node[blackvertex, scale=2] (v23) at (-30:3.6) {};
    \node[blackvertex, scale=2] (v13) at (210:3.6) {};

    \node[blackvertex, scale=2] (v24) at (30:1.6) {};
    \node[blackvertex, scale=2] (v34) at (-90:1.6) {};
    \node[blackvertex, scale=2] (v14) at (150:1.6) {};

    \begin{pgfonlayer}{background}
        \fill[bag=goodred] plot coordinates { (90:4.3) (150:3.3) (210:4.3) (150:0.6) };
        
        \fill[bag=goodgreen] plot coordinates { (90:4.3) (30:3.3) (-30:4.3) (30:0.6) };
        
        \fill[bag=goodblue] plot coordinates { (210:4.3) (-90:3.3) (-30:4.3) (-90:0.6) };

        \fill[bag=goodpurple] plot coordinates { (150:2.4) (30:2.4) (-90:2.4) };
    \end{pgfonlayer}

    \draw[arrow, line width=4pt] (v13) to [bend left = 30] (v12);
    \draw[arrow, line width=4pt] (v12) to [bend left = 30] (v23);
    \draw[arrow, line width=4pt] (v23) to [bend left = 30] (v13);

    \draw[arrow, line width=4pt] (v14) -- (v24);
    \draw[arrow, line width=4pt] (v24) -- (v34);
    \draw[arrow, line width=4pt] (v34) -- (v14);

    \draw[arrow, line width=4pt] (v12) -- (v14);
    \draw[arrow, line width=4pt] (v14) -- (v13);

    \draw[arrow, line width=4pt] (v23) -- (v24);
    \draw[arrow, line width=4pt] (v24) -- (v12);

    \draw[arrow, line width=4pt] (v13) -- (v34);
    \draw[arrow, line width=4pt] (v34) -- (v23);

\end{tikzpicture}%
\end{lrbox}

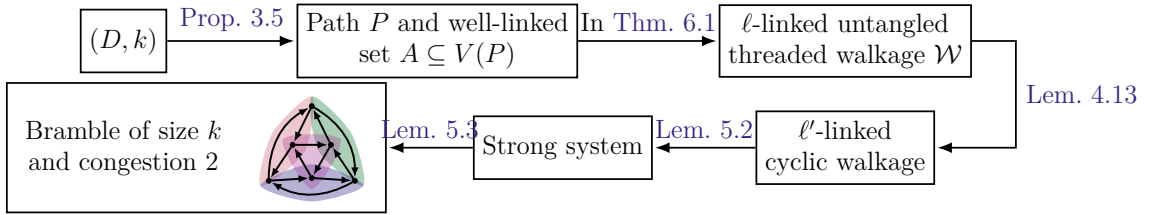
\begin{figure}[h!]
  \centering
  \scalebox{.83}{\begin{tikzpicture}[shorten >=1pt,arrow, thick,yscale=.85]

  \node[rectangle,draw, scale=1,minimum height=1cm,align=center] (P-a) at (0,0) {$(D,k)$};

  \node[rectangle,draw,scale=1,align=center] (P-b) at ($(P-a) + (5,0)$) {Path $P$ and well-linked\\set $A \subseteq V(P)$};

  \draw (P-a) -- node[pos=.5, above] {\hyperref[prop:linked-path]{Prop.~\labelcref{prop:linked-path}}} (P-b);

  \node[rectangle,draw,scale=1,align=center] (P-c) at ($(P-b) + (6.5,0)$) {$\ell$-linked untangled\\threaded walkage $\mathcal{W}$};

  \draw (P-b) -- node[pos=.5, above] {In \hyperref[thm:main-proof]{Thm.~\labelcref{thm:main-proof}} } (P-c);

  \node[rectangle, draw, scale=1, align=center] (P-h) at ($(P-c) + (0,-2)$) {$\ell'$-linked\\cyclic walkage};

  \draw (P-c.east) -- ++(.7,0) |- node[pos=.25, xshift=1.05cm] {\hyperref[lem:cyclic-walk-system]{Lem.~\labelcref{lem:cyclic-walk-system}}} (P-h.east);

  \node[rectangle, draw, align=center, minimum height=1cm] (P-i) at ($(P-h) + (-4.5, 0)$) {Strong system};

  \draw (P-h) -- node[pos=.5, above] {\hyperref[lem:fully-cyclic-walk-system]{Lem.~\labelcref{lem:fully-cyclic-walk-system}}} (P-i);

  \node[align=center] (P-j) at ($(P-i) + (-4, 0)$) {\scalebox{.22}{\usebox{\bramblefig}}};
  \node[align=center] (P-k) at ($(P-j) + (-3,0)$) {Bramble of size $k$\\and congestion 2};
  \node[rectangle, draw, fit=(P-j)(P-k), yscale=0.9] (P-l) {};

  \draw (P-i) -- node[pos=.5, above] {\hyperref[lem:final-bramble]{Lem.~\labelcref{lem:final-bramble}}} (P-l);
\end{tikzpicture}%
}%
  \caption{Overview of the proof of \cref{thm:main-proof-short}. The step from $P$ and $A$ to the $\ell$-linked untangled threaded walkage $\mathcal{W}$ occurs internally in the proof of \cref{thm:main-proof}, the more detailed version of \cref{thm:main-proof-short}. Other outcomes that can produce the bramble before \cref{lem:final-bramble} are ommited.}\label{fig:proof-overview}
\end{figure}

\paragraph{Unpacking \cref{lem:cyclic-walk-system}.}
Our goal is to obtain an $\ell'$-linked cyclic walkage from an $\ell$-linked untangled threaded walkage $\mathcal{W}$.
If we are not able to achieve this, then a large bramble with congestion $2$ is found instead.
We refer the reader to~\cref{fig:proof-overview-of-linkage-cyclic-walk-lemma} for an overview of the proof of \cref{lem:cyclic-walk-system}.

\begin{figure}[h]
  \centering
  \scalebox{.9}{\begin{tikzpicture}[shorten >=1pt,arrow, thick,yscale=.75]

\node[rectangle,draw,scale=1,align=center] (P-a) at (0,0) {Walkages $\mathcal{Q}_i$ starting and\\ending on a well-linked set};
\node[rectangle,draw,scale=1,align=center] (P-b) at ($(P-a) + (7,0)$) {Cyclic walkages\\with congestion 2};
\node[rectangle,draw,scale=1,align=center] (P-c) at ($(P-a) + (0,-2.5)$) {Threaded\\walkages};

\node[rectangle,draw,scale=1,minimum height=1cm,align=center] (P-d) at ($(P-b) + (5, 0)$) {Cyclic walkage};

\node[rectangle,draw,scale=1,align=center] (P-e) at ($(P-c) + (6,0)$) {Untangled threaded\\walkage};

\node[rectangle, draw, scale=1, align=center] (P-f) at ($(P-e) + (6, 0)$) {$\ell'$-linked\\cyclic walkage};

\draw (P-a) -- node[pos = .5] (P-ab) {} (P-b); 
\draw (P-a) -- node[pos = .32] (P-ac) {} (P-c); 
\node (P-abc) at ($(P-ab) + (-1.5, -1)$) {\hyperref[lem:structured-walkage-congestion]{Lem.~\labelcref{lem:structured-walkage-congestion}}};
\draw[-,dashed] (P-abc.north east) -- (P-ab);
\draw[-,dashed] (P-abc) -- (P-ac);

\draw (P-b) -- node[pos=.5, above] {\hyperref[lem:closed-walkage]{Lem.~\labelcref{lem:closed-walkage}}} (P-d);

\draw[shorten >= 7pt] (P-c) -- node[pos=.5] (P-cb) {} (P-b);
\draw (P-c) -- node[pos=.5] (P-ce) {} (P-e);
\node (P-cbe) at ($(P-ce) + (0.5,.5)$) {\hyperref[lem:untangled-walkage]{Lem.~\labelcref{lem:untangled-walkage}}};

\draw (P-d) -- node[pos=.5, xshift=1cm] {\hyperref[prop:independent-set-multipartite]{Prop.~\labelcref{prop:independent-set-multipartite}}} (P-f);

\draw (P-e) -- node[pos=.5, above] {\hyperref[cor:cyclic-extension]{Cor.~\labelcref{cor:cyclic-extension}}} (P-f);

\end{tikzpicture}%
}
  \caption{Overview of the proof of \cref{lem:cyclic-walk-system}.}\label{fig:proof-overview-of-linkage-cyclic-walk-lemma}
\end{figure}
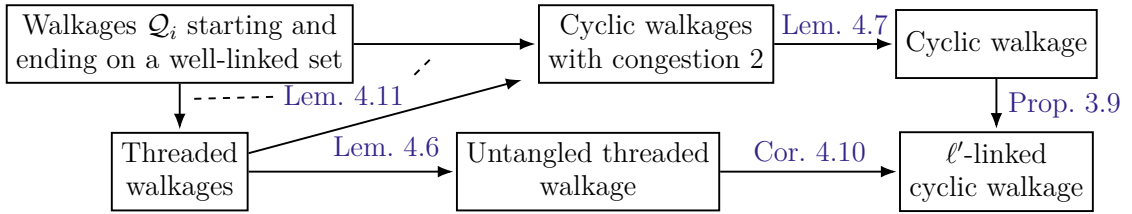

The first step of \cref{lem:cyclic-walk-system} uses the well-linkedness of $\mathcal{W}$ to split it into $w'$ walkages $\mathcal{Q}_1, \ldots, \mathcal{Q}_{w'}$ with the property that, for each $i \in [w']$, the union of the starting vertices $\stLink(\mathcal{Q}_i)$ with the ending vertices $\enLink(\mathcal{Q}_i)$ of $\mathcal{Q}_i$ forms a well-linked set.
By construction, $V(W_i)$ contains a well-linked set $A_i$ for each $W_i \in \mathcal{W}$ with $i \in [w']$ (for some choice of $w'$), and thus we choose $\mathcal{Q}_i$ as the walkage starting on the $k'$ last vertices of $A_i$, ordered as they appear in $W_i$, and ending on the first $k'$ vertices of $A_i$, for some choice of $k'$.
This is represented in \cref{fig:back-linkage} and is the starting point of \cref{fig:proof-overview-of-linkage-cyclic-walk-lemma}.

To each $\mathcal{Q}_i$ we apply \cref{lem:structured-walkage-congestion} which, in turn, exploits the well-linkedness of $\stLink(\mathcal{Q}_i) \cup \enLink(\mathcal{Q}_i)$ to find, for each $Q \in \mathcal{Q}_i$, back and forth walkages between sets of vertices.
Observing how the walks alternate between such sets, we either immediately find a cyclic walkage with congestion $2$ (since we started with two walkages) or a threaded walkage.
In the second case, we apply \cref{lem:untangled-walkage} and either untangle it or find a cyclic walkage with congestion $2$.
If the latter occurs, we clean the congestion via \cref{lem:closed-walkage}.

At this stage, from each $\mathcal{Q}_i$ we have extracted a $\mathcal{Q}'_i$ that is either a cyclic walkage or an untangled threaded walkage.
If any such $\mathcal{Q}'_i$ intersects $\mathcal{W}$ too heavily, we apply \cref{cor:cyclic-extension} to tie both together and obtain an $\ell'$-linked cyclic walkage.
Otherwise, we carefully dissect parts of both $\mathcal{W}$ and $\{\mathcal{Q}'_i \mid i \in [w']\}$ and analyze the degeneracy of the intersection graphs of those parts.
If it is sufficiently large for a given pair of parts, we tie them together into a cyclic walkage or a bramble using \cref{cor:cyclic-extension}.
In case no such pair of pieces is found, we apply \cref{prop:independent-set-multipartite}, exploiting the sparsity to find a pairwise vertex-disjoint set of walks $Q_i \in \mathcal{Q}'_i$, for each $i \in [w']$ which, in turn, are used to build an $\ell'$-linked walkage that is cyclic, completing \cref{lem:cyclic-walk-system}.
Crucially, the bowtie-like constructions of \cref{cor:cyclic-extension} (backed by \cref{lem:untangled-cyclic-extension-congestion,lem:cyclic-extension-congestion}) never yield congestion greater than 2, as we only operate on untangled structures.

\paragraph{Unpacking \cref{lem:fully-cyclic-walk-system}.} Here, we assume that we are given an $\ell$-linked and cyclic walkage $\mathcal{W}$ of size $w$, and the goal is to build a strong system from it.
Again, if we are not able to do this, then we find a large bramble with congestion $2$ instead.
We refer the reader to~\cref{fig:proof-sketch-fully-cyclic-walk-system} for an overview of the proof of~\cref{lem:fully-cyclic-walk-system}.

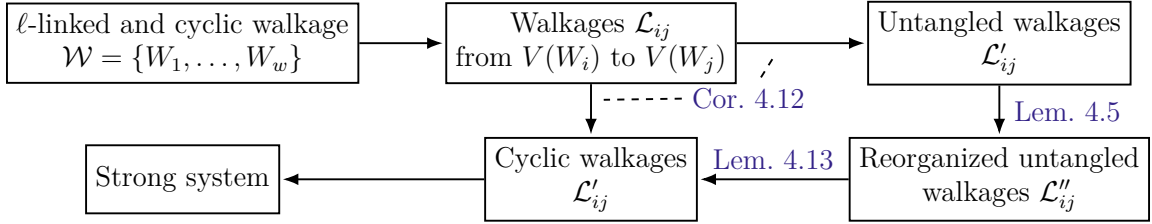
\begin{figure}[h]
  \centering
  \scalebox{.9}{\begin{tikzpicture}[shorten >=1pt,arrow, thick,yscale=.80]

  \node[rectangle,draw,scale=1,align=center] (P-a) at (0,0) {$\ell$-linked and cyclic walkage\\$\mathcal{W} = \{W_1, \ldots, W_w$\}};
  \node[rectangle,draw,scale=1,align=center] (P-b) at ($(P-a) + (6,0)$) {Walkages $\mathcal{L}_{ij}$\\from $V(W_i)$ to $V(W_j)$};
  \node[rectangle,draw,scale=1,align=center] (P-c) at ($(P-b) + (0,-2.5)$) {Cyclic walkages\\$\mathcal{L}'_{ij}$};

  \node[rectangle,draw,scale=1,minimum height=1cm,align=center] (P-d) at ($(P-b) + (6, 0)$) {Untangled walkages\\$\mathcal{L}'_{ij}$};

  \node[rectangle,draw,scale=1,align=center] (P-e) at ($(P-d) + (0,-2.5)$) {Reorganized untangled\\walkages $\mathcal{L}''_{ij}$};

  \node[rectangle, draw, scale=1, minimum height=1cm, align=center] (P-f) at ($(P-c) + (-6, 0)$) {Strong system};

  \draw (P-a) -- (P-b);
  \draw (P-b) -- node[pos = .3] (P-bc) {} (P-c);
  \draw (P-b) -- node[pos = .35] (P-bd) {} (P-d);
  \node (P-bcd) at ($(P-bd) + (-.5, -1)$) {\hyperref[cor:structured-walkage]{Cor.~\labelcref{cor:structured-walkage}}};
  \draw[-,dashed] (P-bcd) -- (P-bc);
  \draw[-,dashed] (P-bcd) -- (P-bd);

  \draw (P-d) -- node[pos=.5, xshift=1cm] {\hyperref[lem:untangled-walkage-to-walkage]{Lem.~\labelcref{lem:untangled-walkage-to-walkage}}} (P-e);

  \draw (P-e) -- node[pos=.5, above] {\hyperref[lem:cyclic-walk-system]{Lem.~\labelcref{lem:cyclic-walk-system}}} (P-c);
  \draw (P-c) -- (P-f);
\end{tikzpicture}%
}%
  \caption{Overview of the proof of~\cref{lem:fully-cyclic-walk-system}.}%
  \label{fig:proof-sketch-fully-cyclic-walk-system}
\end{figure}

By the well-linkedness of $\mathcal{W}$, for each distinct pair $i,j \in [w]$, there is a very large walkage $\mathcal{L}_{ij}$ from $V(W_i)$ to $V(W_j)$.
This allows us to apply, to each $\mathcal{L}_{ij}$, \cref{cor:structured-walkage} and find a large walkage $\mathcal{L}'_{ij}$ that is either cyclic or untangled.
We use the fact that untangled walkages behave like paths to reorganize them into $\ell$-linked walkages $\mathcal{L}''_{ij}$.
In turn, we use \cref{lem:cyclic-walk-system} to turn each $\mathcal{L}''_{ij}$ into a cyclic walkage.
At this point, for every $i$ and $j$, we may assume $\mathcal{L}'_{ij}$ is cyclic.
The construction of a strong system then follows by appropriately selecting cyclic walkages from the ones built so far, and observing that the construction ensures that each such walkage associated with the pair of distinct indices $i,j \in [w]$ contains a vertex in $W_i$ and one in $W_j$.

\cref{lem:fully-cyclic-walk-system} is a prime example of the utility of walkages.
The walkages $\Lcal''_{ij}$ are built by the concatenation of smaller walks, which is not trivial to achieve in a linkage setting because we cannot afford to discard vertices to transform walks back to paths, as this could throw away vertices from the well-linked set.

\paragraph{Unpacking \cref{lem:final-bramble}.} Here, we assume that we are given an $(\ell, k)$-strong system $(\mathcal{L,W})$ and the goal is to build a bramble of size $k$ and congestion 2.
See~\cref{fig:proof-sketch-final-bramble} for an overview of the proof of~\cref{lem:final-bramble}.

\begin{figure}[h]
  \centering
  \scalebox{.82}{\begin{tikzpicture}[shorten >=1pt,arrow, thick,yscale=.75]

\node[rectangle,draw,scale=1,align=center] (P-a) at (0,0) {Strong system\\$(\mathcal{L,W})$};
\node[rectangle,draw,scale=1,align=center] (P-b) at ($(P-a) + (0,-2.5)$) {Intersection graphs of\\$\mathcal{L}_{ij} \cup \mathcal{L}_{i'j'}$};
\node[rectangle,draw,scale=1,align=center, minimum height=1cm] (P-c) at ($(P-a) + (4.5,0)$) {All sparse pairs};
\node[rectangle,draw,scale=1,align=center] (P-d) at ($(P-c) + (5.5,0)$) {Pairwise vertex-disjoint\\walks, one from each $\mathcal{L}_{ij}$};

\node[rectangle,draw,scale=1,align=center, minimum height=1cm] (P-e) at ($(P-c) + (0,-2.5)$) {One dense pair};
\node[rectangle,draw,scale=1,align=center, minimum height=1cm] (P-f) at ($(P-d) + (5,0)$) {Clique minor};

\node[align=center] (P-j) at ($(P-e) + (4.5, 0)$) {\scalebox{.22}{\usebox{\bramblefig}}};
\node[align=center] (P-k) at ($(P-j) + (2.5,0)$) {Bramble of size $k$\\and congestion 2};
\node[rectangle, draw, fit=(P-j)(P-k), yscale=0.9] (P-l) {};

\draw (P-a) -- (P-b); 
\draw (P-b) -- (P-c.south west); 
\draw (P-b) -- (P-e); 
\draw (P-d) -- (P-f); 
\draw (P-c) -- node[pos=.5, above] {\hyperref[prop:independent-set-multipartite]{Prop.~\labelcref{prop:independent-set-multipartite}}} (P-d);
\draw (P-e) -- node[pos=.5, above] {\hyperref[prop:independent-set-multipartite]{Prop.~\labelcref{prop:dense-bramble}}} (P-l);
\draw (P-f) |- (P-l);

\end{tikzpicture}%
}%
  \caption{Overview of the proof of \cref{lem:final-bramble}.}%
  \label{fig:proof-sketch-final-bramble}
\end{figure}
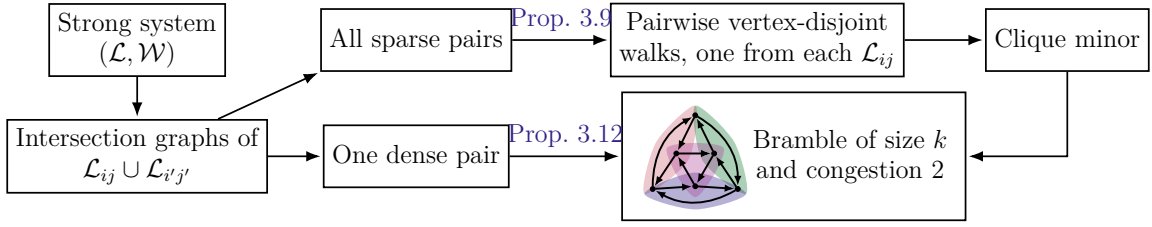

The proof closely emulates the undirected scenario seen in~\cite{RW12}.
We test the degeneracy of the intersection graph generated by each pair of distinct closed walkages $\mathcal{L}_{ij}, \mathcal{L}_{i'j'} \in \mathcal{L}$, all of which have size $\ell$.
If it is larger than some function of $k$ then we apply \cref{prop:dense-bramble} to find the desired bramble of size $k$ and congestion 2.
\cref{prop:dense-bramble} applies \cref{prop:clique-minor} to find a clique minor in the intersection graph which, in turn, is used to build the bramble.
Otherwise, all such pairs of walkages generate sparse intersection graphs, and we apply \cref{prop:independent-set-multipartite} to find a family $\mathcal{Z}$ of pairwise vertex-disjoint closed walks, one for each $\mathcal{L}_{ij}$.
Next, we observe that the intersection graph of $\mathcal{W} \cup \mathcal{Z}$ contains a $1$-subdivision of a large clique, since each pair of vertices associated with walks in $\mathcal{W}$ share a common neighbor associated with a walk in $\mathcal{Z}$.
Hence, the intersection graph contains a large clique minor, from which a large bramble with congestion $2$ is built.

\paragraph{Finding the brambles in polynomial time.}
The reasoning that the constructions of \cref{thm:main-proof,thm:main-proof-congestion-4-masarik,thm:main-proof-congestion-6-masarik} can be done in polynomial time assuming a path system is given is similar to the one done by \textcite{LS26} in order to produce an algorithmic version of the proof in~\cite{MasarikPilipczukRzazewskiSorge22}.
In short, the bulk of both constructions consists in building intersection graphs, testing for degeneracy, and iteratively following walks from walkages (in our case) or paths/walks from a path system (in the case of~\cite{MasarikPilipczukRzazewskiSorge22}), picking pieces satisfying prescribed properties and moving to the next step of the construction.
The degeneracy of the intersection graphs can be computed in polynomial time, as shown by \textcite{MatulaBeck83}.
As is the case of~\cite[Lemma 3.7]{MasarikPilipczukRzazewskiSorge22}, the running time of our bowtie-like lemmas, namely~\cref{lem:untangled-cyclic-extension-congestion,lem:cyclic-extension-congestion}, depends on \cref{prop:partition}, whose proof naturally yields a polynomial-time algorithm.
In addition, an algorithmic version of~\cref{prop:dense-bramble} is stated in~\cite[Lemma 24]{LS26}, and is obtained by emulating the proof available in~\cite[Lemma 3.1]{MasarikPilipczukRzazewskiSorge22} and applying~\cref{prop:clique-minor}.

Similarly to~\cite{LS26}, the main bottleneck on the running time is on the application of an algorithmic version of the Lovász Local Lemma~\cite{LE75}, which is the backend of both~\cref{prop:independent-set-multipartite-non-algo,prop:independent-set-multipartite}.
For simplicity, we elect to directly apply the algorithmic version from~\cite{Harris2022} through~\cref{prop:independent-set-multipartite}.

\paragraph{Brambles with congestion 4 and 6.} In \cref{sec:higher-congestion-with-masarik}, we show how our ideas can be used to modify the method by the authors in~\cite{MasarikPilipczukRzazewskiSorge22} in order to produce brambles with congestion $4$ and $6$ with better bounds on the directed treewidth.
We do so by adapting the definition of path systems into \emph{walk systems} (see \cref{def:walk-system}), by replacing ``linkages'' with ``walkages'' (notice that this is not the same as strong systems, where all the walks are also required to be closed).
Then, in \cref{lem:sparse-bramble}, we adapt the sparse winning scenario of~\cite[Lemma 3.2]{MasarikPilipczukRzazewskiSorge22} to work with walk systems and show that this scenario yields a bramble with congestion at most $4$.
\cref{lem:final-bramble-clean-4,lem:final-bramble-clean-6} are then used to build brambles with congestion at most $4$ and at most $6$, respectively, exploring dense and sparse cases of intersection graphs.
The dense cases are handled by \cref{cor:cyclic-extension,cor:structured-walkage}, and the sparse cases are handled by \cref{lem:sparse-bramble}.
Finally, \cref{thm:main-proof-congestion-4-masarik,thm:main-proof-congestion-6-masarik} wrap up the proofs and bound the dependence on the directed treewidth.
We remark that even though walkages are not strictly necessary in these proofs for higher congestion, we still use them so we can build upon the results of \cref{sec:walkages}.

\section{Preliminaries}\label{sec:pre}

In this section we provide the definitions and notation used in this text.
Given a positive integer $k$ and a set $A$, we denote by $[k]$ the set $\{1, 2, \ldots, k\}$ and by $\binom{A}{k}$ the family of subsets of $A$ of size $k$.

\subsection{Graphs and digraphs}
For basic background on graph theory we refer the reader to~\cite{Bondy2008}.
Given a (di)graph $G$ we denote by $V(G)$ and $E(G)$ the sets of vertices and edges (arcs) of $G$, respectively.
For $X \subseteq V(G)$, we denote by $G[X]$ the subgraph of $G$ \emph{induced} by $X$.
For $v \in V(G)$ we denote by $G - v$ the (di)graph resulting from deleting $v$ from $G$.

If $e$ is an arc of a digraph from a vertex $u$ to a vertex $v$, we say that $e$ has \emph{endpoints} $u$ and $v$, that $e$ is \emph{incident} to $u$ and $v$, and that $e$ is \emph{oriented} from $u$ to $v$.
We may refer to $e$ as the ordered pair $(u,v)$.
In this case, $u$ is the \emph{tail} of $e$ and $v$ is the \emph{head} of $e$.
We also say that $e$ is \emph{leaving} $u$ and \emph{reaching} $v$, and that $u$ and $v$ are \emph{adjacent}.
A \emph{clique} in a (di)graph $D$ is a set of pairwise adjacent vertices of $D$, and an \emph{independent set} is a set of pairwise non-adjacent vertices of $D$.

The \emph{in-degree} (resp.\ \emph{out-degree}) of a vertex $v$ in a digraph $D$ is the number of arcs with head (resp.\ tail) $v$.
We say that $u$ is an \emph{in-neighbor} of $v$ in $D$ if $(u,v) \in E(D)$ and that $u$ is an \emph{out-neighbor} of $v$ if $(v,u) \in E(D)$.
If $u$ is an in-neighbor and/or an out-neighbor of $v$, we may simply say that $u$ is a \emph{neighbor} of $v$.

In an undirected graph $G$, we say that $u$ is a \emph{neighbor} of $v$ if $(u,v) \in E(G)$.
The \emph{degree} of a vertex $u \in V(G)$ is the number of edges having $u$ as an endpoint.
The \emph{degeneracy} $\degen(G)$ of an undirected graph $G$ is the least integer $k$ such that every subgraph $H$ of $G$ contains a vertex of degree at most $k$.

Let $G$ be a graph and $e$ be an edge from  $u$ to $v$ in $G$.
By \emph{contracting} $e$, we generate a new graph $G'$ by identifying vertices $u$ and $v$ into a single new vertex $w$ and, for every edge $(y,z) \in E(G)$ with $z \in \{u,v\}$ other than $e$, we add an edge from $y$ to $w$ in $G'$.
A graph $H$ is a \emph{minor} of $G$ if a copy of $H$ can be obtained from a subgraph of $G$ by a series of edge contractions.

A \emph{walk} in a graph $G$ is an alternating sequence $W$ of vertices and edges that starts and ends with a vertex, and such that for every edge $(u,v)$ in the walk, vertex $u$ (resp.\ vertex $v$) is the element right before (resp.\ right after) arc $(u,v)$ in $W$.
If the first vertex in a walk is $u$ and the last one is $v$, then we say this is a \emph{walk from $u$ to $v$} and a \emph{closed walk} if $u = v$.

A \emph{path} is a graph formed by pairwise distinct vertices $V(P) = \{v_1, \ldots, v_k\}$ and $E(P) = \{(v_i, v_{i+1}) \mid i \in [k-1]\}$ for some positive integer $k$.
We say that $v_1$ is the \emph{first} vertex of $P$, $v_k$ is the \emph{last} vertex of $P$ and $P$ is a path \emph{from} $v_1$ \emph{to} $v_k$.
If $P$ is a walk or a path, we denote by \emph{$\stLink(P)$} and \emph{$\enLink(P)$} the first and last vertices of $P$, respectively.
\emph{Directed paths} and \emph{directed walks} are defined analogously, and we may refer to them as simply ``paths'' and ``walks'' when the context is clear.

Given a family of walks $\Wcal$, directed or not, we denote by $I(\Wcal)$ its \emph{intersection graph}, i.e.,\ the undirected graph that has a vertex for each walk in $\Wcal$ and edges between walks that intersect.
Given two walks $W_1$ and $W_2$ such that $\enLink(W_1) = \stLink(W_2)$, we denote by $W_1 \cdot W_2$ its \emph{concatenation}, i.e.,\ a walk that starts in the first vertex of $W_1$, follows $W_1$ until its last vertex, which is the first vertex of $W_2$, and then follows $W_2$ until its last vertex.

A digraph $D$ is \emph{strongly connected} (or \emph{strong} for short) if for every $u,v \in V(D)$ there are paths from $u$ to $v$ and from $v$ to $u$ in $D$.
A \emph{separator} of $D$ is a set $X \subseteq V(D)$ such that $D \setminus X$ has a single vertex or is \textsl{not} strongly connected.
If $D$ has at least $k+1$ vertices and every separator of $D$ has size at least $k$, we say that $D$ is \emph{$k$-strongly connected} (or \emph{$k$-strong} for short).
A \emph{strongly connected component} (or \emph{strong component} for short) of a digraph $D$ is a maximal induced subgraph of $D$ that is strongly connected.

For the remainder of this article and unless stated otherwise, $n$ is used to denote the number of vertices of the input (di)graph of the problem under consideration.

\subsection{Parameterized complexity}\label{sec:param-compl}
We refer the reader to~\cite{Cygan2015,Downey2013} for background on parameterized complexity, and we define here only the most basic definitions.
A \emph{parameterized problem} is a language $L \subseteq \Sigma^* \times \mathbb{N}$.  For an instance $I=(x,k) \in \Sigma^* \times \mathbb{N}$, $k$ is called the \emph{parameter}.
A parameterized problem $L$ is \emph{fixed-parameter tractable} ({\FPT}) if there exists an algorithm $\mathcal{A}$, a computable function $f$, and a constant $c$ such that given an instance $I=(x,k)$, $\mathcal{A}$   (called an {\FPT} \emph{algorithm}) correctly decides whether $I \in L$ in time bounded by $f(k) \cdot |I|^c$. For instance, the \textsc{Vertex Cover} problem parameterized by the size of the solution is \FPT.
A parameterized problem $L$ is in {\sf XP} if there exists an algorithm $\mathcal{A}$ and two computable functions $f$ and $g$ such that given an instance $I=(x,k)$, $\mathcal{A}$  (called an {\sf XP} \emph{algorithm}) correctly decides whether $I \in L$ in time bounded by $f(k) \cdot |I|^{g(k)}$. For instance,  the \textsc{Clique} problem parameterized by the size of the solution is in  {\sf XP}.

Within parameterized problems, the class {\sf W}[1] may be seen as the parameterized equivalent to the class {\sf NP} of classical decision problems. Without entering into details, a parameterized problem being {\sf W}[1]-\emph{hard} can be seen as a strong evidence that this problem is not {\FPT}.
The canonical example of {\sf W}[1]-hard problem is \textsc{Clique}  parameterized by the size of the solution.

\subsection{Directed treewidth, brambles, and well-linked sets}\label{sec:dtw-brambles-well-linked-sets}
The \emph{directed treewidth} of digraphs was introduced in~\cite{Johnson2001} as a directed analog of treewidth of undirected graphs, together with a dual notion for directed treewidth in the form of \emph{havens}.
Since havens and \emph{brambles} in digraphs are easily interchangeable (see~\cite[Chapter 6]{Matthias2014}, for example), we focus only on the latter.
In addition, as we do not explicitly work with the decompositions, we refer the reader to~\cite{Johnson2001,KawarabayashiK15,CLMS22} for the formal definition.
We begin by formally defining the \emph{congestion} of a family of sets of vertices.

\begin{definition}[Congestion]\label{def:congestion}
  Let $D$ be a digraph and $\mathcal{S}$ be a collection of non-empty subsets of $V(D)$.
  We define the \emph{congestion} of $\mathcal{S}$ as $\max_{v \in V(D)}|\{S \in \mathcal{S} \mid v \in S\}|$.
\end{definition}

Although brambles and well-linked sets were discussed in \cref{sec:intro}, we repeat the definition here for completeness.

\begin{definition}[Brambles in digraphs]\label{def:brambles}
  A \emph{bramble} $\mathcal{B}$ in a digraph $D$ is a collection of strong subgraphs of $D$ such that if $B, B' \in \mathcal{B}$ then $V(B) \cap V(B') \neq \emptyset$ or there are arcs in $D$ from $V(B)$ to $V(B')$ and from $V(B')$ to $V(B)$.
  We say that the elements of $\mathcal{B}$ are the \emph{bags} of $\mathcal{B}$.
  The \emph{order} of $\mathcal{B}$ is the minimum size of a set $X \subseteq V(D)$ intersecting every element of $\mathcal{B}$.
  We denote by $\bn(D)$ the maximum order of a bramble in $D$.
  The \emph{congestion} of $\mathcal{B}$ is the congestion of $\{V(B) \mid B \in \mathcal{B}\}$.
\end{definition}
We say that $\mathcal{B}$ is \emph{strong} if its bags pairwise intersect.
\cref{def:brambles} is the (more permissive) definition most frequently used in related work, while strong brambles are more restrictive and, contrary to the former, closed under butterfly minors.
We remark that the brambles that we build for congestion $2$ are strong.
See \cref{fig:bramble-example} for an example of a (strong) bramble.
Notice that if $\mathcal{B}$ is a bramble of congestion $c$ then its order is at least $\lceil |\mathcal{B}|/c \rceil$.
In undirected graphs, brambles are an exact dual object to treewidth; in digraphs, brambles are still dual, although the relation between $\dtw(D)$ and $\bn(D)$ is not as tight.

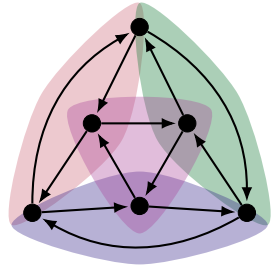
\begin{figure}[h]
  \centering
  \scalebox{.65}{\begin{tikzpicture}[
    scale=.7,
    >=Latex,
    bag/.style={draw=none, fill=#1, opacity=0.35, smooth cycle, tension=0.5}
]
    \node[blackvertex, scale=2] (v12) at (90:3.6) {};
    \node[blackvertex, scale=2] (v23) at (-30:3.6) {};
    \node[blackvertex, scale=2] (v13) at (210:3.6) {};

    \node[blackvertex, scale=2] (v24) at (30:1.6) {};
    \node[blackvertex, scale=2] (v34) at (-90:1.6) {};
    \node[blackvertex, scale=2] (v14) at (150:1.6) {};

    \begin{pgfonlayer}{background}
        \fill[bag=goodred] plot coordinates { (90:4.3) (150:3.3) (210:4.3) (150:0.6) };
        
        \fill[bag=goodgreen] plot coordinates { (90:4.3) (30:3.3) (-30:4.3) (30:0.6) };
        
        \fill[bag=goodblue] plot coordinates { (210:4.3) (-90:3.3) (-30:4.3) (-90:0.6) };

        \fill[bag=goodpurple] plot coordinates { (150:2.4) (30:2.4) (-90:2.4) };
    \end{pgfonlayer}

    \draw[arrow, very thick] (v13) to [bend left = 30] (v12);
    \draw[arrow, very thick] (v12) to [bend left = 30] (v23);
    \draw[arrow, very thick] (v23) to [bend left = 30] (v13);

    \draw[arrow, very thick] (v14) -- (v24);
    \draw[arrow, very thick] (v24) -- (v34);
    \draw[arrow, very thick] (v34) -- (v14);

    \draw[arrow, very thick] (v12) -- (v14);
    \draw[arrow, very thick] (v14) -- (v13);

    \draw[arrow, very thick] (v23) -- (v24);
    \draw[arrow, very thick] (v24) -- (v12);

    \draw[arrow, very thick] (v13) -- (v34);
    \draw[arrow, very thick] (v34) -- (v23);

\end{tikzpicture}}
  \caption{Example of a bramble of order and congestion 2. The bags are denoted by the shaded regions.}\label{fig:bramble-example}
\end{figure}

\begin{proposition}[{See e.g.~\cite[Corollary~6.4.24]{KreutzerOrdyniak14}}]\label{prop:bramble-dtw}
  For any digraph $D$, $\dtw(D) = \Theta(\bn(D))$.
\end{proposition}

\begin{definition}[Well-linked sets]\label{def:well-linked-sets}
  Let $D$ be a digraph and $X \subseteq V(D)$.
  We say that $X$ is \emph{well-linked} if, for all disjoint $A, B \subseteq X$ with $|A| = |B|$, there are $|A|$ disjoint paths from $A$ to $B$ in $D$.
\end{definition}
In~\cite{Johnson2001}, the authors gave an $\XP$ algorithm that either correctly decides that a digraph $D$ has $\dtw(D) \leq 3k-2$ or produces a bramble of order $k$.
This was later improved to an {\FPT} algorithm in~\cite[Theorem~2.17]{CLMS22}.
Then, in~\cite[Theorem 2.23]{CLMS22} the authors show how to use the particular type of bramble obtained by~\cite[Theorem 2.17]{CLMS22} to build, in \FPT{} time, a path $P$ with $V(P)$ containing the vertices of a large well-linked set of $D$.
Formally, pipelining those two aforementioned results, the following is obtained.

\begin{proposition}[{\cite[Theorem 2.17 + Theorem 2.23]{CLMS22}}]\label{prop:linked-path}
  Given $k \in \N$, let $\wLP(k) = 3(k+1)(\lfloor k/2 \rfloor + 1) - 5 = \Ocal(k^2)$.
  There exists an algorithm running in $2^{\Ocal(k^2 \cdot \log k)} \cdot n^{\Ocal(1)}$ time that receives as input a digraph $D$ and outputs either
  \begin{enumerate}[label=\textnormal{(\alph*)}]
    \item An arboreal decomposition of $D$ of width at most $\wLP(k)$; or
    \item A path $P$ in $D$ and a well-linked set $A$ of size $k$ such that $A \subseteq V(P)$.
  \end{enumerate}
\end{proposition}

We see in \cref{sec:walk-systems} that the path $P$ and the well-linked set $A$, as mentioned in the second possible output of \cref{prop:linked-path}, form the backbone of the fundamental structure, namely \emph{path systems}, from where the main proof of~\cite{MasarikPilipczukRzazewskiSorge22} begins.
Our proof for congestion $2$ skips path systems in favor of \emph{strong} systems, while the proofs for congestion $4$ and $6$ simply relax path systems to \emph{walk systems}.

\subsection{Clique minors and multicolored independent sets}
A common strategy shared between this work and~\cite{MasarikPilipczukRzazewskiSorge22} is to explore the interplay of sparse and dense cases of intersection graphs built from sets of closed walks.
To do so, we observe the degeneracy $\degen(G)$ of the obtained graph $G$.
In case it is sufficiently small, the goal is to apply the following result to find multicolored independent sets in $G$.

\begin{proposition}[{\cite[Lemma~4.3]{RW12}}]\label{prop:independent-set-multipartite-non-algo}
  Given $d, r \in \N$, let $\lISMold(d, r) = 4e(r-1)d = \Ocal(d  r)$.
  Let $V_1, V_2, \dots, V_r$ be the color classes in a (not necessarily proper) $r$-coloring of the vertices of a graph $H$.
  If $|V_i| \ge \lISMold(d, r)$ for all $i \in [r]$ and $\degen(H[V_i \cup V_j]) \le d$ for all distinct $i, j \in [r]$, then there exists an independent set $\{v_1, v_2, \dots, v_r\}$ of $H$ such that each $v_i \in V_i$.
\end{proposition}

Otherwise, if $\degen(G)$ is ``large'' then, since every graph of degeneracy $d$ has a subgraph of average degree at least $d$, we can apply the result by \textcite[Theorem~1]{Kostochka84} (see also~\cite{Thomason84}) to find a large clique minor in $G$.

\begin{proposition}[{\cite[Theorem~1]{Kostochka84}}]\label{prop:clique-minor-non-algo}
  There exists a constant $\cCMold > 0$ such that the following holds.
  Given $t \in \N$, let \(\dCMold(t) = \cCMold \cdot t \sqrt{\log t} = \Ocal(t\sqrt{\log t})\).
  Any graph with average degree at least $\dCMold(t)$ contains $K_t$ as a minor.
\end{proposition}

Both \cref{prop:independent-set-multipartite-non-algo,prop:clique-minor-non-algo}, however, are not algorithmic.
For the latter, with the goal of obtaining a polynomial-time algorithm to find a bramble with congestion $2$ when given a path system, we can use the following algorithmic version.

\begin{proposition}[{\cite{DHJRW2013}}]\label{prop:clique-minor}
  There exists a constant $\cCM > 0$ such that the following holds.
  Given $t \in \N$, let \(\dCM(t) = \cCM \cdot t \sqrt{\log t} = \Ocal(t\sqrt{\log t})\).
  Any graph with average degree at least $\dCM(t)$ contains $K_t$ as a minor that can be found in $\mathcal{O}(n)$ time.
\end{proposition}

For an algorithmic version of \cref{prop:independent-set-multipartite-non-algo}, some alternatives need to be considered.
The proof of \cref{prop:independent-set-multipartite-non-algo} depends on the Lovász Local Lemma~\cite{LE75} (henceforth abbreviated as the LLL).
Informally, the LLL says that, when given a set $\mathcal{X}$ of ``undesirable events'' on a probability space, if the events in $\mathcal{X}$ are mostly independent of one another and occur with reasonably low probabilities, then there is a configuration where no event in $\mathcal{X}$ occurs.
In the context of \cref{prop:independent-set-multipartite-non-algo}, the ``undesirable events'' are the existence of edges between distinct parts $V_i, V_j$, and the ``reasonably low'' probability that edges appears comes from the bound on the degeneracy of the graphs of the form $H[V_i \cup V_j]$ with $i \neq j$.

Since the original proof of the LLL is not constructive, a lot of work has been done into finding algorithmic versions of the LLL\@.
\textcite{MT2010} gave a constructive proof that ends in expected polynomial time.
\textcite{CGH2013} showed that an algorithmic proof is also possible with a slightly weaker statement, which is used by \textcite{LS26} to adapt the proof of~\cite{MasarikPilipczukRzazewskiSorge22} into an algorithm.
\textcite{Harris2022} gave an algorithmic version of the LLL with the same bounds as the original proof, but without specifying the degree of the polynomial in its polynomial running time.
For simplicity, and since using the version in~\cite{CGH2013} would imply an increase in the dependence on $k$ in the statement of \cref{thm:main-proof-short}, we elect to use the result by \textcite{Harris2022}, which appears implicitly in the following.

\begin{proposition}[{\cite[Theorem~1.5]{Harris2022}}, degeneracy version]{\label{prop:independent-set-multipartite}}
  Let $\varepsilon > 0$ be a constant.
  Given $d, r \in \N$, let $\lISM(d, r) = (4+2\varepsilon)(r-1)d = \Ocal(d  r)$.
  Let $V_1, V_2, \dots, V_r$ be the color classes in a (not necessarily proper) $r$-coloring of the vertices of a graph $H$.
  If $|V_i| \ge \lISM(d, r)$ for all $i \in [r]$ and $\degen(H[V_i \cup V_j]) \le d$ for all distinct $i, j \in [r]$, then there is an algorithm running in polynomial time that finds an independent set $\{v_1, v_2, \dots, v_r\}$ of $H$ such that each $v_i \in V_i$.
\end{proposition}

We remark that the original statement of~\cite[Theorem~1.5]{Harris2022} uses maximum degree instead of degeneracy to bound the number of edges between distinct parts $V_i, V_j$.
The condition that $\degen(H[V_i \cup V_j]) \le d$ for all distinct $i, j \in [r]$ implies that each vertex of a given $V_i$ has at most $2d(r-1)$ neighbors in the other sets $V_j$ with $j \neq i$.
Thus, we can substitute the condition $|V_i| \geq (2+\varepsilon)\Delta$ of~\cite[Theorem 1.5]{Harris2022}, where $\Delta$ is the maximum degree of the given graph, by $|V_i| \geq (4+2\varepsilon)(r-1)d$ as in \cref{prop:independent-set-multipartite}.
Lastly, we also use the following auxiliary results.

\begin{proposition}\label{prop:average-min-degree}
  Given $d > 0$, if a graph $G$ has average degree \(d\), then \(G\) contains a subgraph of minimum degree larger than \(d/2\).
\end{proposition}

\begin{proof}
  For $v \in V(G)$, let $\deg(v)$ be the degree of $v$.
  First, note that if \(v \in V(G)\) has degree at most \(d/2\), then the average degree of \(G - v\) is at least
  \[
    \frac{|V(G)|d - 2\deg(v)}{|V(G)|-1} \ge \frac{|V(G)|d - d}{|V(G)|-1} = \frac{d(|V(G)|-1)}{|V(G)|-1} = d.
  \]
  Since removing a vertex of degree at most \(d/2\) does not decrease the average degree, iteratively delete from \(G\) any vertex of degree at most \(d/2\) to obtain a subgraph \(H\).
  As \(H\) has positive average degree, it is non-empty and by construction it has no vertex of degree at most \(d/2\).
\end{proof}

We remark that \cref{prop:clique-minor-non-algo,prop:clique-minor} also work for a graph with degeneracy \(\Omega(t\sqrt{\log t})\), possibly with a different constant factor, due to \cref{prop:average-min-degree}.
For simplicity, we assume the constants \(\cCMold, \cCM \) work for both versions of this result.

Given a sequence $X = (x_1, x_2, \dots, x_a)$, a \emph{segment} of $X$ is a set of consecutive elements of $X$, that is, a set of the form $\{x_i, x_{i+1}, \dots, x_j\}$ for some $1 \le i \le j \le a$.

\begin{proposition}[{\cite[Lemma~15]{MasarikMuziPilipczukRzazewskiSorge21}}]\label{prop:partition}
  Given $k, r \in \N$, let \(\dP(k,r) = 2^9 \cdot k  r = \Ocal(k  r)\).
  Let $G$ be a bipartite graph with vertex-ordered bipartition classes $X = (x_1, x_2, \dots, x_a)$ and $Y = (y_1, y_2, \dots, y_b)$.
  If the minimum degree of $G$ is at least $\dP(k,r)$, then there are $k$ sets $X_1, X_2, \dots, X_k$, and $k$ sets $Y_1, Y_2, \dots, Y_k$, such that:
  \begin{enumerate}[label=\textnormal{(\alph*)}]
    \item\label{item:partition-prop-1} For each $i \in [k]$, the set $X_i$ is a segment of $X$ and the set $Y_i$ is a segment of $Y$;
    \item\label{item:partition-prop-2} For each distinct $i, j \in [k]$, we have $X_i \cap X_j = \emptyset$ and $Y_i \cap Y_j = \emptyset$; and
    \item\label{item:partition-prop-3} For each $i \in [k]$, the average degree of the graph $G[X_i \cup Y_i]$ is at least $r$.
  \end{enumerate}
\end{proposition}

\begin{proposition}[{\cite[Lemma~3.1]{MasarikPilipczukRzazewskiSorge22}}]\label{prop:dense-bramble}
  If a digraph $D$ contains a family $\Zcal$ of closed walks with congestion at most $c$ such that $\degen(\intersec(\Zcal)) \ge \dCM(k)$, then $D$ contains a bramble of size $k$ and congestion at most $c$.
\end{proposition}

We remark that the proof of \cref{prop:dense-bramble} given in~\cite{MasarikPilipczukRzazewskiSorge22} is not algorithmic as it uses~\cite[Theorem~1]{Kostochka84}.
We can, however, substitute this latter result by its algorithmic counterpart given in \cref{prop:clique-minor}, as done by \textcite{LS26}, and assume \cref{prop:dense-bramble} is also algorithmic.

\section{Walkages}\label{sec:walkages}

In this section we introduce the concept of \emph{walkages}, which generalize linkages by replacing paths with walks.
Walkages are a key notion in this work, as they allow us to simplify arguments that use both linkages and well-linked sets.
We also establish several properties of walkages that play a central role in the process of building low congestion brambles.

\begin{definition}[Walkages]\label{def:walkages}
  Given a digraph $D$, a \emph{walkage} is a collection of pairwise vertex-disjoint walks in $D$.
  Given a digraph $D$ and $A, B \subseteq V(D)$, we say a walkage $\Wcal$ is an $(A, B)$-walkage if the first (resp.\ last) vertex of each walk in $\Wcal$ is in $A$ (resp.\ in $B$).
  A \emph{subwalkage} of a walkage \(\Wcal \) is a walkage which is a subset of \(\Wcal \) (that is, every walk of a subwalkage of \(\Wcal \) is a walk of \(\Wcal \)).
  A walkage is \emph{closed} if all of its walks are closed.
  The \emph{size} of the walkage is the number of walks it contains.
  A walkage $\Wcal$ of size $w$ is said to be \emph{$\ell$-linked} if each walk $W_i \in \Wcal$ contains a set $A_i \subseteq V(W_i)$ of size at least $\ell$ such that $\bigcup_{i \in [w]} A_i$ is well-linked.

  For a positive integer $c$, a \emph{walkage with congestion $c$} is a set of walks $\mathcal{W}$ where the congestion of $\{V(W) \mid W \in \mathcal{W}\}$ is $c$.
  We extend the notions of $(A, B)$-walkages, subwalkages, closed walkages, and $\ell$-linked walkages to walkages with congestion.

\end{definition}

Linkages and well-linked sets are fundamental tools in structural (di)graph theory.
A large linkage provides a way to route many disjoint paths between prescribed sets of vertices, while a large well-linked set serves both as an obstruction to  small directed treewidth and as a way to find such linkages.
Linkages in which each path contains a large subset of a well-linked set are particularly useful to build highly connected objects, such as path systems (as in~\cite{KK22, MasarikPilipczukRzazewskiSorge22, HKMM24}), which are the basic building blocks for more complex structures, like brambles of low congestion~\cite{MasarikPilipczukRzazewskiSorge22} or even the cylindrical grid~\cite{KK22}.

At first glance, walkages and linkages are equivalent: the existence of one implies the existence of the other, the main difference being that walks in a walkage may self-intersect, whereas paths in a linkage do not.
However, working with walkages is simpler, as we avoid walk-to-path polishing steps.
In \cref{fig:walkages}, we see an example of a walkage which can be converted into a linkage, but if we convert the walks into paths, the set of removed vertices may contain vertices in a particular well-linked set.
We introduce walkages to overcome this issue, and to develop new techniques, which are not possible with linkages.

\begin{figure}[htbp]
  \centering
  \begin{tikzpicture}[yscale=.75]
    \node[draw, circle, fill = black, inner sep = 0pt, minimum size = 3pt] (00) at (0,0) {};
    \node[draw, circle, fill = black, inner sep = 0pt, minimum size = 3pt] (01) at (1,0) {};
    \node[draw, circle, fill = black, inner sep = 0pt, minimum size = 3pt] (02) at (2,0) {};
    \node[draw, circle, fill = black, inner sep = 0pt, minimum size = 3pt] (03) at (3,0) {};
    \node[draw, circle, fill = black, inner sep = 0pt, minimum size = 3pt] (04) at (1,1) {};
    \node[draw, circle, fill = black, inner sep = 0pt, minimum size = 3pt] (05) at (2,1) {};

    \draw[thick, -latex] (00) -- (01);
    \draw[thick, -latex] (01) to[bend left] (04);
    \draw[thick, -latex] (04) to[bend left] (01);
    \draw[thick, -latex] (01) -- (02);
    \draw[thick, -latex] (02) to[bend left] (05);
    \draw[thick, -latex] (05) to[bend left] (02);
    \draw[thick, -latex] (02) -- (03);

    \node[draw, circle, fill = black, inner sep = 0pt, minimum size = 3pt] (10) at (0,-1.2) {};
    \node[draw, circle, fill = black, inner sep = 0pt, minimum size = 3pt] (11) at (1.5,-1.2) {};
    \node[draw, circle, fill = black, inner sep = 0pt, minimum size = 3pt] (12) at (3,-1.2) {};

    \draw[thick, -latex] (10) -- (11);
    \draw[thick, -latex] (11) -- (12);

    \node[draw, circle, fill = black, inner sep = 0pt, minimum size = 3pt] (20) at (0,-2.4) {};
    \node[draw, circle, fill = black, inner sep = 0pt, minimum size = 3pt] (21) at (1.5,-2.4) {};
    \node[draw, circle, fill = black, inner sep = 0pt, minimum size = 3pt] (22) at (3,-2.4) {};
    \node[draw, circle, fill = black, inner sep = 0pt, minimum size = 3pt] (23) at (1,-3.2660254) {};
    \node[draw, circle, fill = black, inner sep = 0pt, minimum size = 3pt] (24) at (2,-3.2660254) {};

    \draw[thick, -latex] (20) -- (21);
    \draw[thick, -latex] (21) -- (23);
    \draw[thick, -latex] (23) -- (24);
    \draw[thick, -latex] (24) -- (21);
    \draw[thick, -latex] (21) -- (22);

    \draw[dashed] (0,-1.2) ellipse (0.6cm and 2.2cm);
    \draw[dashed] (3,-1.2) ellipse (0.6cm and 2.2cm);

    \node[left] at (-0.6,-1.2) {$A$};
    \node[right] at (3.6,-1.2) {$B$};
  \end{tikzpicture}

  \caption{Example of an $(A,B)$-walkage of size $3$ in a digraph.}\label{fig:walkages}
\end{figure}
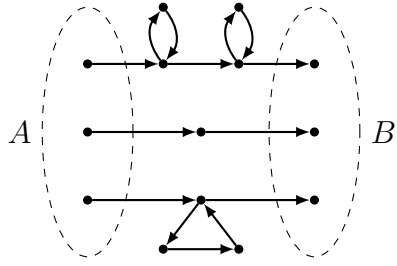

We remark that \cref{def:walkages} builds on the concept of half-integral linkages (see e.g.~\cite{Edwards2017,MasarikMuziPilipczukRzazewskiSorge21}) and generalizes it further allowing an arbitrary congestion.
This is particularly useful to reduce the degree of the polynomial that bounds the directed treewidth when building brambles of larger congestion, namely $4$ and $6$.
The relaxation of the congestion allows us to save up on some constructions that need not be as strict as in the congestion $2$ case.

We now extend the notion of \emph{threaded linkages}, introduced in~\cite[Definition~3.4]{MasarikPilipczukRzazewskiSorge22}\footnote{This concept seems to have first appeared in~\cite{MasarikMuziPilipczukRzazewskiSorge21} under the name of \emph{interlaced-walk} but played a smaller role there.}.

\begin{definition}[Threaded walkages]\label{def:threaded-walkage}
  Let \(\Wcal = \{ W_1, W_2, \dots, W_\ell \} \) be a walkage in a digraph $D$.
  A \emph{threaded walkage on} \(\Wcal \) is a pair \((\Wcal, \Tcal)\), where \(\Tcal = \{T_1, T_2, \dots, T_{\ell - 1}\} \) is a walkage in $D$ and for each \(i \in [\ell - 1] \) we have \(\stLink(T_i) = \enLink(W_i)\) and \(\enLink(T_i) = \stLink(W_{i+1})\).
  The walks in $\Tcal $ are called \emph{threads}.
  For all $1 \le i \le j \le \ell$, we define $\Wcal(i, j)$ as the walk $W_i \cdot T_i \cdot W_{i+1} \cdots T_{j-1} \cdot W_j$, which starts at the first vertex of $W_i$ and ends at the last vertex of $W_j$; in particular, $\Wcal(i, i) = W_i$.
  The \emph{concatenation} of \((\Wcal, \Tcal)\) is the walk given by $\Wcal(1, \ell)$.
  A threaded walkage is called \emph{untangled} if for each \(i \in [\ell - 1] \) the thread $T_i$ only intersects the rest of $\Wcal $ in $W_i$ or $W_{i+1}$.
  The \emph{size} of the threaded walkage \((\Wcal, \Tcal)\) is the number $\ell$ of walks in $\Wcal$.
\end{definition}

Threaded walkages allow us to connect the walks of a walkage $\Wcal$ using walks from another walkage $\Tcal$.
Untangled threaded walkages ensure that we can connect these walks in a reasonably disjoint way, which is important for maintaining a low congestion when building the bramble.
See \cref{fig:untangled-threaded-walkage} for an example of threaded walkages.

\begin{figure}[htbp]
  \centering
  \begin{tikzpicture}[scale = 0.5]
    \draw[red, walkage, decoration = {coil, segment length = 4.1mm, amplitude = 1.1mm}] (0,3.94) -- (4,0);

    \draw[red, walkage] (4.0,4) .. controls (3.2,5.2) and (-2.0,5.2) .. (-1.4,2.6) .. controls (-1.4,1.3) and (1.8,-1.2) .. (8.0,0);

    \draw[walkage, decoration = {segment length = 4mm}] (0,0) -- (0,4);
    \draw[walkage, decoration = {segment length = 4mm}] (4,0) -- (4,4);
    \draw[walkage, decoration = {segment length = 4mm}] (8,0) -- (8,4);

    \node[below = 0.2cm] at (0, 0) {$W_1$};
    \node[below = 0.2cm] at (4, 0) {$W_2$};
    \node[below = 0.2cm] at (8, 0) {$W_3$};

    \node[above right, red] at (2,2) {$T_1$};
    \node[left = 0.18cm, red] at (-1.4,2.6) {$T_2$};

    \begin{scope}[scale = 2, xshift = 6.5cm]

      \draw[red, walkage, decoration = {segment length = 2.4mm, post length = 1.5mm}] (0,2) arc (90:240:0.4) -- (2.2, 0.7464) arc (60:-90:0.4);
      \draw[red, walkage, decoration = {segment length = 2.4mm, post length = 1.5mm}] (2,2) arc (90:240:0.4) -- (4.2, 0.7464) arc (60:-90:0.4);
      
      \draw[walkage] (0,0) -- (0,2);
      \draw[walkage] (2,0) -- (2, 2);
      \draw[walkage] (4,0) -- (4, 2);

      \node[below] at (0,0) {$W'_1$};
      \node[below] at (2,0) {$W'_2$};
      \node[below] at (4,0) {$W'_3$};

      \node[below left, red] at (1,1) {$T'_1$};
      \node[above right, red] at (3,1) {$T'_2$};
    \end{scope}
  \end{tikzpicture}

  \caption{Example of two threaded walkages $(\Wcal, \Tcal)$ and $(\Wcal', \Tcal')$ of size $3$. The latter is untangled, while the former is not.}\label{fig:untangled-threaded-walkage}
\end{figure}
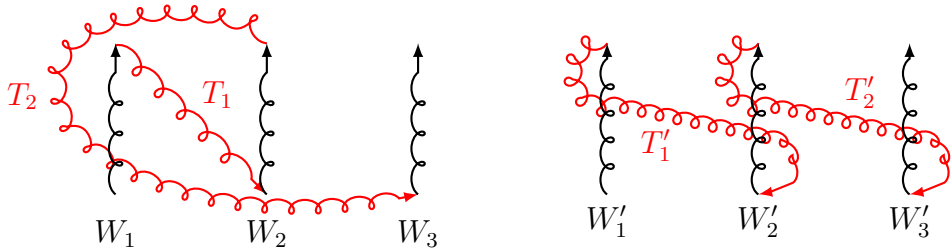

\cref{def:threaded-walkage} is somewhat different to~\cite[Definition~3.4]{MasarikPilipczukRzazewskiSorge22}.
We avoid arguing about the \emph{overlap}\footnote{The overlap of a family of walks is the maximum number of times a vertex is repeated in all of them.} of threaded walkages by ensuring the threads are vertex-disjoint.
In~\cite[Lemma~3.5]{MasarikPilipczukRzazewskiSorge22}, on the other hand, threaded linkages are built in such a way that allows threads to intersect slightly (cf.~\cref{lem:structured-walkage-congestion}).
These are then combined to create closed walks, and their congestion is bounded using their overlap.
Our modification costs us a little bit of extra directed treewidth but ensures the final congestion never exceeds $2$.

In the process of finding a bramble of low congestion, we must be able to go from an untangled threaded walkage to some object that allows us to build strong subgraphs.
This object is a \emph{cyclic walkage}.

\begin{definition}[Cyclic walkages]\label{def:cyclic-walkage}
  Let \(\Wcal = \{ W_1, W_2, \dots, W_\ell \} \) be a walkage in a digraph $D$.
  A \emph{cyclic walkage with congestion $c$ on} \(\Wcal \) is a pair \((\Wcal, \Ccal)\), where $\Ccal = \{C_1, C_2, \dots, C_\ell \}$ is a walkage (possibly with congestion greater than $1$) in $D$, for each $i \in [\ell]$ we have \(\stLink(C_i) = \enLink(W_i)\) and \(\enLink(C_i) = \stLink(W_i)\), and $\{ W_i \cdot C_i \mid i \in [\ell] \}$ is a closed walkage with congestion $c$.
  The \emph{size} of the cyclic walkage \((\Wcal, \Ccal)\) is the number $\ell$ of walks in $\Wcal$.
\end{definition}

The rough notion of a cyclic walkage appears naturally in the main proof of~\cite{MasarikPilipczukRzazewskiSorge22}.
The authors deal with these objects simply by viewing them as combinations of two threaded walkages.
By taking advantage of their cyclic structure instead, we obtain a tighter control over the congestion.
\cref{fig:cyclic-walkage} shows an example of cyclic walkage.

\begin{figure}[htbp]
  \centering
  \begin{tikzpicture}
    \draw[walkage, color = red, decoration = {coil, aspect = 1, amplitude = 1mm, segment length = 1.8mm, post length = 1.5mm}] (0,1.97) arc (90:-90:0.4925) arc (90:270:0.4925);
    \draw[walkage] (0,0) -- (0,2);
    \node[right, red] at (0.6,1.4) {$C_1$};

    \begin{scope}[xshift = 2cm]
      \draw[walkage, color = red, decoration = {aspect = 1, amplitude = 1mm, segment length = 2.5mm, post length = 1mm}] (0,1.97) arc (90:-90:0.985);
      \draw[walkage] (0,0) -- (0,2);
      \node[right, red] at (1.1,1) {$C_2$};
    \end{scope}

    \begin{scope}[xshift = 4.6cm]
      \draw[walkage, color = red, decoration = {aspect = 1, amplitude = 1mm, segment length = 1.8mm, post length = 1.5mm}] (0,1.97) arc (90:270:0.4925) arc (90:-90:0.4925);
      \draw[walkage] (0,0) -- (0,2);
      \node[right, red] at (0.6,0.6) {$C_3$};
    \end{scope}

    \node[below] at (0,0) {$W_1$};
    \node[below] at (2.2,0) {$W_2$};
    \node[below] at (4.4,0) {$W_3$};

  \end{tikzpicture}
  \caption{Example of a cyclic walkage $(\Wcal, \Ccal)$ of size $3$.}\label{fig:cyclic-walkage}
\end{figure}
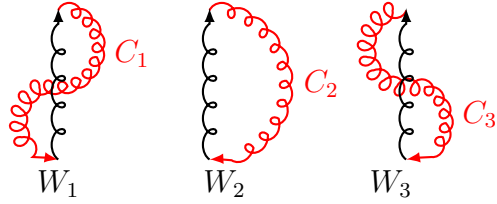

Starting from a walkage, we build untangled or cyclic walkages by finding an object that connects its walks appropriately.
We call such an object an \emph{extension}.

\begin{definition} [Walkage extensions]\label{def:walkage-extensions}
  Let \(\Wcal = \{W_1, W_2, \dots, W_\ell \} \) be a walkage in a digraph $D$.
  An \emph{extension} of \(\Wcal \) is a walkage $\Scal = \{S_1, S_2, \dots, S_s\}$, for $s \le \ell$ in $D$, such that $\stLink(S_i) = \enLink(W_i)$ for all $i \in [s]$.
  We say that \(\Scal \) is a \emph{threaded}, \emph{untangled} or \emph{$c$-cyclic} extension if \((\Wcal, \Scal)\) is a threaded walkage, untangled threaded walkage or a cyclic walkage with congestion at most $c$, respectively.
  We abuse the notation and say that a walkage \(\Wcal \) \emph{is threaded}, \emph{untangled} or \emph{$c$-cyclic} if it has an extension which is threaded, untangled or $c$-cyclic, respectively.
  Hence, $c$-cyclic stands for cyclic with congestion at most $c$, and we may simply say \emph{cyclic} instead of $1$-cyclic.
  Furthermore, we say that $\Wcal$ is \emph{structured with congestion $c$} if it is either $c$-cyclic or untangled, and \emph{structured} if it is structured with congestion $1$.
\end{definition}

\cref{fig:untangled-threaded-walkage,fig:cyclic-walkage} illustrate examples of walkage extensions in red for the walkages in black.
We show how to find extensions for walkages $\Wcal$ for which the set $\stLink(\Wcal) \cup \enLink(\Wcal)$ is well-linked in \cref{lem:structured-walkage-congestion}.
Before that, we establish some basic results.

\begin{lemma}\label{lem:untangled-walkage-to-walkage}
  Let $D$ be a digraph containing an untangled threaded walkage \((\Wcal, \Tcal)\), where \(\Wcal = \{W_1\), \(W_2\), \(\dots \), \(W_\ell \} \) and \(\Tcal = \{T_1, T_2, \dots, T_{\ell - 1}\} \).
  For any \(1 \le a_1 \le b_1 < a_2 \le b_2 < \cdots < a_{k-1} \le b_{k-1} < a_{k} \le b_k \le \ell \), we have that \( \{\Wcal(a_{j}, b_j) \mid j \in [k] \} \) is a walkage in $D$.
\end{lemma}

\begin{proof}
  First note that any vertex appears at most once in \(\Wcal \) and at most once in \(\Tcal \).
  Therefore, if \( \{\Wcal(a_{j}, b_j) \mid j \in [k]\} \) is not a walkage, then there is a vertex \(v\) and distinct indices \(i, j\) and indices \(r, s\) such that \(a_{i} \le r \le b_i\), \(a_{j} \le s \le b_j\) and \(v\in V(W_{r}) \cap V(T_{s})\).
  This is a contradiction to the fact that \((\Wcal, \Tcal)\) is untangled.
\end{proof}

\cref{lem:untangled-walkage-to-walkage} essentially tells us that the structure of an untangled threaded walkage is disjoint enough for us to define new walks by connecting a small set of consecutive walks in $\Wcal$ using walks in $\Tcal$ and still obtain a walkage in the end, regardless of how these small sets of walks in $\Wcal$ are picked.
\cref{fig:untangled-threaded-walkage-to-walkage} shows an example of this.

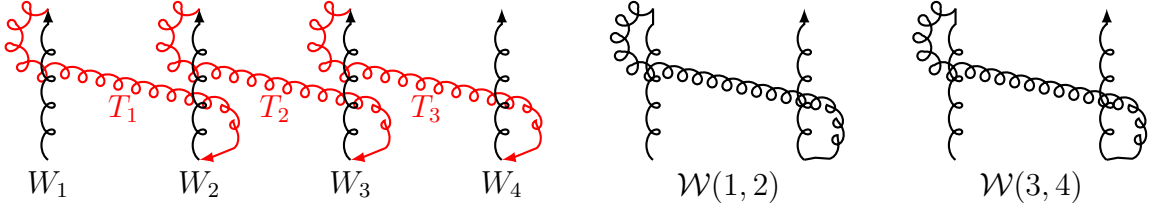
\begin{figure}[htbp]
  \centering
  \begin{tikzpicture}

    \draw[red, walkage, decoration = {segment length = 2.4mm, post length = 1.5mm}] (0,2) arc (90:240:0.4) -- (2.2, 0.7464) arc (60:-90:0.4);
    \draw[red, walkage, decoration = {segment length = 2.4mm, post length = 1.5mm}] (2,2) arc (90:240:0.4) -- (4.2, 0.7464) arc (60:-90:0.4);
    \draw[red, walkage, decoration = {segment length = 2.4mm, post length = 1.5mm}] (4,2) arc (90:240:0.4) -- (6.2, 0.7464) arc (60:-90:0.4);
    
    \draw[walkage] (0,0) -- (0,2);
    \draw[walkage] (2,0) -- (2, 2);
    \draw[walkage] (4,0) -- (4, 2);
    \draw[walkage] (6,0) -- (6, 2);

    \node[below] at (0,0) {$W_1$};
    \node[below] at (2,0) {$W_2$};
    \node[below] at (4,0) {$W_3$};
    \node[below] at (6,0) {$W_4$};

    \node[below, red] at (1,1) {$T_1$};
    \node[below, red] at (3,1) {$T_2$};
    \node[below, red] at (5,1) {$T_3$};

    \begin{scope}[xshift = 8cm]

      \node[below] at (1,0) {$\Wcal(1, 2)$};
      \node[below] at (5,0) {$\Wcal(3,4)$};

      \draw[thick, decorate, decoration = {coil, post length = 1.5mm, aspect = 1.2}] (0,0) -- (0,2);
      \draw[thick, decorate, decoration = {coil, segment length = 2mm, post length = 1.5mm, aspect = 1.2}] (0,2) arc (90:240:0.4) -- (2.2, 0.7464) arc (60:-90:0.4);
      \draw[walkage, decorate, decoration = {coil, post length = 1.5mm, aspect = 1.2}] (2,0) -- (2,2);

      \draw[thick, decorate, decoration = {coil, post length = 1.5mm, aspect = 1.2}] (4,0) -- (4,2);
      \draw[thick, decorate, decoration = {coil, segment length = 2mm, post length = 1.5mm, aspect = 1.2}] (4,2) arc (90:240:0.4) -- (6.2, 0.7464) arc (60:-90:0.4);
      \draw[walkage, decorate, decoration = {coil, post length = 1.5mm, aspect = 1.2}] (6,0) -- (6,2);

    \end{scope}
  \end{tikzpicture}

  \caption{Extracting a walkage from segments of an untangled threaded walkage via \cref{lem:untangled-walkage-to-walkage}.}\label{fig:untangled-threaded-walkage-to-walkage}
\end{figure}

The next result is a direct adaptation of~\cite[Lemma~3.6]{MasarikPilipczukRzazewskiSorge22} from threaded linkages to threaded walkages.
For completeness, we provide a proof that follows closely the original one.
The core idea is that walkages that are threaded can be polished into walkages that are either untangled or cyclic, although possibly with some congestion.
This is useful because both these structures give us more control over the structure of the walks.

\begin{lemma}[{Adapted from~\cite[Lemma~3.6]{MasarikPilipczukRzazewskiSorge22}}]\label{lem:untangled-walkage}
  Given $s, u \in \N$, let $\wUW(s, u) = s(u + 1) - 1 = \Ocal(s  u)$.
  Let $D$ be a digraph containing a walkage \(\Wcal \) of size \(w\) that is threaded.
  If $w  \ge \wUW(s,u)$, then \(D\) contains either
  \begin{enumerate}[label=\textnormal{(\alph*)}]
    \item\label{item:untangled-walkage-1} A subwalkage of \(\Wcal \) of size $s$ that is $2$-cyclic; or
    \item\label{item:untangled-walkage-2} A subwalkage of \(\Wcal \) of size $u$ that is untangled.
  \end{enumerate}
\end{lemma}

\begin{proof}
  Let \((\Wcal, \Tcal )\) be a threaded walkage of size $w \ge \wUW(s,u)$ on \(\Wcal = \{W_1, W_2, \dots, W_w\} \).
  We call a walk $\Wcal(i,j)$ a \emph{useful walk} if it contains a closed subwalk that contains at least one walk from $\Wcal$ entirely, as in \cref{fig:useful-walk}.
  We say that $\Wcal(i,j)$ \emph{contains a useful walk} if $\Wcal(i',j')$ is a useful walk for some $i \le i' \le j' \le j$.
  Note that $\Wcal(i,j)$ may be a useful walk only when $i < j$.
  Furthermore, if $\Wcal(1, w)$ contains no useful walk, then $(\Wcal, \Tcal)$ is untangled.
  Indeed, if $(\Wcal, \Tcal)$ is not untangled then there is a vertex $v \in V(T_i) \cap V(W_r)$ for some $i \in [w - 1]$ and $r \in [w] \setminus \{i, i+1\}$.
  If $r < i$, then the subwalk of $\Wcal(r, i+1)$ from an occurrence of $v$ in $W_r$ to an occurrence of $v$ in $T_i$ is a closed walk containing $W_{r+1}$, and thus $\Wcal(r, i+1)$ is a useful walk.
  Similarly, if $r > i + 1$, then the subwalk of $\Wcal(i, r)$ from an occurrence of $v$ in $T_i$ to an occurrence of $v$ in $W_r$ is a closed walk containing $W_{i+1}$, and thus $\Wcal(i, r)$ is a useful walk.

  \begin{figure}[htbp]
    \centering
    \begin{tikzpicture}[scale = 0.5]
      \begin{scope}[line cap = round]
        \draw[decoration = {coil, aspect = 1.2, post length = 2.1mm, segment length = 4mm}, decorate, line width = 4pt, red!30] (4,0) -- (4,4);
        \begin{scope}
          \clip (0.2,0) -- (-0.245,3) -- (5,3) -- (5,5.5) -- (-2,5.5) -- (-2,0);
          \draw[decoration = {coil, aspect = 1.2}, decorate, line width = 4pt, red!30] (4.0,4)
          .. controls (3.2,5.2) and (-2.0,5.2) ..
          (-1.4,2.6)
          .. controls (-1.4,1.3) and (1.8,-1.2) ..
          (8.0,0);
        \end{scope}
        \begin{scope}[line cap = round]
          \clip (-2,0.8) -- (0.12,0.8) -- (0.035, 0.98) -- (3,3) -- (0,4.5) -- cycle;
          \draw[decoration = {coil, aspect = 1.2, post length = 2.1mm, segment length = 4mm}, decorate, line width = 4pt, red!30] (0,0) -- (0,4);
        \end{scope}
        \draw[decoration = {coil, aspect = 1.2}, decorate, line width = 4pt, red!30] (0,4) -- (4,0);
      \end{scope}

      \draw[red, walkage]
      (0,4) -- (4,0);

      \draw[red, walkage]
      (4.0,4)
      .. controls (3.2,5.2) and (-2.0,5.2) ..
      (-1.4,2.6)
      .. controls (-1.4,1.3) and (1.8,-1.2) ..
      (8.0,0);

      \begin{scope}[xshift = 6cm, xscale = 0.5]
        \draw[red, walkage]
        (4,3.94)
        .. controls (3.2,5.2) and (-2.0,5.2) ..
        (-1.4,2.6)
        .. controls (-1.4,1.3) and (1.8,0) ..
        (12.0,0);
      \end{scope}

      \draw[walkage, decoration = {segment length = 4mm}] (0,0) -- (0,4);
      \draw[walkage, decoration = {segment length = 4mm}] (4,0) -- (4,4);
      \draw[walkage, decoration = {segment length = 4mm}] (8,0) -- (8,4);
      \draw[walkage, decoration = {segment length = 4mm}] (12,0) -- (12,4);

      \node[below] at (0, 0) {$W_1$};
      \node[below = 0.2cm] at (4, 0) {$W_2$};
      \node[below] at (8, 0) {$W_3$};
      \node[below] at (12, 0) {$W_4$};
    \end{tikzpicture}

    \caption{$\Wcal(1,3)$ is a useful walk, as it contains a closed subwalk (in light red) that contains $W_2$ entirely, but $\Wcal(3, 4)$ does not contain a useful walk, as none of its closed subwalks contains a walk in $\Wcal$ entirely.}\label{fig:useful-walk}
  \end{figure}
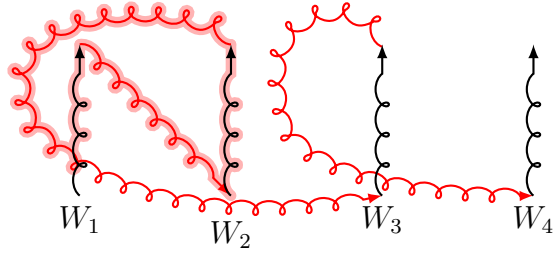

  If $\Wcal(1, w)$ does not contain any useful walks, then $\Wcal$ itself is untangled, as it contains an untangled extension.
  Otherwise, we greedily construct a sequence $I_1, I_2, \dots, I_\ell$ of useful walks, for some $\ell \ge 1$, as follows.
  Let $I_1 = \Wcal(i_1, j_1)$ be a useful walk such that $j_1$ is smallest possible.
  Given $I_{m} = \Wcal(i_{m}, j_{m})$, let $I_{m+1} = \Wcal(i_{m+1}, j_{m+1})$ be a useful walk such that $i_{m+1} > j_{m}$ and \(j_{m+1}\) is smallest possible (if at least one such useful walk exists).

  Note that the choices of $j_m$ and $i_{m+1}$ guarantee that no walk in \(\Wcal \cup \Tcal \) is used in more than one of the constructed useful walks.
  Since \(\Wcal \) and \(\Tcal \) are walkages, this implies that the sequence of useful walks has congestion $2$.
  If the number $\ell$ of constructed useful walks is at least $s$, then we get the first desired outcome by selecting one closed walk that contains a walk in $\Wcal$ from each useful walk.
  The congestion for these closed walks must be $2$ as any vertex in $D$ can be in at most one walk in $\Wcal$ and at most one walk in $\Tcal$.

  Hence, we may assume that $\ell < s$.
  Define $\ell + 1$ walks $W'_0, W'_1, \dots, W'_\ell$ as follows.
  Let $W'_0 = \Wcal(1, j_1 - 1)$, $W'_\ell = \Wcal(j_{\ell} + 1, w)$ (which is possibly empty if $j_{\ell} = w$), and $W'_m = \Wcal(j_m + 1, j_{m+1} - 1)$, which is never empty as $j_m < i_{m+1} < j_{m+1}$, for all $m \in [\ell - 1]$.
  Note that $W'_m$ does not contain a useful walk, for if it did, then $j_{m+1}$ would not be the smallest possible choice for $I_{m+1}$ (or, if $m = \ell$, the useful walk $I_{\ell+1}$ would exist).

  All walks in $\Wcal$ appear in exactly one of these walks, except for $W_{j_1}, W_{j_2}, \dots, W_{j_\ell}$, so the union of $W'_0, W'_1, \dots, W'_\ell$ covers at least $w - \ell \ge w  - s + 1$ walks in $\Wcal$.
  Hence, there exists some $k \in \{0, 1, \dots, \ell\}$ such that $W'_k$ contains at least
  \[
    \frac{w - s + 1}{\ell + 1} \ge \frac{s \cdot u}{s} = u
  \]
  walks from $\Wcal$.
  Write $W'_k = \Wcal(j_k + 1, j_{k+1} - 1)$, where we set $j_0 = 0$ and $j_{\ell+1} = w + 1$ (it may be the case that $W'_k = W'_0$ or $W'_k = W'_\ell$).
  It follows that the subwalkage $\{W_{j_k+1}, W_{j_k + 2}, \dots, W_{j_{k+1} - 1}\} \subseteq \Wcal$ is untangled as $W'_k$ contains no useful walk.
\end{proof}

\cref{lem:closed-walkage} below essentially tells us that finding a walkage that is cyclic but is slightly congested, as in the outcome~\ref{item:untangled-walkage-1} of \cref{lem:untangled-walkage}, is fine because if the congestion cannot be easily removed, then we actually find a large bramble of low congestion.

\begin{lemma}\label{lem:closed-walkage}
  Given $k, s \in \N$, let \(\wCW(k,s) = s \cdot \dCM(k) = \Ocal(k  s\sqrt{\log k})\).
  Let $D$ be a digraph containing a walkage \(\Wcal \) of size $w$ that is $c$-cyclic.
  If $w \ge \wCW(k, s)$, then $D$ contains either
  \begin{enumerate}[label=\textnormal{(\alph*)}]
    \item\label{item:closed-walkage-1} A bramble of size $k$ and congestion at most $c$; or
    \item\label{item:closed-walkage-2} A subwalkage of \(\Wcal \) of size $s$ that is cyclic.
  \end{enumerate}
\end{lemma}

\begin{proof}
  Let $\Zcal$ be the family of closed walks with congestion at most $c$ given by the $c$-cyclic walkage on $\Wcal$.
  If \(\degen(\intersec(\Zcal)) \ge \dCM(k)\), then $D$ contains a bramble of size $k$ and congestion at most $c$ by \cref{prop:dense-bramble}.
  Hence, we may assume that \(\degen(\intersec(\Zcal)) \le \dCM(k) - 1\).
  Iteratively choose a vertex of minimum degree in \(\intersec(\Zcal)\) and remove it and all of its neighbors from \(\intersec(\Zcal)\) to find a stable set $S$ in \(\intersec(\Zcal)\).
  By the degeneracy of \(\intersec(\Zcal)\), every step removes at most \(\dCM(k)\) of its vertices, so \(S\) has size at least \(s\).
  Since \(S\) is a stable set, its associated closed walks are pairwise vertex-disjoint, and thus the walks of $\Wcal$ they contain define a subwalkage of $\Wcal$ in \(D\) that is cyclic, whose closed walks are walks of $\Zcal$.
\end{proof}

The next three results show how to combine walkages that have appropriate extensions into walkages that are cyclic, though possibly a little bit congested.
\cref{lem:untangled-cyclic-extension-congestion} is a direct adaptation of the Bowtie Lemma~\cite[Lemma~3.7]{MasarikPilipczukRzazewskiSorge22} to the context of walkages.
For completeness, we present a proof that loosely follows the original one.

\begin{lemma}[{Adapted from~\cite[Lemma~3.7]{MasarikPilipczukRzazewskiSorge22}}]\label{lem:untangled-cyclic-extension-congestion}
  Given $s \in \N$, let $\dCEC(s) = \dP(s, 5) = \Ocal(s)$.
  Let $D$ be a digraph that contains two walkages $\Wcal$ and $\Scal$ that are untangled.
  If $\degen(\intersec(\Wcal \cup \Scal)) \ge \dCEC(s)$, then \(D\) contains a subwalkage of \(\Wcal \) and a subwalkage of $\Scal$ of size $s$ that are $2$-cyclic.
  Moreover, in this case, $D$ also contains a family of $s$ closed walks with congestion $2$ such that each closed walk contains a distinct walk from $\Wcal$ and a distinct walk from $\Scal$.
\end{lemma}

\begin{proof}
  Let \(\Wcal'\) and $\Scal'$ be untangled extensions of \(\Wcal \) and \(\Scal \), respectively, and let $W$ and $S$ be the concatenations of $(\Wcal, \Wcal')$ and $(\Scal, \Scal')$, respectively.
  Let $G$ be a subgraph of $\intersec(\Wcal \cup \Scal)$ with minimum degree at least $\dP(s, 5)$.
  Consider the walks in \(\Wcal \) and $\Scal$ ordered by their appearance in \(W \) and $S$, respectively.
  From \cref{prop:partition} in $G$ using the orders provided for the walks in $\Wcal$ and \(\Scal \), we get a collection of sets $X_1, \ldots, X_{s}$ and $Y_1, \ldots, Y_{s}$ satisfying properties~\ref{item:partition-prop-1},~\ref{item:partition-prop-2} and~\ref{item:partition-prop-3} of the proposition.

  As $G[X_i \cup Y_i]$ has average degree at least $5$, let $G_i$ be a subgraph of $G[X_i \cup Y_i]$ with minimum degree at least $3$, which exists by \cref{prop:average-min-degree}.
  Let \(S_3^i\) be the last walk from \(Y_i\) in $G_i$ considering the order of \(\Scal \) and \(W_1^i, W_2^i, W_3^i\) be $3$ neighbors of \(S_3^i\) in \(X_i\) which have this order in \(\Wcal \).
  Let \(S_1^i\) and $S_2^i$ be two neighbors of \(W_3^i\) in \(Y_i\) distinct from \(S_3^i\) which appear in this order in $\Scal$.
  Note that there is a walk from any vertex in \(V(W_1^i)\) to any vertex in \(V(W_3^i)\) entirely contained in $W$ and any such walk contains \(W_2^i\); similarly, there is a walk from any vertex in \(V(S_1^i)\) to any vertex in \(V(S_3^i)\) entirely contained in $S$ and any such walk contains $S_2^i$.
  As $W_1^i$ and $S_3^i$ are neighbors in $G_i$ and so are $W_3^i$ and $S_1^i$, let $u_i \in V(W_1^i) \cap V(S_3^i)$ and $v_i \in V(W_3^i) \cap V(S_1^i)$.
  Let \(C_{W}^i\) be the subwalk of $W$ from \(u_i\) to \(v_i\) and \(C_{S}^i\) be the subwalk of $S$ from \(v_i\) to \(u_i\).
  Define now $C_i$ as the concatenation of $C_W^i$ and $C_S^i$, as in \cref{fig:closed-walk-ci}.

  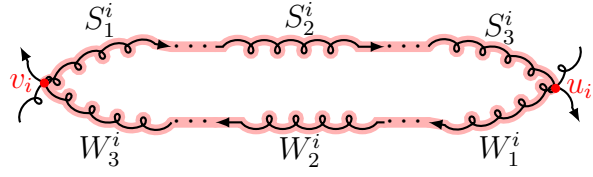
\begin{figure}[htbp]
    \centering
    \begin{tikzpicture}
      \path[use as bounding box] (-1,-1.2) rectangle (8.4,2.2);
      \begin{scope}[line cap = round]
        \clip (3.7, 4) circle (4.88);
        \draw[decorate, decoration = {coil, aspect = 1.2}, red!30, line width = 4pt] (0,0) to[bend left] (2,1);
        \draw[red!30, line width = 4pt] (2,1) -- (2.58196601125,1);
        \draw[decorate, decoration = {coil, aspect = 1.2, pre length = 1mm, post length = 0.4mm}, red!30, line width = 4pt] (2.58196601125,1) to (4.81803398875,1);
        \draw[red!30, line width = 4pt] (4.81803398875,1) -- (5.4,1);
        \draw[decorate, decoration = {coil, aspect = 1.2}, red!30, line width = 4pt] (5.4,1) to[bend left] (7.4,0);
      \end{scope}

      \begin{scope}[line cap = round]
        \clip (3.7, -3) circle (4.88);
        \draw[decorate, decoration = {coil, aspect = 1.2}, red!30, line width = 4pt] (2,0) to[bend left] (0,1);
        \draw[red!30, line width = 4pt] (2,0) -- (2.58196601125,0);
        \draw[decorate, decoration = {coil, aspect = 1.2, pre length = 1mm, post length = 0.4mm}, red!30, line width = 4pt] (4.81803398875,0) to (2.58196601125,0);
        \draw[red!30, line width = 4pt] (4.81803398875,0) -- (5.4,0);
        \draw[decorate, decoration = {coil, aspect = 1.2}, red!30, line width = 4pt] (7.4,1) to[bend left] (5.4,0);
      \end{scope}

      \draw[walkage] (0,0) to[bend left] node[above = 0.1cm, pos = 0.6] {$S_1^i$} (2,1);
      \draw[walkage, decoration = {pre length = 1mm, post length = 0.4mm}] (2.58196601125,1) to node[above, pos = 0.5] {$S_2^i$} (4.81803398875,1);
      \draw[walkage] (5.4,1) to[bend left] node[above, pos = 0.4] {$S_3^i$} (7.4,0);
      \node at (2.290983005625, 1) {$\cdots$};
      \node at (5.109016994375, 1) {$\cdots$};

      \draw[walkage] (2,0) to[bend left]  node[below = 0.1cm, pos = 0.4] {$W_3^i$} (0,1);
      \draw[walkage, decoration = {pre length = 1mm, post length = 0.4mm}] (4.81803398875,0) to node[below, pos = 0.5] {$W_2^i$} (2.58196601125,0);
      \draw[walkage] (7.4,1) to[bend left] node[below = 0.1cm, pos = 0.6] {$W_1^i$} (5.4,0);
      \node at (2.291983005625,0) {$\cdots$};
      \node at (5.110016994375,0) {$\cdots$};

      \node[draw, red, circle, fill = red, inner sep = 0, minimum size = 3pt] (v) at (0.32, 0.52) {};
      \node[draw, red, circle, fill = red, inner sep = 0, minimum size = 3pt] (u) at (7.08, 0.45) {};
      \node[left, red] at (v) {$v_i$};
      \node[right, red] at (u) {$u_i$};
    \end{tikzpicture}
    \caption{Sketch of the construction of the closed walk $C_i$ (in light red).}\label{fig:closed-walk-ci}
  \end{figure}

  Since \(C_{W}^i\) starts and ends in \(\Wcal \) and $C_{S}^i$ starts and ends in $\Scal$, \( \{C_{W}^1, C_{W}^2, \ldots,\allowbreak C_{W}^{s}\} \) and $\{C_{S}^1, C_{S}^2, \dots,\allowbreak C_{S}^s\}$ are walkages by \cref{lem:untangled-walkage-to-walkage}.
  Therefore, \( \Ccal = \{ C_1, \ldots, C_{s}\} \) is a collection of closed walks of size \(s\) and congestion $2$.
  As each walk from $\Ccal$ contains a walk from $\Wcal$ and a walk from $\Scal$, there is a subwalkage of $\Wcal$ and a subwalkage of $\Scal$ both of which are $2$-cyclic and of size $s$.
\end{proof}

\cref{lem:cyclic-extension-congestion} extends the core idea of \cref{lem:untangled-cyclic-extension-congestion} to work also with walkages that are cyclic, which provides a better treatment of cyclic walkages to avoid exceeding congestion $2$.
The proof is very similar to that of \cref{lem:untangled-cyclic-extension-congestion}.

\begin{lemma}\label{lem:cyclic-extension-congestion}
  Let $D$ be a digraph that contains a walkage $\Wcal$ that is untangled and a walkage $\Scal$ that is $c$-cyclic for some $c \in \N$.
  Given $s \in \N$, if $\degen(\intersec(\Wcal \cup \Scal)) \ge \dCEC(s)$, then \(D\) contains a subwalkage of \(\Wcal \) and a subwalkage of $\Scal$ of size $s$ that are $(c+1)$-cyclic.
  Moreover, in this case, $D$ also contains a family of $s$ closed walks with congestion at most $c+1$ such that each closed walk contains a distinct walk from $\Wcal$ and a distinct walk from $\Scal$.
\end{lemma}

\begin{proof}
  Let \(\Wcal'\) and $\Scal'$ be the untangled and the $c$-cyclic extensions of \(\Wcal \) and $\Scal$, respectively, and \(W\) be the concatenation of $(\Wcal, \Wcal')$.
  Let $G$ be a subgraph of $\intersec(\Wcal \cup \Scal)$ with minimum degree at least $\dP(s, 5)$.
  Consider the walks in \(\Wcal \) ordered by their appearance in \(W \) and the walks in \(\Scal \) ordered arbitrarily.
  From \cref{prop:partition} in $G$ using the orders provided for the walks in $\Wcal$ and \(\Scal \), we get a collection of sets $X_1, \ldots, X_{s}$ and $Y_1, \ldots, Y_{s}$ satisfying properties~\ref{item:partition-prop-1},~\ref{item:partition-prop-2} and~\ref{item:partition-prop-3} of the proposition.

  Let $G_i$ be a subgraph of $G[X_i \cup Y_i]$ with minimum degree at least $3$, which exists by \cref{prop:average-min-degree}.
  Let \(S^i\) be the first walk from \(Y_i\) in $G_i$ considering the order of \(\Scal \) and \(W_1^i, W_2^i, W_3^i\) be $3$ neighbors of \(S^i\) in \(X_i\) which have this order in \(\Wcal \).
  Note that there is a walk from any vertex in \(V(W_1^i)\) to any vertex in \(V(W_3^i)\) entirely contained in $W$ and any such walk contains \(W_2^i\).
  As $W_1^i$ and $W_3^i$ are neighbors of $S^i$ in $G_i$, let $u_i \in V(W_1^i) \cap V(S^i)$ and $v_i \in V(W_3^i) \cap V(S^i)$.
  Let \(C_i\) be the concatenation of the subwalk of $W$ from $u_i$ to $v_i$ and a walk from $v_i$ to $u_i$ in the closed walk of the $c$-cyclic walkage $(\Scal, \Scal')$ containing $S^i$ that passes through all vertices in this closed walk, as in \cref{fig:closed-walk-ci-2}.

  \begin{figure}[htbp]
    \centering
    \begin{tikzpicture}[scale = 1.5]
      \begin{scope}[line cap = round]
        \draw[red!30, line width = 4pt] (0.7,1.15) -- (0.91,0.86);
        \draw[red!30, line width = 4pt] (2.2,1.15) -- (2.41,0.86);
      \end{scope}
      \begin{scope}[line cap = round]
        \draw[red!30, line width = 4pt, decoration = {coil, aspect = 1.2, post length = 2mm}, decorate] (1.5, 0) -- (1.5, 2);
        \begin{scope}
          \clip (0, 1.966667) circle (0.63333);
          \draw[red!30, line width = 4pt, decoration = {coil, aspect = 1.2, post length = 2mm}, decorate] (0, 0) -- (0, 2);
        \end{scope}
        \begin{scope}
          \clip (2.5,-0.5) -- (3.5,-0.5) -- (3.5,1.25) to[bend right] (2.5, 1.25) -- cycle;
          \draw[red!30, line width = 4pt, decoration = {coil, aspect = 1.2, post length = 2mm}, decorate] (3, 0) -- (3, 2);
        \end{scope}
        \draw[red!30, line width = 4pt, decoration = {coil, aspect = 1.2, post length = 2mm}, decorate] (-1, 1.333333) -- (4,1.333333);
        \draw[red!30, line width = 4pt, decoration = {coil, aspect = 1.2, post length = 2mm}, decorate] (4,1.333333) to[bend left] (-1, 1.333333);
        \begin{scope}
          \clip (-1, 1.12) -- (4,1.12) -- (4,2.5) -- (-1,2.5) -- cycle;
          \draw[red!30, line width = 4pt, decoration = {coil, aspect = 1.2, post length = 2mm}, decorate] (0,2) -- (1.5,0);
          \draw[red!30, line width = 4pt, decoration = {coil, aspect = 1.2, post length = 2mm}, decorate] (1.5,2) -- (3,0);
          \draw[walkage, red, decoration = {post length = 2mm}] (1.5,2) -- (3,0);
          \draw[walkage, red, decoration = {post length = 2mm}] (0,2) -- (1.5,0);
        \end{scope}
        \begin{scope}
          \clip (0.885,0.913) -- (1.7,0.913) -- (1.7,-0.5) -- (0.885,-0.5) -- cycle;
          \draw[red!30, line width = 4pt, decoration = {coil, aspect = 1.2, post length = 2mm}, decorate] (0,2) -- (1.5,0);
          \draw[walkage, red, decoration = {post length = 2mm}] (0,2) -- (1.5,0);
        \end{scope}
        \begin{scope}
          \clip (2.385,0.913) -- (3.2,0.913) -- (3.2,-0.5) -- (2.385,-0.5) -- cycle;
          \draw[red!30, line width = 4pt, decoration = {coil, aspect = 1.2, post length = 2mm}, decorate] (1.5,2) -- (3,0);
          \draw[walkage, red, decoration = {post length = 2mm}] (1.5,2) -- (3,0);
        \end{scope}
      \end{scope}
      \node[rotate = -57] at (0.815,0.98) {$\cdots$};
      \node[rotate = -57] at (2.315,0.98) {$\cdots$};

      \draw[walkage, red, decoration = {post length = 2mm}] (4,1.333333) to[bend left] (-1,1.333333);
      \draw[walkage, decoration = {post length = 1mm}] (-1,1.333333) -- (4,1.333333);

      \draw[walkage, decoration = {post length = 2mm}] (0,0) -- (0,2);
      \draw[walkage, decoration = {post length = 2mm}] (1.5,0) -- (1.5,2);
      \draw[walkage, decoration = {post length = 2mm}] (3,0) -- (3,2);

      \node[below] at (0,0) {$W_1^i$};
      \node[below] at (1.5,0) {$W_2^i$};
      \node[below] at (3,0) {$W_3^i$};
      \node[right] at (4,1.333333) {$S^i$};

      \node[draw, red, circle, fill = red, inner sep = 0, minimum size = 3pt] (v) at (-0.05, 1.35) {};
      \node[draw, red, circle, fill = red, inner sep = 0, minimum size = 3pt] (u) at (2.96, 1.37) {};
      \node[above right, red] at (u) {$v_i$};
      \node[above left, red] at (v) {$u_i$};
    \end{tikzpicture}
    \caption{Sketch of the construction of the closed walk $C_i$ (in light red).}\label{fig:closed-walk-ci-2}
  \end{figure}
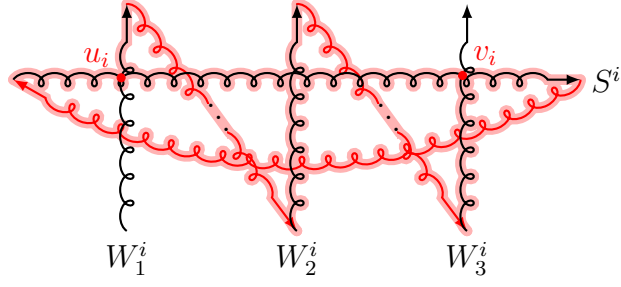

  Since each vertex in $D$ can be in at most $c$ parts from the partition of $Y$ and in at most one part from the partition of $X$, the congestion of $\Ccal = \{C_1, C_2, \dots, C_s\}$ is at most $c + 1$.
  Therefore, \(\Ccal \) is a collection of closed walks of size \(s\) and congestion at most $c + 1$.
  As each walk from $\Ccal$ contains a walk from $\Wcal$ and a walk from $\Scal$, there is a subwalkage of $\Wcal$ and a subwalkage of $\Scal$ both of which are $(c+1)$-cyclic and of size $s$.
\end{proof}

The next corollary simply wraps up \cref{lem:cyclic-extension-congestion,lem:untangled-cyclic-extension-congestion}.

\begin{corollary}\label{cor:cyclic-extension}
  Given $k, s \in \N$, let $\dCE(k, s) = \dCEC(\wCW(k, s)) = \Ocal(k  s\sqrt{\log k})$.
  Let $D$ be a digraph that contains a walkage $\Wcal$ that is untangled and a walkage $\Scal$ that is structured.
  If $\degen(\intersec(\Wcal \cup \Scal)) \ge \dCE(k, s)$, then \(D\) contains either
  \begin{enumerate}[label=\textnormal{(\alph*)}]
    \item\label{item:cyclic-extension-1} A bramble of size $k$ and congestion $2$; or
    \item\label{item:cyclic-extension-2} A subwalkage of \(\Wcal \) and a subwalkage of $\Scal$ of size $s$ that are cyclic, and a closed walkage of size $s$ such that each closed walk contains a distinct walk from $\Wcal$ and a distinct walk from $\Scal$.
  \end{enumerate}
\end{corollary}

\begin{proof}
  If $\Scal$ is untangled, then we get a subwalkage $\Wcal' \subseteq \Wcal$ and a subwalkage $\Scal' \subseteq \Scal$ of size $\wCW(k,s)$ that are $2$-cyclic by \cref{lem:untangled-cyclic-extension-congestion}.
  Similarly, if $\Scal$ is cyclic, then we get a subwalkage $\Wcal' \subseteq \Wcal$ and a subwalkage $\Scal' \subseteq \Scal$ of size $\wCW(k, s)$ that are $2$-cyclic by \cref{lem:cyclic-extension-congestion}.
  Moreover, in either case we get a family $\Zcal$ of $\wCW(k,s)$ closed walks with congestion $2$ such that each closed walk contains a distinct walk from $\Wcal$ and a distinct walk from $\Scal$.
  By \cref{lem:closed-walkage}, we either get a bramble of size $k$ and congestion $2$ (and we are done) or a subwalkage $\Wcal'' \subseteq \Wcal'$ of size $s$ that is cyclic.
  If the outcome of \cref{lem:closed-walkage} is~\ref{item:closed-walkage-2}, note that each walk $W''_i \in \Wcal''$ is contained in some closed walk $Z''_i \in \Zcal$, which also contains a distinct walk $S''_i \in \Scal'$.
  Hence, we have $s$ distinct walks from $\Scal$ that are vertex-disjoint, for if not then $\Wcal''$ would not be cyclic (it would be $c$-cyclic, for some $c \ge 2$), which give us a subwalkage of $\Scal$ of size $s$ that is cyclic.
  Furthermore, the closed walks $Z''_i$ obtained from each $W''_i \in \Wcal''$ give us a closed walkage of size $s$ such that each closed walk contains a distinct walk from $\Wcal$ and a distinct walk from $\Scal$.
\end{proof}

\cref{lem:structured-walkage-congestion} and \cref{cor:structured-walkage} are our main tools for building walkages that are either cyclic or untangled.
The former shows how to go from an appropriate well-linked walkage to a walkage that is either untangled or cyclic, possibly with some congestion, and the latter simply gets rid of the congestion.

\begin{lemma}\label{lem:structured-walkage-congestion}
  Given $s, u \in \N$, let $\wSWC(s, u) = (s - 1)(\wUW(s, u) - 1) +  1 = \Ocal(s^2  u)$.
  Let $D$ be a digraph containing a walkage \(\Wcal \) of size \(w\) such that \(\enLink(\Wcal) \cup \stLink(\Wcal)\) is well-linked.
  If \(w \ge \wSWC(s, u)\), then $D$ contains either
  \begin{enumerate}[label=\textnormal{(\alph*)}]
    \item A subwalkage of $\Wcal$ of size \(s\) that is $2$-cyclic; or
    \item A subwalkage of $\Wcal$ of size $u$ that is untangled.
  \end{enumerate}
\end{lemma}

\begin{proof}
  By well-linkedness, let $\Tcal$ be an $(\enLink(\Wcal), \stLink(\Wcal))$-walkage of size $w$.
  Let \(C\) be the auxiliary digraph with vertex set \(\stLink(\Wcal)\cup \enLink(\Wcal)\) and an arc $(\stLink(Q), \enLink(Q))$ for each walk $Q \in \Wcal \cup \Tcal$.
  Note that, even if walks from $\Wcal \cup \Tcal$ intersect in $D$, this is not reflected in $C$.
  Indeed, $C$ is a bipartite digraph with bipartition $(\stLink(\Wcal), \enLink(\Wcal))$ in which $\Wcal$ and $\Tcal$ are each represented by a perfect matching, oriented in opposite directions of the bipartition.
  Since every vertex in \(C\) has in-degree and out-degree exactly $1$, \(C\) decomposes into a collection of disjoint cycles.
  Define a \emph{cyclic walk} to be a walk in \(D\) obtained by the concatenation of the walks corresponding to the arcs of a cycle in \(C\).
  Note that the set of all cyclic walks has congestion $2$, as these walks are obtained from a partition of the walkages \(\Wcal \) and \(\Tcal \).

  Since $w \ge \wSWC(s, u)$, in $C$ we either have a strong component with at least \( \wUW(s, u) \) arcs corresponding to walks in $\Wcal$, or at least $s$ strong components.
  If we have a strong component in \(C\) with at least \(\wUW(s, u) \) arcs corresponding to walks in $\Wcal$, then the walks associated with these arcs induce a subwalkage \(\Wcal' \subseteq \Wcal \) of size at least \(\wUW(s, u)\) that is threaded.
  By \cref{lem:untangled-walkage}, we either get a subwalkage of $\Wcal'$ of size $s$ that is $2$-cyclic or a subwalkage of $\Wcal'$ of size \(u\) that is untangled.

  Otherwise, we have at least $s$ strong components and each component gives a distinct walk in $\Wcal$ and a distinct walk in $\Tcal$ such that their concatenation is a closed walk.
  Hence, we get a subwalkage of $\Wcal$ of size $s$ that is $2$-cyclic, as each vertex in $D$ can be in at most one walk from $\Wcal$ and one walk from $\Tcal$.
\end{proof}

\cref{lem:structured-walkage-congestion} should be compared directly to an application of~\cite[Lemma~3.5]{MasarikPilipczukRzazewskiSorge22} followed by an application of~\cite[Lemma~3.6]{MasarikPilipczukRzazewskiSorge22}, which essentially finds either an untangled threaded linkage of size $u$ or a family of $s$ closed walks with congestion $3$ starting from an appropriate linkage of size linear in $s  u$.
The main advantage of our approach is that the congestion never exceeds $2$, although this comes at a small extra cost in the size of the input walkage.

\begin{corollary}\label{cor:structured-walkage}
  Given $k, s, u \in \N$, let $\wSW(k, s, u) = \wSWC (\wCW(k, s), u)  = \Ocal(k^2  s^2  u \cdot \log k )$.
  Let $D$ be a digraph containing a walkage \(\Wcal \) of size \(w\) such that \(\enLink(\Wcal) \cup \stLink(\Wcal)\) is well-linked.
  If \(w \ge \wSW(k, s, u)\), then $D$ contains either
  \begin{enumerate}[label=\textnormal{(\alph*)}]
    \item\label{item:structured-walkage-1} A bramble of size $k$ and congestion $2$;
    \item\label{item:structured-walkage-2} A subwalkage of $\Wcal$ of size \(s\) that is cyclic; or
    \item\label{item:structured-walkage-3} A subwalkage of $\Wcal$ of size $u$ that is untangled.
  \end{enumerate}
\end{corollary}

\begin{proof}
  Using \cref{lem:structured-walkage-congestion} on $\Wcal$, we either get a subwalkage of $\Wcal$ of size $\wCW(k, s)$ that is $2$-cyclic or a subwalkage of $\Wcal$ of size $u$ that is untangled.
  If the subwalkage is $2$-cyclic, we apply \cref{lem:closed-walkage} to either get a bramble of size $k$ and congestion $2$ or a subwalkage of size $s$ that is cyclic.
\end{proof}

We are now ready to prove the main result of this section, which is our first step towards building the bramble with congestion $2$.
Its main objective is to show how to go from a well-linked walkage that is untangled and arrive at one that is cyclic and is still reasonably well-linked.
A walkage that is cyclic is much more useful because each of its elements is a connected subgraph of our digraph, meaning we can use it directly to build the bags of the bramble.

\begin{lemma}\label{lem:cyclic-walk-system}
  Given $k, \ell, w \in \N$, let
  \[
    \begin{array}{@{}r@{}l@{}l@{}}
      s                 & {} = \lISM(\dCM(k), w) & {} = \Ocal(wk \sqrt{\log k}), \\
      u                 & {} = \lISM(\dCE(k, s), w) & {} = \Ocal(w^2k^2\log k) \\
      s'                & {} = s + (2\dCE(k, 2w - 1) - 3)(2w - 2) & {} = \Ocal(w^2k \sqrt{\log k}) \\
      u'                & {} = u + (2\dCE(k, 2w - 1) - 3)(2w - 2) & {} = \Ocal(w^2k^2 \log k) \\
      \lCWS(k, \ell, w) & {} = \ell + 2\wSW(k, s', u') & {} = \Ocal(\ell + w^6 k^6 \cdot \log^3 k), \\[5pt]
      \wCWS(k, w)       & {} = u'(2w - 2) + 1 & {} = \Ocal(w^3 k^2 \cdot \log k).
    \end{array}
  \]
  Let $D$ be a digraph that contains an $\ell'$-linked walkage $\Wcal$ of size $w'$ that is untangled.
  If $\ell' \ge \lCWS(k, \ell, w)$ and $w' \ge \wCWS(k, w)$, then $D$ contains either
  \begin{enumerate}[label=\textnormal{(\alph*)}]
    \item\label{item:cyclic-walk-system-1} A bramble of size $k$ and congestion $2$; or
    \item\label{item:cyclic-walk-system-2} An $\ell$-linked walkage of size $w$ that is cyclic.
  \end{enumerate}
\end{lemma}

\begin{proof}
  Let $\Wcal = \{W_1, W_2, \dots, W_{w'}\}$.
  By well-linkedness, for each $i \in [w']$ there is a walkage $\Qcal_i$ of size $\wSW(k, s', u')$ from the last $\wSW(k, s', u')$ vertices of $A_i$ to the first $\wSW(k, s', \allowbreak u')$ vertices of $A_i$, where $A_i \subseteq V(W_i)$ is the subset of size $\lCWS(k, \ell, w)$ of a well-linked set.
  By \cref{cor:structured-walkage}, each $\Qcal_i$ gives us a bramble of size $k$ and congestion $2$ (and we are done) or a subwalkage $\Qcal'_i$ that is either cyclic and of size $s'$ or untangled and of size $u'$.

  If the outcome of \cref{cor:structured-walkage} is~\ref{item:structured-walkage-2} or~\ref{item:structured-walkage-3}, we first analyze the intersection of the walkages $\Qcal'_i$ and $\Wcal$ for all $i \in [w']$.
  If $\degen(\intersec(\Wcal \cup \Qcal'_i)) \ge \dCE(k, 2w - 1)$ for some $i \in [w']$, by \cref{cor:cyclic-extension} we get a bramble of size $k$ and congestion $2$ (and we are done) or an $\ell'$-linked subwalkage of $\Wcal$ of size $2w - 1 \ge w$ that is cyclic (and we are done, as $\ell' \ge \ell$).
  Otherwise, for each $i \in [w']$, define $\Tcal_i$ as the set of walks in $\Wcal \setminus \{W_i\}$ that intersect at least $2(\dCE(k, 2w - 1) - 1)$ walks in $\Qcal'_i$.
  Note that $|\Tcal_i| \le |\Qcal'_i| - 1$ for each $i \in [w']$, otherwise we could pick $\Tcal'_i \subseteq \Tcal_i$ of size exactly $|\Qcal'_i|$ and the average degree of $\intersec(\Tcal'_i \cup \Qcal'_i)$ would be at least
  \[
    \frac{2 \cdot 2(\dCE(k, 2w - 1) - 1)|\Qcal'_i|}{2|\Qcal'_i|} = 2(\dCE(k, 2w - 1) - 1),
  \]
  which implies $\degen(\intersec(\Tcal'_i \cup \Qcal'_i)) \ge \dCE(k, 2w - 1)$ by \cref{prop:average-min-degree}, a contradiction to the fact that $\degen(\intersec(\Wcal \cup \Qcal'_i)) < \dCE(k, 2w - 1)$, as $\Tcal_i \subseteq \Wcal$.

  Define a subwalkage $\Wcal' \subseteq \Wcal$ of size $2w - 1$ as follows.
  Choose a walk $W_i \in \Wcal$ arbitrarily and remove itself together with all the walks in $\Tcal_i$ from $\Wcal$.
  Note that at most $|\Qcal'_i|$ walks were removed from $\Wcal$.
  Let $W'_1 = W_i$, add it to $\Wcal'$, and repeat.
  As $w' \ge |\Qcal'_i|(2w - 2) + 1$, we can iterate this process $2w - 1$ times.
  By construction, each walk in $\Wcal' \setminus \{W'_i\}$ intersects at most $2(\dCE(k, 2w - 1) - 1) - 1$ walks in $\Qcal'_i$ for all $i \in [2w - 1]$, so we can remove $2\dCE(k, 2w - 1) - 3$ walks from $\Qcal'_i$ for each walk in $\Wcal' \setminus \{W'_i\}$ to obtain a subset $\Qcal''_i \subseteq \Qcal'_i$ of size $|\Qcal'_i| - (2\dCE(k, 2w - 1) - 3)(2w - 2)$ such that the walks in $\Qcal''_i$ only intersect $W'_i$.
  This means that if $\Qcal'_i$ has size $s'$ and is cyclic, then $\Qcal''_i$ has size at least $s$ and is cyclic, while if $\Qcal'_i$ has size $u'$ and is untangled, then $\Qcal''_i$ has size at least $u$ and is untangled.

  Now, for each distinct $i, j \in [2w - 1]$, if $\Qcal''_i$ and $\Qcal''_j$ are both untangled and $\degen(\intersec(\Qcal''_i \cup \Qcal''_j)) \ge \dCE(k, s)$, then by \cref{cor:cyclic-extension} we get a bramble of size $k$ and congestion $2$ (and we are done) or subwalkages $\Qcal^*_i \subseteq \Qcal''_i$ and $\Qcal^*_j \subseteq \Qcal''_j$ of size $s$ that are cyclic.
  If the outcome of \cref{cor:cyclic-extension} is~\ref{item:cyclic-extension-2} for all distinct $i, j \in [2w - 1]$, there must exist $Z \subseteq [2w - 1]$ of size at least $w$ such that either $\Qcal''_i$ contains a subwalkage $\Qcal^*_i$ of size $s$ that is cyclic for all $i \in Z$ or $\Qcal''_i$ is untangled and of size $u$ for all $i \in Z$, and $\degen(\intersec(\Qcal''_i \cup \Qcal''_j)) < \dCE(k, s)$ for all $i, j \in Z$.

  We now analyze the intersection between the walkages $\Qcal^*_i$ for each $i \in Z$ to either obtain the desired bramble or disjoint walks which will be used to build our $\ell$-linked walkage that is cyclic.
  Assume first that $\Qcal''_i$ has a subwalkage $\Qcal^*_i$ that is cyclic and of size $s$ for all $i \in Z$.
  If $\degen(\intersec(\Qcal^*_i \cup \Qcal^*_j)) \ge \dCM(k)$ for some $i, j \in Z$, then we get a bramble of size $k$ and congestion at most $2$ by \cref{prop:dense-bramble}.
  Otherwise, as $s \ge \lISM(\dCM(k), w)$, we can apply \cref{prop:independent-set-multipartite} to obtain a walk $Q_i$ from each $\Qcal^*_i$ such that no two walks intersect.
  Otherwise, we may assume $\Qcal''_i$ is untangled and of size $u$ for all $i \in Z$, and $\degen(\intersec(\Qcal''_i \cup \Qcal''_j)) < \dCE(k, s)$ for all $i, j \in Z$.
  Again by \cref{prop:independent-set-multipartite}, as $u \ge \lISM(\dCE(k, s), w)$, we obtain a walk $Q_i$ from each $\Qcal''_i$ such that no two walks intersect.

  Let $W''_i$ be the concatenation of $Q_i$ and the middle part of $W'_i$ that starts at $\enLink(Q_i)$ and ends at $\stLink(Q_i)$.
  Note that $W''_i$ is a closed walk, and it contains at least $\ell$ vertices from the original well-linked set, as we may remove at most $\wSW(k, s', u')$ vertices at the start of $W'_i$ and at most $\wSW(k, s', u')$ at the end.
  Hence, $\Wcal'' = \{W''_1, W''_2, \dots, W''_w\}$ is an $\ell$-linked walkage of size $w$ that is cyclic.
\end{proof}

See \cref{fig:proof-overview-of-linkage-cyclic-walk-lemma} for an overview of the proof of \cref{lem:cyclic-walk-system}.
We remark that the walkage in outcome~\ref{item:cyclic-walk-system-2} of \cref{lem:cyclic-walk-system} need not be a subwalkage of $\Wcal$, unlike in the previous results.
We allow this for simplicity, as we do not need this fact in later proofs.
However, if one wishes to use \cref{lem:cyclic-walk-system} and ensure that the output walkage is a subwalkage of the input one, simply create a new walkage $\Wcal^* = \{W^*_1, W^*_2, \dots, W^*_{w'/3}\}$ such that $W^*_i = W_{3i-2} \cdot W_{3i-1} \cdot W_{3i}$ and apply the lemma to $\Wcal^*$.

\section{Strong systems}\label{sec:walk-systems}

Path systems are a fundamental object in many structural proofs involving directed treewidth~\cite{KK22, MasarikPilipczukRzazewskiSorge22, HKMM24}.
They consist essentially of a linkage in which the union of the first few and last few vertices of each path is well-linked.

In undirected graphs, \textcite{RW12} use a structure that closely resembles a path system to find a large grid-like minor of order $\ell$ (which is also a bramble with congestion $2$) in a graph with treewidth $\Omega(\ell^4 \sqrt{\log \ell})$.

In~\cite{MasarikPilipczukRzazewskiSorge22}, the authors use a similar strategy to find a bramble of size $k$ and congestion $8$ in a digraph with directed treewidth $\Omega(k^{48} \cdot \log^{13} k)$.
For digraphs, however, using path systems as a building block for brambles is a little harder, as an (undirected) path is a connected graph but a directed path is not a strong digraph.
Thus, the first step is now to find out how to use the directed paths to build the strong subgraphs that compose the bags of the bramble, which constitutes the bulk of the main proof in~\cite{MasarikPilipczukRzazewskiSorge22}.

We simplify these arguments using walkages.
We achieve this by defining a new object called \emph{strong system}, which is a structure analogous to the one used in~\cite{RW12}, that allows us to translate the ideas that previously only worked for undirected graphs.

\begin{definition}[Strong systems]\label{def:strong-system}
  An \emph{$(\ell,w)$-strong system} $\Scal$ in a digraph $D$ consists of a pair $(\Lcal, \Wcal)$, where $\Wcal = \{W_1, W_2, \dots, W_w\}$ is a closed walkage of size $w$ and $\Lcal = \{\Lcal_{ij} \mid \{i, j\} \in \binom{[w]}{2}\}$ is a collection of $\binom{w}{2}$ closed walkages, each of size $\ell$, such that for each $W_i, W_j \in \Wcal$ every closed walk in $\Lcal_{ij}$ intersects both $W_i$ and $W_j$.
  We define the \emph{pair set} of $\Scal$ to be $\pairs(\Scal) = \binom{[w]}{2}$.
\end{definition}

\cref{fig:strong-system} shows an example of a strong system.
Strong systems are much simpler to work with than path systems because we avoid the notion of order within the walks, as they are closed, and this essentially puts us back in an undirected setting.

\begin{figure}[htbp]
  \centering
  \begin{tikzpicture}[scale=.75]
    \draw[thick, red] (0,3.7475952641916446) ellipse (50pt and 10pt);
    \draw[thick, red, -latex] (-0.1,4.0975952641916446) -- (0.1,4.0975952641916446);
    \node[below] at (0,-0.4) {$W_1$};
    \node[above, red] at (0,4.1) {$\Lcal_{23}$};

    \begin{scope}[xshift = 1.875cm, yshift = 3.2475952641916446cm, rotate = -60]
      \draw[thick, red] (0,-3.7475952641916446) ellipse (50pt and 10pt);
      \draw[thick, red, -latex] (-0.1,-4.0975952641916446) -- (0.1,-4.0975952641916446);
    \end{scope}

    \begin{scope}[xshift = -1.875cm, yshift = 3.2475952641916446cm, rotate = 60]
      \draw[thick, red] (0,-3.7475952641916446) ellipse (50pt and 10pt);
      \draw[thick, red, -latex] (-0.1,-4.0975952641916446) -- (0.1,-4.0975952641916446);
      \draw[thick] (0,0) ellipse (24pt and 12pt);
      \draw[thick, -latex] (-0.1,12pt) -- (0.1, 12pt);
      \node[above left] at (0,0.4) {$W_2$};
      \node[below right, red] at (0,-4) {$\Lcal_{13}$};
    \end{scope}

    \begin{scope}[xshift = 1.875cm, yshift = 3.2475952641916446cm, rotate = -60]
      \draw[thick] (0,0) ellipse (24pt and 12pt);
      \draw[thick, -latex] (-0.1,12pt) -- (0.1, 12pt);
      \node[above right] at (0,0.35) {$W_3$};
      \node[below left, red] at (0,-3.9) {$\Lcal_{12}$};
    \end{scope}

    \draw[thick] (0,0) ellipse (24pt and 12pt);
    \draw[thick, -latex] (-0.1,12pt) -- (0.1, 12pt);
  \end{tikzpicture}

  \caption{Example of a $(1,3)$-strong system $\Scal = (\Lcal, \Wcal)$, the walkage $\Wcal$ is in black and all walkages in $\Lcal$ are in red.}\label{fig:strong-system}
\end{figure}
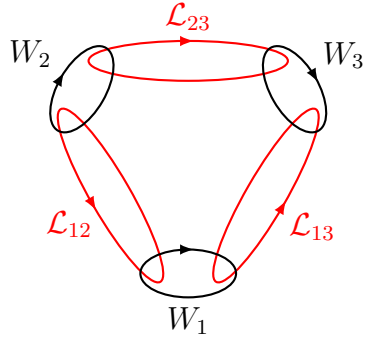

For higher congestion, however, we avoid using strong systems as their construction requires the directed treewidth to be at least a polynomial with a large degree, whereas \textcite[Lemma~4.6]{KK22} find $(\ell, p)$-path systems on digraphs with directed treewidth only quadratic in $\ell  p$.
As a result, in \cref{sec:higher-congestion-with-masarik} we simply redefine path systems in terms of walks to avoid the blow up in the exponent given by our next result.

\begin{lemma}\label{lem:fully-cyclic-walk-system}
  Given $k, \ell \in \N$, let $\lFCWS(k, \ell) = \wSW(k, \ell, \lCWS(k, 0, \ell) \cdot \wCWS(k, \ell)) = \Ocal(k^{10}  \ell^{11} \cdot \log^5 k)$.
  Let $D$ be a digraph containing an $\ell'$-linked walkage $\Wcal$ of size $w$ that is cyclic.
  If $\ell' \ge \lFCWS(k, \ell)$, then $D$ contains either
  \begin{enumerate}[label=\textnormal{(\alph*)}]
    \item\label{item:fully-cyclic-walk-system-1} A bramble of size $k$ and congestion $2$; or
    \item\label{item:fully-cyclic-walk-system-2} An $(\ell, w)$-strong system.
  \end{enumerate}
\end{lemma}

\begin{proof}
  Let $Z = \binom{[w]}{2}$ and let $\Wcal'$ be the closed walkage obtained by $\Wcal$ and its cyclic extension.
  By well-linkedness, there is a walkage $\Lcal_{ij}$ of size $\ell'$ from $W_i$ to $W_j$, for each $\{i, j\} \in Z$.
  As $\ell' \ge \wSW(k, \ell, \lCWS(k, 0, \ell) \cdot \wCWS(k, \ell))$, we can apply \cref{cor:structured-walkage} to $\Lcal_{ij}$ for each $\{i, j\} \in Z$ to get a bramble of size $k$ and congestion $2$ or a subwalkage $\Lcal'_{ij} \subseteq \Lcal_{ij}$ that is either cyclic and of size $\ell$ or untangled and of size $\lCWS(k, 0, \ell) \cdot \wCWS(k, \ell)$.
  Let $Z' \subseteq Z$ be the set of pairs $\{i, j\}$ for which $\Lcal'_{ij}$ has size $\lCWS(k, 0, \ell) \cdot \wCWS(k, \ell)$ and is untangled.

  Let $L_{ij}$ be the concatenation of the untangled threaded walkage on $\Lcal'_{ij}$ for each $\{i, j\} \in Z'$.
  By \cref{lem:untangled-walkage-to-walkage}, we can define a new walkage
  \[
    \Lcal''_{ij} = \{ L_{ij}((m-1) \lCWS(k, 0, \ell) + 1, m \cdot \lCWS(k, 0, \ell)) \mid m \in [\wCWS(k, \ell)] \},
  \]
  that has size $\wCWS(k, \ell)$ and is $\lCWS(k, 0, \ell)$-linked, as each walk in $\Lcal_{ij}'$ contains $2$ elements from a common well-linked set.
  Then, for each $\{i, j\} \in Z'$ we can apply \cref{lem:cyclic-walk-system} to get a bramble of size $k$ and congestion $2$ (and we are done) or a subwalkage $\Lcal^*_{ij} \subseteq \Lcal''_{ij}$ of size $\ell$ that is cyclic.

  Note that all walks in $\Lcal^*_{ij}$ must contain at least one vertex from $W_i$ and at least one vertex from $W_j$, as the first and last vertices, say $u_{ij}$ and $v_{ij}$, of any walk $L^*_{ij} \in \Lcal^*_{ij}$ were part of the well-linked set associated with $\Lcal''_{ij}$, which was originated by vertices in $W_i$ and $W_j$, and if $u_{ij}$ and $v_{ij}$ both belong to $W_i$, then there is some vertex $z_{ij}$ in $W_j$ that comes after $u_{ij}$ and before $v_{ij}$ in $L^*_{ij}$, as the construction of $\Lcal''_{ij}$ picks vertices in $W_i$ and $W_j$ alternately.

  Let $\tilde{\Lcal}_{ij}$ be the closed walkage obtained either by $\Lcal^*_{ij}$ if $\Lcal'_{ij}$ is untangled or by $\Lcal'_{ij}$ itself if it is already cyclic, and its cyclic extension for each $\{i, j\} \in Z$.
  Note that each closed walk in $\tilde{\Lcal}_{ij}$ contains a vertex in $W_i$ and a vertex in $W_j$ by construction.
  Define now $\tilde{\Lcal}$ as the collection of closed walkages between each $\{i, j\} \in Z$.
  It follows that $(\tilde{\Lcal}, \Wcal')$ is an $(\ell, w)$-strong system.
\end{proof}

\cref{lem:fully-cyclic-walk-system} is a prime example of the utility of walkages.
The walkages $\Lcal''_{ij}$ are built by the concatenation of smaller walks, which is not trivial to achieve in a linkage setting because we cannot afford to discard vertices to transform walks back to paths, as this could throw away vertices from the well-linked set.

We can now proceed to find the desired bramble inside a strong system.
The next lemma essentially adapts the arguments of~\cite{RW12}, replacing grid-like minors in graphs by brambles with congestion $2$ in digraphs.

\begin{lemma}[{Adapted from~\cite[Theorem~1.2]{RW12}}]\label{lem:final-bramble}
  Given $k \in \N$, let $\lB(k) = \lISM(k^2-k, \dCM(k)) = \Ocal(k^3 \sqrt{\log k})$.
  Let $D$ be a digraph containing an $(\ell, k)$-strong system $\Scal = (\Lcal, \Wcal)$.
  If $\ell \ge \lB(k)$, then $D$ contains a bramble of size $k$ and congestion $2$.
\end{lemma}

\begin{proof}
  If $\degen(\intersec(\Lcal_{ij} \cup \Lcal_{i'j'})) \ge \dCM(k)$ for some $\{i,j\}, \{i',j'\} \in \pairs(\Scal)$, then $D$ contains a bramble of size $k$ and congestion $2$ by \cref{prop:dense-bramble}.
  Otherwise, as $\ell \ge \lISM(k^2-k, \dCM(k))$, by \cref{prop:independent-set-multipartite} we can find a family $\Zcal$ of disjoint closed walks, one for each $\Lcal_{ij}$.
  Note that $\intersec(\Wcal \cup \Zcal)$ contains a $1$-subdivision of the complete graph $K_k$, because each pair of vertices associated with walks in $\Wcal$ have a distinct common neighbor associated with a walk in $\Zcal$, as in \cref{fig:clique-subdivision}.
  Moreover, a $1$-subdivision of $K_k$ contains $K_k$ as a minor.

  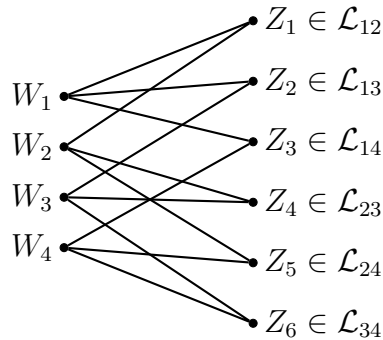
\begin{figure}[htbp]
    \centering
    \begin{tikzpicture}
      \node[draw, fill = black, circle, inner sep = 0, minimum size = 3pt] (00) at (0,2) {};
      \node[draw, fill = black, circle, inner sep = 0, minimum size = 3pt] (01) at (0,1.33333) {};
      \node[draw, fill = black, circle, inner sep = 0, minimum size = 3pt] (02) at (0,0.66667) {};
      \node[draw, fill = black, circle, inner sep = 0, minimum size = 3pt] (03) at (0,0) {};

      \node[draw, fill = black, circle, inner sep = 0, minimum size = 3pt] (10) at (2.5,3) {};
      \node[draw, fill = black, circle, inner sep = 0, minimum size = 3pt] (11) at (2.5,2.2) {};
      \node[draw, fill = black, circle, inner sep = 0, minimum size = 3pt] (12) at (2.5,1.4) {};
      \node[draw, fill = black, circle, inner sep = 0, minimum size = 3pt] (13) at (2.5,0.6) {};
      \node[draw, fill = black, circle, inner sep = 0, minimum size = 3pt] (14) at (2.5,-0.2) {};
      \node[draw, fill = black, circle, inner sep = 0, minimum size = 3pt] (15) at (2.5,-1) {};

      \draw[thick] (00) -- (10) -- (01);
      \draw[thick] (00) -- (11) -- (02);
      \draw[thick] (00) -- (12) -- (03);
      \draw[thick] (01) -- (13) -- (02);
      \draw[thick] (01) -- (14) -- (03);
      \draw[thick] (02) -- (15) -- (03);

      \node[left] at (00) {$W_1$};
      \node[left] at (01) {$W_2$};
      \node[left] at (02) {$W_3$};
      \node[left] at (03) {$W_4$};

      \node[right] at (10) {$Z_1 \in \Lcal_{12}$};
      \node[right] at (11) {$Z_2 \in \Lcal_{13}$};
      \node[right] at (12) {$Z_3 \in \Lcal_{14}$};
      \node[right] at (13) {$Z_4 \in \Lcal_{23}$};
      \node[right] at (14) {$Z_5 \in \Lcal_{24}$};
      \node[right] at (15) {$Z_6 \in \Lcal_{34}$};
    \end{tikzpicture}

    \caption{The closed walkage $\Zcal$ consists of one closed walk of $\Lcal_{ij}$ for each $\{i, j\} \in \binom{[k]}{2}$.}\label{fig:clique-subdivision}
  \end{figure}

  Now, for each vertex $v$ in this minor, define a closed walk that corresponds to the concatenation of all closed walks that were contracted into $v$.
  Any two of these closed walks intersect and the congestion of the family of all closed walks is at most $2$, as each vertex of $D$ can be in at most one walk in $\Wcal$ and at most one walk in $\Zcal$, so it defines a bramble of size at least $k$ and congestion $2$.
\end{proof}

\section{Congestion two and applications}\label{sec:main-proof}

This section contains our main result and some applications to the \textsc{Directed Disjoint Paths with Congestion $2$} problem and to half-integral packings of subgraphs with large directed treewidth.

\subsection{Finding the bramble}\label{sec:finding-the-bramble}

We first show how to build a large bramble with congestion $2$.
\cref{thm:main-proof} is simply a restatement of \cref{thm:main-proof-short} in more detail.
Its proof uses all lemmata provided so far to reach our goal; many intermediate steps either produce a large bramble with congestion $2$ or a more well-behaved object.
We begin with a sufficiently well-linked walkage that is cyclic obtained via an application of \cref{lem:cyclic-walk-system} on the output of \cref{prop:linked-path}, assuming that the digraph has large directed treewidth.
Assuming no bramble with congestion $2$ is found in intermediate steps, we go from the well-linked walkage to a strong system by \cref{lem:fully-cyclic-walk-system}, and then we find the desired bramble in it using \cref{lem:final-bramble}.

\begin{theorem}\label{thm:main-proof}
  Given $k \in \N$, let $\dB(k) = \wLP(\wCWS(k, k)\lCWS(k, \lFCWS(k, \lB(k)), k)) + 1 = \Ocal(k^{96} \cdot \log^{23} k)$.
  Any digraph $D$ with $\dtw(D) \ge \dB(k)$ contains a bramble of size $k$ and congestion $2$.
\end{theorem}

\begin{proof}
  Let $w = \wCWS(k, k) \cdot \lCWS(k, \lFCWS(k, \lB(k)), k)$.
  As $\dtw(D) \ge \wLP(w) + 1$, we get a path $P$ and a well-linked set $A \subseteq V(P)$ of size $w$ by \cref{prop:linked-path}.
  Order the vertices of $A$ according to their appearance in $P$.
  For each $i \in [\wCWS(k, k)]$, let $W_i$ be the walk that starts at the $((i-1)\lCWS(k, \lFCWS(k, \lB(k)), k) + 1)$-th vertex of $A$ in $P$ and ends at the $(i\lCWS(k, \lFCWS(k, \lB(k)), k))$-th vertex of $A$ in $P$.
  As $w \ge \wCWS(k, k) \cdot \lCWS(k, \lFCWS(k, \lB(k)), k)$ there are always $\wCWS(k, k)$ such walks.
  Let $\tilde{w} = \wCWS(k, k)$, $\tilde{\ell} = \lCWS(k, \lFCWS(k, \lB(k)), k)$ and define $\Wcal = \{W_1, W_2, \dots, \allowbreak W_{\tilde{w}}\}$.
  Note that $\Wcal$ is an $\tilde{\ell}$-linked walkage of size $\wCWS(k, k)$ that is untangled, as there is a walk (a path, in fact) from the last vertex of $W_i$ to the first vertex of $W_{i+1}$ for all $i \in [\tilde{w} - 1]$ by construction.

  We now apply \cref{lem:cyclic-walk-system} on $\Wcal$ to get a bramble of size $k$ and congestion $2$ (and we are done) or an $\lFCWS(k, \lB(k))$-linked walkage $\Wcal'$ of size $k$ that is cyclic.
  We can then apply \cref{lem:fully-cyclic-walk-system} on $\Wcal'$ to get a bramble of size $k$ and congestion $2$ (and we are done) or an $(\lB(k), k)$-strong system $\Scal$.
  Finally, \cref{lem:final-bramble} on $\Scal$ gives us a bramble of size $k$ and congestion $2$.
\end{proof}

We remark that with small modifications in the results used so far, we can show that a digraph with directed treewidth $\Omega(k^{90} \cdot \polylog k)$ or $\Omega(k^{78} \cdot \polylog k)$ contains a bramble of size $k$ and congestion $4$ or $6$, respectively.
We omit the technical details, as in the next section we show alternative proofs for these results using smaller directed treewidth by modifying the arguments in~\cite{MasarikPilipczukRzazewskiSorge22}.

The justification on how the construction of \cref{thm:main-proof} can be done in polynomial time, assuming the path $P$ and the well-linked set $A \subseteq V(P)$ originating from \cref{prop:linked-path} are given, is similar to the one done by \textcite{LS26} in order to produce an algorithmic version of the main proof in~\cite{MasarikPilipczukRzazewskiSorge22}.
In short, the bulk of both constructions consists in building intersection graphs, testing for degeneracy, and iteratively following walks from walkages (in our case) or paths/walks from a path system (in the case of~\cite{MasarikPilipczukRzazewskiSorge22}), picking pieces satisfying prescribed properties and moving to the next step of the construction.
\textcite{MatulaBeck83} showed that the degeneracy of a graph can be found in $\mathcal{O}(n+m)$ time, where $m$ is the number of edges in the graph, and the number of graphs for which this test needs to happen in our construction depends on the size of the $\ell$-linked walkage built in the first step of the proof of \cref{thm:main-proof}.

As is the case of~\cite[Lemma 3.7]{MasarikPilipczukRzazewskiSorge22}, the running time of our bowtie-like lemmas, namely~\cref{lem:untangled-cyclic-extension-congestion,lem:cyclic-extension-congestion}, depend on \cref{prop:partition}.
The proof of the latter consists of counting degrees of sequences of vertices of a bipartite graph to find subgraphs with a guaranteed minimum number of edges, and proceeding recursively on these subgraphs.
This process naturally yields a polynomial-time algorithm depending linearly on the number of parts one wishes to find.

\cref{prop:dense-bramble}, which is one of the winning cases of~\cref{sec:higher-congestion-with-masarik}, resorts to finding clique minors in graphs of large degeneracy, and the original statement uses \cref{prop:clique-minor-non-algo}.
As discussed in~\cite{LS26}, it can easily be made algorithmic by applying \cref{prop:clique-minor} instead, as stated in~\cite[Lemma 24]{LS26}.

As mentioned in \cref{sec:intro}, and as is the case in~\cite{LS26}, the main bottleneck on the running time is on the application of an algorithmic version of the Lovász Local Lemma~\cite{LE75}.
\textcite{LS26} elect to apply the version by \textcite{CGH2013} which, when compared to the version by \textcite{Harris2022}, has the advantage of giving a precise bound on the running time, but with a slightly worse dependence on the required parameters needed to apply the result.
The authors in~\cite{LS26} apply the result in~\cite{CGH2013} to prove their own version of~\cref{prop:independent-set-multipartite}, and then proceed through a careful analysis on how the higher requirements impact on the requested directed treewidth of the given digraph in order to guarantee the existence of a bramble of size $k$ and congestion 8.
For simplicity, we take the other alternative and directly apply the algorithmic version from~\cite{Harris2022} through~\cref{prop:independent-set-multipartite}.
Although this allows us to nullify any change to the statement of \cref{thm:main-proof} due to the Lovász Local Lemma, we remark that a precise running time on \cref{prop:independent-set-multipartite} is not known, and thus we can only claim that with \cref{thm:main-proof} one can build the bramble in polynomial time.

\subsection{Directed Disjoint Paths with Congestion}\label{sec:ddp-congestion}
In this section, we discuss the application of \cref{thm:main-proof} to the \textsc{Directed Disjoint Paths with Congestion $c$} problem.
Formally, in this problem we are given a digraph $D$, positive integers $k, c$, sets $S, T \subseteq V(D)$ with $S = \{s_i \mid i \in [k]\}$ and $T = \{t_i \mid i \in [k]\}$, and the goal is to decide if $D$ contains a collection $\{P_i \mid i \in [k]\}$ such that each $P_i$ is a path from $s_i$ to $t_i$ in $D$ and $\{V(P_i) \mid i \in [k]\}$ has congestion at most $c$.
We adopt the following notation by \textcite{LS26}.
\begin{definition}[$(k,c)$-linked digraphs]
  For integers $k,c$, we say that a digraph $D$ is \emph{$(k,c)$-linked} if, for every pair $S,T \subseteq V(D)$ with $S = \{s_i \mid i \in [k]\}$ and $T = \{t_i \mid i \in [k]\}$, there is a collection $\{P_i \mid i \in [k]\}$ such that each $P_i$ is a path from $s_i$ to $t_i$ in $D$ such that $\{V(P_i) \mid i \in [k]\}$ has congestion at most $c$.
\end{definition}
In~\cite{CamposSLS25}, the authors showed that the existence of a bramble of size roughly $2c k^2$ and congestion $2$ suffices to guarantee that a $k$-strong digraph is $(k,c)$-linked.

\begin{proposition}[{\cite[Theorem 4.2]{CamposSLS25}}]\label{prop:k-c-linked-digraphs}
  Given $k, c \in \N$, let $\fref{prop:k-c-linked-digraphs}(k,c) = 2k(c  k - c +2) + c(k-1) = \Ocal(c  k^2)$.
  If $D$ is a $k$-strong digraph containing a bramble of size at least $\fref{prop:k-c-linked-digraphs}(k,c)$ and congestion $c \geq 2$, then for every choice of $S$ and $T$ the instance $(D, k, c, S, T)$ of {\sc Directed Disjoint Paths with Congestion} $c$ is positive, and a solution can be found in time $\mathcal{O}(k^4 \cdot n^2)$.
\end{proposition}
\textcite{LS26} used \cref{prop:k-c-linked-digraphs} and an algorithmic version of the bramble construction in~\cite{MasarikPilipczukRzazewskiSorge22} to show that there is a polynomial function $f(k)$ such that every $k$-strong digraph with $\dtw(D) \geq f(k)$ is $(k,8)$-linked.
Applying \cref{thm:main-proof} together with \cref{prop:k-c-linked-digraphs}, we maintain a polynomial bound on the directed treewidth while reducing the congestion to 2.

\begin{theorem}\label{thm:k-2-linked-digraphs}
  Given $k \in \N$, let $\fref{thm:k-2-linked-digraphs}(k) = \dB(\fref{prop:k-c-linked-digraphs}(k,2)) = \Ocal(k^{192} \cdot \log^{23} k)$.
  Every $k$-strong digraph with $\dtw(D) \geq \fref{thm:k-2-linked-digraphs}(k)$ is $(k,2)$-linked and a solution can be found in polynomial time.
\end{theorem}
\begin{proof}
  Follows directly by applying \cref{thm:main-proof} to find a bramble of size $\fref{prop:k-c-linked-digraphs}(k,2)$ and congestion $2$, then applying \cref{prop:k-c-linked-digraphs}.
\end{proof}

\cref{thm:k-2-linked-digraphs} together with the \XP{} algorithm for {\sc Directed Disjoint Paths} parameterized by $k + \dtw(D)$ of~\cite{Johnson2001} immediately imply that {\sc Directed Disjoint Paths with Congestion} $2$ can be solved in $\mathcal{O}(n^{\poly(k)})$ time.
We remark that the algorithm by the authors of~\cite{Johnson2001} can be easily adapted to a constant congestion case by making copies of vertices of the input digraph and falling back to the disjoint case.
This fact is formally stated and proved in~\cite{AKM2019}.
From this discussion, we obtain the following.

\begin{theorem}\label{thm:XP-algo-DDP-k-strong}
  There is a polynomial function $\fref{thm:XP-algo-DDP-k-strong}(k)$ such that the {\sc Directed Disjoint Paths with Congestion} $2$ is solvable in $\mathcal{O}(n^{f(k)})$ time in $k$-strong digraphs.
\end{theorem}

\subsection{Half-integral packings of subgraphs with high directed treewidth}

\textcite{ChekuriC13} show how to partition an undirected graph $G$ with $\tw(G) = k$ into $h$ subgraphs of treewidth at least $r$ each if $k/\polylog k = \Omega(\min \{h  r^2, h^3  r\})$ in randomized $\poly(k, |V(G)|)$ time.
The authors discuss applications of this result on polylogarithmic approximation algorithms for routing with constant congestion, Erd\H{o}s-Pósa type results and bidimensionality for FPT algorithms.

A \emph{quarter-integral} (resp.\ \emph{half-integral}) \emph{packing} of subgraphs of a digraph $D$ is a collection of subgraphs of $D$ such that any vertex in $D$ is in at most $4$ (resp.\ $2$) subgraphs of the collection.
In~\cite{MasarikMuziPilipczukRzazewskiSorge21}, the authors obtain polynomial bounds for an Erd\H{o}s-Pósa type of problem in digraphs: quarter- and half-integral packings of directed cycles vs.\ feedback vertex sets.
The authors also prove a more general result concerning quarter-integral packings of subgraphs in a similar spirit to the result of \textcite{ChekuriC13}, namely that every digraph $D$ with $\dtw(D) = \Omega(a^6  b^8 \cdot \polylog (a  b))$ contains a quarter-integral packing of $a$ subgraphs each having directed treewidth at least $b$.
Our next result improves this to a half-integral packing, at the cost of some extra directed treewidth, using \cref{thm:main-proof}.

\begin{theorem}\label{thm:half-integral-packing}
  There exists an absolute constant $C > 0$ such that the following holds.
  Given $k, d \in \N$, let $\dHIP(d, k) = \dB(C  d  k) = \Ocal((d  k){}^{96} \cdot \log^{23} (d  k))$.
  Any digraph $D$ with $\dtw(D) \ge \dHIP(d, k)$ contains a half-integral packing of $k$ subgraphs each having directed treewidth at least $d$.
\end{theorem}

\begin{proof}
  \cref{prop:bramble-dtw} implies that there exists a constant $C' > 0$ such that $\bn(D) \le C' \cdot \dtw(D)$, so we define $C = 2C'$.
  As $\dtw(D) \ge \dB(C  d  k)$, $D$ contains a bramble $\Bcal$ of size $C  d  k$ and congestion $2$ by \cref{thm:main-proof}.
  Let $\Bcal = \{B_1, B_2, \dots, B_{Cdk}\}$ and define the brambles $\Bcal_i = \{ B_{(i-1)Cd+1},\allowbreak B_{(i-1)Cd+2},\allowbreak \dots,\allowbreak B_{iCd} \}$ for each $i \in [k]$.
  Note that $\Bcal_i$ has size $C  d$ and congestion $2$, which implies that the order of $\Bcal_i$ is at least $\lceil C \cdot d/2 \rceil$ for each $i \in [k]$.
  Letting $D_i$ be the subgraph of $D$ induced by the vertices in the bramble $\Bcal_i$, it follows that $\bn(D_i) \ge \lceil C \cdot d/2 \rceil \ge C' \cdot d$ and, hence, $\dtw(D_i) \ge d$ for each $i \in [k]$.
  Finally, as each vertex in $D$ is in at most $2$ bags of $\Bcal$, we have that $D_1, D_2, \dots, D_k$ is a half-integral packing of subgraphs of $D$.
\end{proof}

We remark that \cref{thm:half-integral-packing} actually proves something stronger: Each subgraph in the half-integral packing is strong and each pair of subgraphs must intersect.
Because of this, we believe the upper bound for $\dHIP(d, k)$ is far from tight and an improvement would be of interest.
Another way to improve \cref{thm:half-integral-packing} is to replace the half-integral packing to simply a packing, although this seems to be a much harder problem, as it would imply an upper bound for Younger's Conjecture (see e.g.~\cite{RRST1996, HKMM24}).

\section{Higher congestion}\label{sec:higher-congestion-with-masarik}

In this section we prove results analogous to those in \cref{sec:main-proof}.
In particular, we adapt and improve the main proof in~\cite{MasarikPilipczukRzazewskiSorge22} for brambles with congestion $4$ or $6$.
This adaptation consists in introducing walkages, adapting the results to work in this setting and using them to get more control over the final congestion.
We also simplify many steps and present a cleaner proof.

We first show how to find large brambles with congestion $4$ or $6$.
The next definition adapts path systems, which constitute the starting point of the proof in~\cite{MasarikPilipczukRzazewskiSorge22}, to the context of walkages.

\begin{definition}[Walk systems]\label{def:walk-system}
  Let $D$ be a digraph, $\ell, w \in \N$ and \(Z = \{(i,j) \in [w] \times [w] \mid i \ne j\} \).
  An \emph{$(\ell, w)$-walk system $\Scal$} in $D$ consists of a
  pair $(\Acal, \Wcal)$, where $\Wcal = \{W_1, W_2, \dots, W_w\}$ is a walkage in $D$ and $\Acal = \{A_i^{\inOp} \cup A_i^{\outOp} \mid i \in [w] \}$ is a
  family of disjoint subsets of \(V(D)\) such that for all $i \in [w]$, $A_i^{\inOp}, A_i^{\outOp} \subseteq V(W_i)$, $|A_i^{\inOp}| = |A_i^{\outOp}| = \ell$,
  each vertex
  in $A_i^{\inOp}$ appears
  in $W_i$ before all
  vertices in $A_i^{\outOp}$ do,
  and $\bigcup_{1 \leq i \leq w} (A_i^{\inOp} \cup A_i^{\outOp})$ is
  well-linked.
  The \emph{pair set of \(\Scal \)} is \( \pairs(\Scal) =  Z \).
\end{definition}

\cref{fig:walk-system} provides an example of a walk system.

\begin{figure}[htbp]
  \centering
  \begin{tikzpicture}
    \begin{scope}[line cap = round, transparency group, opacity = 0.65, decoration = {markings, mark = at position 0.55 with {\arrow{stealth[red]}}}]
      \draw[-latex, Red2, line width = 5pt] (2,0) to[bend right] (4,0);
    \end{scope}
    \begin{scope}[line cap = round, transparency group, opacity = 0.65, decoration = {markings, mark = at position 0.55 with {\arrow{stealth[red]}}}]
      \draw[-latex, Firebrick2, line width = 5pt] (6,0) to[bend right] (8,0);
    \end{scope}
    \begin{scope}[line cap = round, transparency group, opacity = 0.65, decoration = {markings, mark = at position 0.55 with {\arrow{stealth[red]}}}]
      \draw[-latex, IndianRed2, line width = 5pt] (2,0) to[bend left] (8,0);
    \end{scope}
    \begin{scope}[line cap = round, transparency group, opacity = 0.65, decoration = {markings, mark = at position 0.55 with {\arrow{stealth[red]}}}]
      \draw[-latex, OrangeRed2, line width = 5pt] (6,0) to[bend right] (0,0);
    \end{scope}
    \begin{scope}[line cap = round, transparency group, opacity = 0.65, decoration = {markings, mark = at position 0.55 with {\arrow{stealth[red]}}}]
      \draw[-latex, Brown2, line width = 5pt] (10,0) to[bend left] (0,0);
    \end{scope}
    \begin{scope}[line cap = round, transparency group, opacity = 0.65, decoration = {markings, mark = at position 0.55 with {\arrow{stealth[red]}}}]
      \draw[-latex, Tomato2, line width = 5pt] (10,0) to[bend right] (4,0);
    \end{scope}

    \draw[red, decoration = {brace, mirror}, decorate] (0,-0.1) -- (0.3,-0.1);
    \draw[red, decoration = {brace, mirror}, decorate] (1.7,-0.1) -- (2,-0.1);
    \draw[red, decoration = {brace, mirror}, decorate] (4,-0.1) -- (4.3,-0.1);
    \draw[red, decoration = {brace, mirror}, decorate] (5.7,-0.1) -- (6,-0.1);
    \draw[red, decoration = {brace, mirror}, decorate] (8,-0.1) -- (8.3,-0.1);
    \draw[red, decoration = {brace, mirror}, decorate] (9.7,-0.1) -- (10,-0.1);

    \node[below = 0.15em, red] at (0.23,0) {\footnotesize$A_1^{\inOp}$};
    \node[below = 0.15em, red] at (2,0) {\footnotesize$A_1^{\outOp}$};
    \node[below = 0.15em, red] at (4.23,0) {\footnotesize$A_2^{\inOp}$};
    \node[below = 0.15em, red] at (6,0) {\footnotesize$A_2^{\outOp}$};
    \node[below = 0.15em, red] at (8.23,0) {\footnotesize$A_3^{\inOp}$};
    \node[below = 0.15em, red] at (10,0) {\footnotesize$A_3^{\outOp}$};

    \draw[walkage] (0,0) -- (2,0);
    \draw[walkage] (4,0) -- (6,0);
    \draw[walkage] (8,0) -- (10,0);

    \node[below] at (1,0) {$W_1$};
    \node[below] at (5,0) {$W_2$};
    \node[below] at (9,0) {$W_3$};
  \end{tikzpicture}

  \caption{A walk system, the walkage $\Wcal$ is in black, the thick colored arrows represent walkages of size $\ell$ between walks in $\Wcal$, and the sets $A_i^{\inOp}$ and $A_i^{\outOp}$ of size $\ell$ are in red.}\label{fig:walk-system}
\end{figure}
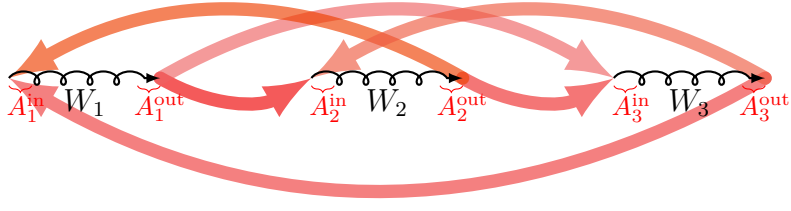

We now show that~\cite[Lemma~3.2]{MasarikPilipczukRzazewskiSorge22}, which composes the sparse case, also works with walk systems.
We provide a proof that follows loosely the original one for completeness.
The sparse case in walk systems occurs when there are many disjoint walks that connect a large portion of $\Wcal$, which enables us to define many closed walks and pick some that are guaranteed to intersect.

\begin{lemma}[{Adapted from~\cite[Lemma~3.2]{MasarikPilipczukRzazewskiSorge22}}]\label{lem:sparse-bramble}
  Given $k \in \N$, let $\wSB(k) = 1 + 5\dCM(k(k-1)/2) = \Ocal(k^2 \sqrt{\log k})$.
  Let $D$ be a digraph containing an $(\ell, w)$-walk system $\Scal = (\Acal, \Wcal)$ and a set $J \subseteq \pairs(\Scal)$ of size at least $0.6(w^2 - w)$ such that there exists a walkage $\Tcal$ that contains at least one walk from an $(A_i^{\outOp}, A_j^{\inOp})$-walkage for all $(i, j) \in J$.
  If $w \ge \wSB(k)$ and $\ell \ge 1$, then $D$ contains a bramble of size $k$ and congestion at most $4$.
\end{lemma}

\begin{proof}
  Let $G$ be a graph with a vertex for each walk in $\Wcal$ and edges between walks $W_i$ and $W_j$ if there are two walks $L_{ij}, L_{ji} \in \Tcal$ such that $\stLink(L_{ij}) \in A_i^{\outOp}$, $\enLink(L_{ij}) \in A_j^{\inOp}$, $\stLink(L_{ji}) \in A_j^{\outOp}$, and $\enLink(L_{ji}) \in A_i^{\inOp}$ (i.e.,\ $L_{ij}$ goes from $W_i$ to $W_j$ and $L_{ji}$ goes from $W_j$ to $W_i$).
  Note that $|J| \ge 0.6(w^2 - w)$ implies that $G$ has at least $0.1w(w-1)$ edges.
  As the average degree of $G$ is at least $2 \cdot 0.1w(w-1)/w = 0.2(w-1) \ge \dCM(k(k - 1)/2)$, $G$ contains a clique minor of order $k(k-1)/2$ by \cref{prop:clique-minor}.
  Label the vertices of this minor as the elements of the set $\binom{[k]}{2}$.
  For each $i \in [k]$, let $\Bcal_i$ be the union of the branch sets of the vertices with label $\{i, j\}$ for some $j \in [k] \setminus \{i\}$, and let $\Wcal_i$ be the set of walks from $\Wcal$ such that their associated vertices in $G$ are in $\Bcal_i$.

  We build a closed walk $C_{ij}$ for each (unordered) pair of walks $W_i, W_j \in \Wcal$ which induce an edge in $G$ as follows.
  Let $L_{ij}, L_{ji} \in \Tcal$ be two walks going from $W_i$ to $W_j$ and vice versa, which always exist if $W_i$ and $W_j$ induce an edge in $G$.
  Start at the first vertex of $L_{ij}$ and follow it until its final vertex, which is in $W_j$, then follow $W_j$ until the first vertex of $L_{ji}$, then follow $L_{ji}$ until its final vertex, which is in $W_i$, then follow $W_i$ until the first vertex of $L_{ij}$ again, as in \cref{fig:closed-walk-cij}.
  Note that $C_{ij}$ contains the last vertex of $A_i^{\inOp}$ and the last vertex of $A_j^{\inOp}$, as well as the first vertex of $A_i^{\outOp}$ and the first vertex of $A_j^{\outOp}$, so $C_{ij}$ always intersects $C_{i\ell}$ and $C_{mj}$ for any $\ell \ne j$ and $m \ne i$.

  \begin{figure}[htbp]
    \centering
    \begin{tikzpicture}
      \begin{scope}[line cap = round]
        \draw[decoration = {coil, aspect = 1.2, post length = 2mm}, decorate, red!30, line width = 4pt] (2.7,0) -- (0.15,2.08);
        \draw[decoration = {coil, aspect = 1.2, post length = 2mm}, decorate, red!30, line width = 4pt] (2.7,2) -- (0.172,0.086);
      \end{scope}
      \begin{scope}[line cap = round]
        \clip (0.15,1.7) rectangle (2.7,2.2);
        \draw[decoration = {coil, aspect = 1.2, post length = 2mm}, decorate, red!30, line width = 4pt] (0,2) -- (3,2);
      \end{scope}
      \begin{scope}[line cap = round]
        \clip (0.15,-0.7) rectangle (2.7,0.2);
        \draw[decoration = {coil, aspect = 1.2, post length = 2mm}, decorate, red!30, line width = 4pt] (0,0) -- (3,0);
      \end{scope}

      \draw[walkage, decoration = {post length = 2mm}, red] (2.7,0) to node[above right, pos = 0.15] {$L_{ji}$} (0.15,2.08);
      \draw[walkage, decoration = {post length = 2mm}, red] (2.7,2) to node[below right, pos = 0.15] {$L_{ij}$} (0.172,0.086);

      \draw[walkage, decoration = {post length = 2mm}] (0,2) -- (3,2);
      \draw[walkage, decoration = {post length = 2mm}] (0,0) -- (3,0);

      \node[right] at (3,2) {$W_i$};
      \node[right] at (3,0) {$W_j$};
    \end{tikzpicture}

    \caption{Sketch of the construction of the closed walk $C_{ij}$ (in light red).}\label{fig:closed-walk-cij}
  \end{figure}
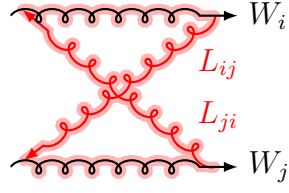

  We now define, for each $i \in [k]$, the closed walk $C_i$ as the concatenation of $C_{ij}$ for all $W_j \in \Wcal_i$.
  Note that for distinct $i, j \in [k]$, $C_i$ and $C_j$ intersect precisely at $C_{ij}$.
  Moreover, each vertex $v \in V(D)$ is in at most one walk $W \in \Wcal$ and at most one walk $L \in \Tcal$, so $v$ is in at most two closed walks $C_{ij}$ and, thus, in at most four closed walks $C_i$.
  Hence, the family $(C_i){}_{i \in [k]}$ is a bramble of size $k$ and congestion at most $4$.
\end{proof}

We find \cref{lem:sparse-bramble} to be the biggest bottleneck in~\cite{MasarikPilipczukRzazewskiSorge22}, in the sense that lowering the congestion from $8$ to $4$ is only a matter of adapting a few results to maintain better control over the closed walks and taking extra care when applying the LLL\@.
Lowering the congestion to $2$, however, is much harder.
In fact, we were not able to find any way of achieving this without passing through strong systems, the strategy of \cref{thm:main-proof} being the cleanest version of our proof.

An interesting remark is that \cref{lem:sparse-bramble} can be improved to find brambles with congestion $2$ if the input walk system is \emph{clean}.
A walk system $\Scal = (\Acal, \Wcal)$ is said to be clean if for all $(i,j) \in \pairs(\Scal)$ there exists an $(A_i^{\outOp}, A_j^{\inOp})$-walkage that does not intersect any walk in $\Wcal \setminus \{W_i, W_j\}$.
The congestion in \cref{lem:sparse-bramble} is strengthened from at most $4$ to at most $2$ in this case simply because each vertex in $D$ would then be in at most one walk $W \in \Wcal$ and at most one walk $L \in \Tcal$, but $W$ must be in the same bag as all walks $L \in \Tcal$ that intersect it.
This means that each vertex in $D$ is in at most one closed walk $C_{ij}$ and, thus, at most two closed walks $C_i$.

Clean walk systems were first defined in~\cite[Definition~4.5]{KK22} (technically, they define clean \emph{path} systems, but any path system is also a walk system), but the authors provide no way of forcing their existence in digraphs of large directed treewidth, only as a possible outcome among others.
With some work we can actually show that any digraph with $\poly(\ell  w)$ directed treewidth contains either a clean $(\ell, w)$-walk system or an $(\ell, w)$-strong system, but we omit the technical details as the process of finding a bramble with congestion $2$ using this result demands a much larger directed treewidth than the method of \cref{thm:main-proof}.

The next two results find brambles with congestion $4$ or $6$ in large enough walk systems.
They essentially find brambles in a dense configuration using \cref{cor:cyclic-extension,cor:structured-walkage} and in a sparse one using \cref{lem:sparse-bramble}.
The arguments presented here both follow loosely the main proof in~\cite{MasarikPilipczukRzazewskiSorge22}, but they simplify many steps and strengthen the final result.

\begin{lemma}\label{lem:final-bramble-clean-4}
  Given $k, w \in \N$, let
  \[
    \begin{array}{@{}r@{}l@{}l@{}}
      s'           & {} = \lISM(\dCM(k), 0.6(w^2 - w))               & {} = \Ocal(w^2 k \sqrt{\log k}), \\
      s            & {} = \lISM(\dCEC(s'), \allowbreak 0.6(w^2 - w)) & {} = \Ocal(w^4 k \sqrt{\log k}), \\
      u            & {} = \lISM(\dCE(k, s), 0.6(w^2 - w))            & {} = \Ocal(w^6 k^2 \log k), \\
      \lFBCC(k, w) & {} = \wSW(k, s, u)                              & {} = \Ocal(w^{14} k^6 \cdot \log^3 k).
    \end{array}
  \]
  Let $D$ be a digraph containing an $(\ell,w)$-walk system $\Scal = (\Acal, \Wcal)$.
  If $\ell \ge \lFBCC(k, w)$ and $w \ge \wSB(k)$, then $D$ contains a bramble of size $k$ and congestion at most $4$.
\end{lemma}

\begin{proof}
  Let $Z = \pairs(\Scal)$ and let $\Lcal'_{ij}$ be an $(A_i^{\outOp}, A_j^{\inOp})$-walkage for all $(i, j) \in Z$, which exists by well-linkedness.
  As $\ell \ge \wSW(k, s, u)$, we apply \cref{cor:structured-walkage} on $\Lcal'_{ij}$ for each $(i, j) \in Z$ to get a bramble of size $k$ and congestion $2$, and we are done, or a subwalkage $\Lcal_{ij} \subseteq \Lcal'_{ij}$ that is either cyclic and has size $s$ or untangled and has size $u$.
  We may assume that every $(i, j) \in Z$ which is untangled satisfies $\degen(\intersec(\Lcal_{ij} \cup \Lcal_{i'j'})) < \dCE(k, s)$ for all $(i', j') \in Z$.
  Indeed, if this is not the case, we may apply \cref{cor:cyclic-extension} to get a bramble of size $k$ and congestion $2$, and we are done, or a walkage $\Lcal''_{ij} \subseteq \Lcal_{ij}$ of size $s$ that is cyclic which we may use as $\Lcal_{ij}$ instead.

  We say that a pair $(i, j) \in Z$ is \emph{cyclic} or \emph{untangled} if $\Lcal_{ij}$ is cyclic or untangled, respectively.
  Let $Z_1 \subseteq Z$ be the set of cyclic pairs.
  Let $Z_2 \subseteq Z$ be the set of untangled pairs $(i,j)$ such that there exists $(i', j') \in Z$ for which $\degen(\intersec(\Lcal_{ij} \cup \Lcal_{i'j'})) \ge \dCEC(s')$.
  For each $(i, j) \in Z_2$, let $\Lcal^2_{ij} \subseteq \Lcal_{ij}$ be a walkage of size $s'$ that is $2$-cyclic, which exists by \cref{lem:untangled-cyclic-extension-congestion,lem:cyclic-extension-congestion}.

  We now show that one of the sets $Z \setminus Z_1$, $Z_1 \cup Z_2$ or $Z \setminus Z_2$ has size at least $0.6|Z|$.
  To do so, note that if none of these three cases occurs, then $2|Z \setminus Z_1| + 2|Z_1 \cup Z_2| + |Z \setminus Z_2| < 3|Z|$, so it suffices to show that $2|Z \setminus Z_1| + 2|Z_1 \cup Z_2| + |Z \setminus Z_2| \ge 3|Z|$.
  For the following analysis, note that $Z_1$ and $Z_2$ are disjoint by definition.
  We consider how each of the terms $2|Z \setminus Z_1|$, $2|Z_1 \cup Z_2|$ and $|Z \setminus Z_2|$ count the pairs in $Z$.
  Each pair in $Z_1$ is counted three times: twice in $2|Z_1 \cup Z_2|$ and once in $|Z \setminus Z_2|$.
  Each pair in $Z_2$ is counted four times: twice in $2|Z \setminus Z_1|$ and twice in $2|Z_1 \cup Z_2|$.
  Finally, each pair in $Z \setminus (Z_1 \cup Z_2)$ is counted three times: twice in $2|Z \setminus Z_1|$ and once in $|Z \setminus Z_2|$.
  Since each pair in $Z$ is counted at least $3$ times, we have $2|Z \setminus Z_1| + 2|Z_1 \cup Z_2| + |Z \setminus Z_2| \ge 3|Z|$.

  We find the bramble in each of the three cases.
  First, consider $|Z \setminus Z_1| \ge 0.6|Z|$.
  Since $Z_1$ consists of the set of cyclic pairs, every pair $(i, j) \in Z \setminus Z_1$ is untangled.
  Furthermore, $\degen(\intersec(\Lcal_{ij} \cup \Lcal_{i'j'})) < \dCE(k, s)$, for every $(i, j), (i', j') \in Z \setminus Z_1$, by construction.
  Therefore, as $u \ge \lISM(\dCE(k, s), 0.6(w^2 - w))$, we can use \cref{prop:independent-set-multipartite} to get a walkage $\Tcal$ that contains one walk from $\Lcal_{ij}$ for each $(i, j) \in Z \setminus Z_1$.
  As $w \ge \wSB(k)$, we can apply \cref{lem:sparse-bramble} on $\Scal$ and $\Tcal$ to get a bramble of size $k$ and congestion at most $4$.

  We now consider $|Z_1 \cup Z_2| \ge 0.6|Z|$.
  By considering $\Lcal^2_{ij} = \Lcal_{ij}$ when $(i,j) \in Z_1$, it follows that $\Lcal^2_{ij}$ has size at least $s'$ and is either cyclic or $2$-cyclic.
  Hence, $\Lcal^2_{ij}$ originates a closed walkage $\Zcal_{ij}$ of congestion at most $2$.
  If $\degen(\intersec(\Zcal_{ij} \cup \Zcal_{i'j'})) \ge \dCM(k)$ for some $(i, j), (i', j') \in Z_1 \cup Z_2$, we get a bramble of size $k$ and congestion at most $4$ by \cref{prop:dense-bramble}.
  Otherwise, as $s' \ge \lISM(\dCM(k), 0.6(w^2 - w)) \ge \dCM(k)$, there exists a walkage $\Tcal$ that contains one walk from $\Lcal^2_{ij}$ for each $(i, j) \in Z_1 \cup Z_2$ by \cref{prop:independent-set-multipartite}.
  As $w \ge \wSB(k)$, we can apply again \cref{lem:sparse-bramble} on $\Scal$ and $\Tcal$ to get a bramble of size $k$ and congestion at most $4$.

  Finally, we consider $|Z \setminus Z_2| \ge 0.6|Z|$.
  If $\degen(\intersec(\Lcal_{ij} \cup \Lcal_{i'j'})) \ge \dCM(k)$ for any $(i, j), (i', j') \in Z_1$, we get a bramble of size $k$ and congestion $2$, by \cref{prop:dense-bramble}.
  Thus, we may assume $\degen(\intersec(\Lcal_{ij} \cup \Lcal_{i'j'})) < \dCM(k) < \dCEC(s')$ for every $(i, j), (i', j') \in Z \setminus Z_2$.
  Indeed, this is true when either pair is untangled, as they are not in $Z_2$.
  As $s, u \ge \lISM(\dCEC(s'), 0.6(w^2 - w))$, there exists a walkage $\Tcal$ that contains one walk from $\Lcal_{ij}$ for each $(i, j) \in Z \setminus Z_2$ by \cref{prop:independent-set-multipartite}.
  As $w \ge \wSB(k)$, we can apply \cref{lem:sparse-bramble} on $\Scal$ and $\Tcal$ to get a bramble of size $k$ and congestion at most $4$.
\end{proof}

The proof of \cref{lem:final-bramble-clean-6} is very similar to that of \cref{lem:final-bramble-clean-4}, so we skip some technical details already present in the proof above for simplicity but repeat the core ideas for completeness.

\begin{lemma}\label{lem:final-bramble-clean-6}
  Given $k, w \in \N$, let
  \[
    \begin{array}{@{}r@{}l@{}l@{}}
      s'          & {} = \lISM(\dCM(k), 0.6(w^2 - w))               & {} = \Ocal(w^2 k \sqrt{\log k}), \\
      s           & {} = \lISM(\dCEC(s'), \allowbreak 0.6(w^2 - w)) & {} = \Ocal(w^4 k \sqrt{\log k}), \\
      u           & {} = \lISM(\dCEC(s), 0.6(w^2 - w))              & {} = \Ocal(w^6 k \sqrt{\log k}), \\
      \lFBC(k, w) & {} = \wSWC(s, u)                                & {} = \Ocal(w^{14} k^3 \log^{3/2} k).
    \end{array}
  \]
  Let $D$ be a digraph containing an $(\ell,w)$-walk system $\Scal = (\Acal, \Wcal)$.
  If $\ell \ge \lFBC(k, w)$ and $w \ge \wSB(k)$, then $D$ contains a bramble of size $k$ and congestion at most $6$.
\end{lemma}

\begin{proof}
  Let $Z = \pairs(\Scal)$ and let $\Lcal'_{ij}$ be an $(A_i^{\outOp}, A_j^{\inOp})$-walkage for all $(i, j) \in Z$, which exists by well-linkedness.
  As $\ell \ge \wSWC(s, u)$, we apply \cref{lem:structured-walkage-congestion} on $\Lcal'_{ij}$ for each $(i, j) \in Z$ to get a subwalkage $\Lcal_{ij} \subseteq \Lcal'_{ij}$ that is either $2$-cyclic and has size $s$ or untangled and has size $u$.
  We may assume that every $(i, j) \in Z$ which is untangled satisfies $\degen(\intersec(\Lcal_{ij} \cup \Lcal_{i'j'})) < \dCEC(s)$ for all $(i', j') \in Z$.
  Indeed, if this is not the case, we may apply \cref{lem:cyclic-extension-congestion} to get a walkage $\Lcal''_{ij} \subseteq \Lcal_{ij}$ of size $s$ that is $2$-cyclic which we may use as $\Lcal_{ij}$ instead.

  We say that a pair $(i, j) \in Z$ is \emph{cyclic} or \emph{untangled} if $\Lcal_{ij}$ is $2$-cyclic or untangled, respectively.
  Let $Z_1 \subseteq Z$ be the set of cyclic pairs.
  Let $Z_2 \subseteq Z$ be the set of untangled pairs $(i,j)$ such that there exists $(i', j') \in Z$ for which $\degen(\intersec(\Lcal_{ij} \cup \Lcal_{i'j'})) \ge \dCEC(s')$.
  For each $(i, j) \in Z_2$, let $\Lcal^2_{ij} \subseteq \Lcal_{ij}$ be a walkage of size $s'$ that is either $2$- or $3$-cyclic, which exists by \cref{lem:untangled-cyclic-extension-congestion,lem:cyclic-extension-congestion}.

  Again, one of the sets $Z \setminus Z_1$, $Z_1 \cup Z_2$ or $Z \setminus Z_2$ has size at least $0.6|Z|$.
  We find the bramble in each of the three cases.
  First, consider $|Z \setminus Z_1| \ge 0.6|Z|$.
  Since $Z_1$ consists of the set of cyclic pairs, every pair $(i, j) \in Z \setminus Z_1$ is untangled.
  Furthermore, $\degen(\intersec(\Lcal_{ij} \cup \Lcal_{i'j'})) < \dCEC(s)$, for every $(i, j), (i', j') \in Z \setminus Z_1$, by construction.
  Therefore, as $u \ge \lISM(\dCEC(s), 0.6(w^2 - w))$, we can use \cref{prop:independent-set-multipartite} to get a walkage $\Tcal$ that contains one walk from $\Lcal_{ij}$ for each $(i, j) \in Z \setminus Z_1$.
  As $w \ge \wSB(k)$, we can apply \cref{lem:sparse-bramble} on $\Scal$ and $\Tcal$ to get a bramble of size $k$ and congestion at most~$4$.

  We now consider $|Z_1 \cup Z_2| \ge 0.6|Z|$.
  By considering $\Lcal^2_{ij} = \Lcal_{ij}$ when $(i,j) \in Z_1$, it follows that $\Lcal^2_{ij}$ has size at least $s'$ and is either $2$- or $3$-cyclic.
  Hence, $\Lcal^2_{ij}$ originates a closed walkage $\Zcal_{ij}$ of congestion at most $3$.
  If $\degen(\intersec(\Zcal_{ij} \cup \Zcal_{i'j'})) \ge \dCM(k)$ for some $(i, j), (i', j') \in Z_1 \cup Z_2$, we get a bramble of size $k$ and congestion at most $6$ by \cref{prop:dense-bramble}.
  Otherwise, as $s' \ge \lISM(\dCM(k), 0.6(w^2 - w)) \ge \dCM(k)$, there exists a walkage $\Tcal$ that contains one walk from $\Lcal^2_{ij}$ for each $(i, j) \in Z_1 \cup Z_2$ by \cref{prop:independent-set-multipartite}.
  As $w \ge \wSB(k)$, we can apply again \cref{lem:sparse-bramble} on $\Scal$ and $\Tcal$ to get a bramble of size $k$ and congestion at most $4$.

  Finally, we consider $|Z \setminus Z_2| \ge 0.6|Z|$.
  If $\degen(\intersec(\Lcal_{ij} \cup \Lcal_{i'j'})) \ge \dCM(k)$ for any $(i, j), (i', j') \in Z_1$, we get a bramble of size $k$ and congestion $2$, by \cref{prop:dense-bramble}.
  Thus, we may assume $\degen(\intersec(\Lcal_{ij} \cup \Lcal_{i'j'})) < \dCM(k) < \dCEC(s')$ for every $(i, j), (i', j') \in Z\setminus Z_2$.
  Indeed, this is true when either pair is untangled, as they are not in $Z_2$.
  As $s, u \ge \lISM(\dCEC(s'), 0.6(w^2 - w))$, there exists a walkage $\Tcal$ that contains one walk from $\Lcal_{ij}$ for each $(i, j) \in Z \setminus Z_2$ by \cref{prop:independent-set-multipartite}.
  As $w \ge \wSB(k)$, we can apply \cref{lem:sparse-bramble} on $\Scal$ and $\Tcal$ to get a bramble of size $k$ and congestion at most $4$.
\end{proof}

The following two theorems are analogous to \cref{thm:main-proof} but for higher congestion.
Their proofs are again very similar.

\begin{theorem}\label{thm:main-proof-congestion-4-masarik}
  Given $k \in \N$, let $\dBCCM(k) = \wLP(2\wSB(k)\lFBCC(k, \wSB(k))) + 1 = \Ocal(k^{72} \cdot \log^{21} k)$.
  Any digraph $D$ with $\dtw(D) \ge \dBCCM(k)$ contains a bramble of size $k$ and congestion at most $4$.
\end{theorem}

\begin{proof}
  Let $w = \wSB(k)\cdot 2\lFBCC(k, \wSB(k))$.
  As $\dtw(D) \ge \wLP(w) + 1$, we get a path $P$ and a well-linked set $A \subseteq V(P)$ of size $w$ by \cref{prop:linked-path}.
  Order the vertices of $A$ according to their appearance in $P$.
  For each $i \in [\wSB(k)]$, let $W_i$ be the walk that starts at the $((i-1) \cdot 2\lFBCC(k, \wSB(k)) + 1)$-th vertex of $A$ in $P$ and ends at the $(i \cdot 2\lFBCC(k, \wSB(k)))$-th vertex of $A$ in $P$.
  As $w \ge \wSB(k) \cdot 2\lFBCC(k, \wSB(k))$, there are always $\wSB(k)$ such walks.
  Define $\Wcal = \{W_1, W_2, \dots, \allowbreak W_{\wSB(k)}\}$ and note that $\Wcal$ is a $2\lFBCC(k, \wSB(k))$-linked walkage of size $\wSB(k)$.

  By well-linkedness, there is an $\lFBCC(k, \wSB(k))$-walkage from the last $\lFBCC(k, \wSB(k))$ vertices of $A$ in $W_i$ to the first $\lFBCC(k, \wSB(k))$ vertices of $A$ in $W_j$ for any distinct $i, j \in [\wSB(k)]$.
  Let $\Acal = \{A_i^{\inOp} \cup A_i^{\outOp} \mid i \in [\wSB(k)]\}$ where $A_i^{\inOp}$ and $A_i^{\outOp}$ are the sets of the first and last $\lFBCC(k, \wSB(k))$ vertices of $A$ in $W_i$.
  Define an $(\lFBCC(k, \wSB(k)), \wSB(k))$-walk system $\Scal = (\Acal, \Wcal)$.
  Finally, apply \cref{lem:final-bramble-clean-4} to $\Scal$ and obtain a bramble of size $k$ and congestion at most $4$.
\end{proof}

\begin{theorem}\label{thm:main-proof-congestion-6-masarik}
  Given $k \in \N$, let $\dBCM(k) = \wLP(2\wSB(k)\lFBC(k, \wSB(k))) + 1 = \Ocal(k^{66} \cdot \log^{18} k)$.
  Any digraph $D$ with $\dtw(D) \ge \dBCM(k)$ contains a bramble of size $k$ and congestion at most $6$.
\end{theorem}

\begin{proof}
  Let $w = \wSB(k)\cdot 2\lFBC(k, \wSB(k))$.
  As $\dtw(D) \ge \wLP(w) + 1$, we get a path $P$ and a well-linked set $A \subseteq V(P)$ of size $w$ by \cref{prop:linked-path}.
  Order the vertices of $A$ according to their appearance in $P$.
  For each $i \in [\wSB(k)]$, let $W_i$ be the walk that starts at the $((i-1) \cdot 2\lFBC(k, \wSB(k)) + 1)$-th vertex of $A$ in $P$ and ends at the $(i \cdot 2\lFBC(k, \wSB(k)))$-th vertex of $A$ in $P$.
  As $w \ge \wSB(k) \cdot 2\lFBC(k, \wSB(k))$, there are always $\wSB(k)$ such walks.
  Define $\Wcal = \{W_1, W_2, \dots, \allowbreak W_{\wSB(k)}\}$ and note that $\Wcal$ is a $2\lFBC(k, \wSB(k))$-linked walkage that has size $\wSB(k)$.

  By well-linkedness, there is an $\lFBC(k, \wSB(k))$-walkage from the last $\lFBC(k, \wSB(k))$ vertices of $A$ in $W_i$ to the first $\lFBC(k, \wSB(k))$ vertices of $A$ in $W_j$ for any distinct $i, j \in [\wSB(k)]$.
  Let $\Acal = \{A_i^{\inOp} \cup A_i^{\outOp} \mid i \in [\wSB(k)]\}$ where $A_i^{\inOp}$ and $A_i^{\outOp}$ are the sets of the first and last $\lFBC(k, \wSB(k))$ vertices of $A$ in $W_i$.
  Define an $(\lFBC(k, \wSB(k)), \wSB(k))$-walk system $\Scal = (\Acal, \Wcal)$.
  Finally, we apply \cref{lem:final-bramble-clean-6} to $\Scal$ and obtain a bramble of size $k$ and congestion at most $6$.
\end{proof}

We remark that the discussion after \cref{thm:main-proof} concerning the construction of the bramble with congestion $2$ also applies for \cref{thm:main-proof-congestion-4-masarik,thm:main-proof-congestion-6-masarik}.
Moreover, \cref{thm:main-proof-congestion-4-masarik,thm:main-proof-congestion-6-masarik} can be used to prove analogous versions of \cref{thm:k-2-linked-digraphs,thm:XP-algo-DDP-k-strong} that give better bounds on $\fref{thm:k-2-linked-digraphs}(k)$ and $\fref{thm:XP-algo-DDP-k-strong}(k)$ but for \textsc{Directed Disjoint Paths with Congestion} $c$ for $c = 4$ and $c = 6$, respectively.

\clearpage

\printbibliography{}

\end{document}